\documentclass[a4paper,fleqn]{cas-sc}

\usepackage{core}
\usepackage{hcalculation}
\usepackage{tikzit}
\usepackage{gla}
\usepackage{makecell}

\tikzstyle{Black}=[fill=black, draw=black, shape=circle, inner sep=0, minimum size=6pt]
\tikzstyle{White}=[fill=white, draw=black, shape=circle, inner sep=0, minimum size=6pt]
\tikzstyle{Red}=[fill=red, draw=black, shape=rectangle]
\tikzstyle{Rel}=[fill=white, draw=black, shape=chamfered rectangle, minimum size=8, chamfered rectangle corners=north east, chamfered rectangle xsep=2, chamfered rectangle ysep=2, tikzit shape=rectangle]
\tikzstyle{Subspace}=[fill=white, draw=black, shape=rectangle, minimum size=8, tikzit shape=rectangle]
\tikzstyle{SubspaceGray}=[fill=gray, draw=black, shape=rectangle, minimum size=8, tikzit shape=rectangle]
\tikzstyle{CoRel}=[fill=white, draw=black, shape=chamfered rectangle, minimum size=8, chamfered rectangle corners=north west, chamfered rectangle xsep=2, chamfered rectangle ysep=2, tikzit shape=rectangle, tikzit draw=red]
\tikzstyle{RelBig}=[fill=white, draw=black, shape=chamfered rectangle, minimum size=24, chamfered rectangle corners=north east, chamfered rectangle xsep=4, chamfered rectangle ysep=4, tikzit shape=rectangle]
\tikzstyle{CoRelBig}=[fill=white, draw=black, shape=chamfered rectangle, minimum size=24, chamfered rectangle corners=north west, chamfered rectangle xsep=4, chamfered rectangle ysep=4, tikzit shape=rectangle, tikzit draw=red]
\tikzstyle{RelSmal}=[fill=white, draw=black, shape=chamfered rectangle, minimum size=6, chamfered rectangle corners=north east, chamfered rectangle xsep=2, chamfered rectangle ysep=2, tikzit shape=rectangle]
\tikzstyle{CoRelSma?}=[fill=white, draw=black, shape=chamfered rectangle, minimum size=6, chamfered rectangle corners=north west, chamfered rectangle xsep=2, chamfered rectangle ysep=2, tikzit shape=rectangle, tikzit draw=red]
\tikzstyle{Map}=[fill=white, draw=black, shape=rounded rectangle, minimum width=23, minimum height=16, rounded rectangle west arc=none, tikzit shape=rectangle]
\tikzstyle{MapSmall}=[fill=white, draw=black, shape=rounded rectangle, minimum width=12, minimum height=8, rounded rectangle west arc=none, tikzit shape=rectangle]
\tikzstyle{CoMap}=[fill=white, draw=black, shape=rounded rectangle, minimum width=23, minimum height=16, rounded rectangle east arc=none, tikzit shape=rectangle, tikzit draw=red]
\tikzstyle{CoMapSmall}=[fill=white, draw=black, shape=rounded rectangle, minimum width=12, minimum height=8, rounded rectangle east arc=none, tikzit shape=rectangle, tikzit draw=red]
\tikzstyle{MapGray}=[fill=gray, draw=black, shape=rounded rectangle, minimum width=23, minimum height=16, rounded rectangle west arc=none, tikzit shape=rectangle]
\tikzstyle{CoMapGray}=[fill=gray, draw=black, shape=rounded rectangle, minimum width=23, minimum height=16, rounded rectangle east arc=none, tikzit shape=rectangle, tikzit draw=red]
\tikzstyle{MapBig}=[fill=white, draw=black, shape=rounded rectangle, minimum width=24, minimum height=32, rounded rectangle west arc=none, tikzit shape=rectangle]
\tikzstyle{CoMapBig}=[fill=white, draw=black, shape=rounded rectangle, minimum width=24, minimum height=32, rounded rectangle east arc=none, tikzit shape=rectangle, tikzit draw=red]
\tikzstyle{RelGray}=[fill=gray, draw=black, shape=chamfered rectangle, minimum size=8, chamfered rectangle corners=north east, chamfered rectangle xsep=2, chamfered rectangle ysep=2, tikzit shape=rectangle, tikzit fill={rgb,255: red,191; green,191; blue,191}]
\tikzstyle{CoRelGray}=[fill=gray, draw=black, shape=chamfered rectangle, tikzit draw=red, tikzit fill={rgb,255: red,191; green,191; blue,191}, minimum size=8, chamfered rectangle corners=north west, chamfered rectangle xsep=2, chamfered rectangle ysep=2, tikzit shape=rectangle]
\tikzstyle{Empty}=[draw=black, shape=rectangle, dashed, minimum size=16, tikzit draw=black]
\tikzstyle{Inv}=[fill={blue!20}, draw=blue, shape=rounded rectangle, minimum width=23, minimum height=16, rounded rectangle west arc=none, tikzit shape=rectangle]
\tikzstyle{CoInv}=[fill={blue!20}, draw=blue, shape=rounded rectangle, minimum width=23, minimum height=16, rounded rectangle east arc=none, tikzit shape=rectangle, tikzit draw=red]
\tikzstyle{Ortho}=[fill={red!20}, draw=red, shape=rounded rectangle, minimum width=23, minimum height=16, rounded rectangle west arc=none, tikzit shape=rectangle]
\tikzstyle{CoOrtho}=[fill={red!20}, draw=red, shape=rounded rectangle, minimum width=23, minimum height=16, rounded rectangle east arc=none, tikzit shape=rectangle, tikzit draw=red]
\tikzstyle{InvBig}=[fill={blue!20}, draw=blue, shape=rounded rectangle, minimum width=24, minimum height=32, rounded rectangle west arc=none, tikzit shape=rectangle]
\tikzstyle{CoInvBig}=[fill={blue!20}, draw=blue, shape=rounded rectangle, minimum width=24, minimum height=32, rounded rectangle east arc=none, tikzit shape=rectangle, tikzit draw=red]
\tikzstyle{RelGraySmall}=[fill=gray, draw=black, shape=chamfered rectangle, minimum size=6, chamfered rectangle corners=north east, chamfered rectangle xsep=1, chamfered rectangle ysep=1, tikzit shape=rectangle, tikzit fill={rgb,255: red,191; green,191; blue,191}]
\tikzstyle{CoRelGraySmall}=[fill=gray, draw=black, shape=chamfered rectangle, tikzit draw=red, tikzit fill={rgb,255: red,191; green,191; blue,191}, minimum size=6, chamfered rectangle corners=north west, chamfered rectangle xsep=1, chamfered rectangle ysep=1, tikzit shape=rectangle]

\tikzstyle{new edge style 3}=[-, dashed]
\tikzstyle{arrow}=[->]
\tikzstyle{arrow2}=[->, dashed]
\tikzstyle{arrowlarge}=[->, thick=3pt]
\tikzstyle{arrowbluelarge}=[-, draw=blue, thick=3pt, line width=2pt]
\tikzstyle{arrowredlarge}=[-, draw=red, thick=2pt, line width=2pt]
\tikzstyle{arrow3}=[<->]
\tikzstyle{arrow3large}=[<->, thick=3pt]
\tikzstyle{arrowBlue}=[->, draw=blue, thick=2.2pt]
\tikzstyle{arrowBlueduplo}=[<->, draw=blue, thick=2.2pt]
\tikzstyle{arrowRedduplo}=[<->, draw=red, thick=2.2pt]
\tikzstyle{arrowRed}=[->, draw=red, thick=2.2pt]
\tikzstyle{arrowgreen}=[->, draw=green, thick=2.2pt]
\tikzstyle{arroworange}=[->, draw=orange, thick=2.2pt]

\usepackage{subcaption}

\usepackage[numbers]{natbib}

\DeclareMathOperator{\argmin}{argmin}
\DeclareMathOperator{\rank}{rank}

\DeclareMathOperator{\Img}{Im}
\DeclareMathOperator{\Ker}{Ker}
\DeclareMathOperator{\CoKer}{Nil}
\DeclareMathOperator{\Dom}{Dom}

\DeclareMathOperator{\Tot}{Tot}
\DeclareMathOperator{\Det}{Det}
\DeclareMathOperator{\Inj}{Inj}
\DeclareMathOperator{\Sur}{Sur}

\newtheorem{theorem}{Theorem}
\newtheorem{lemma}[theorem]{Lemma}
\newtheorem{example}{Example}

\newtheorem{remark}{Remark}
\newtheorem{corollary}[theorem]{Corollary}
\newtheorem{proposition}[theorem]{Proposition}
\newtheorem{definition}{Definition}
\newdefinition{rmk}{Remark}
\newproof{pf}{Proof}
\newproof{pot}{Proof of Theorem \ref{thm}}

\begin{document}
\let\WriteBookmarks\relax
\def\floatpagepagefraction{1}
\def\textpagefraction{.001}

\shorttitle{}

\shortauthors{J A Mota; I Leal de Freitas; L Rufino; J Paixão}

\title [mode = title]{Least-Squares and Low-Rank Approximation for Linear Relations Using a Diagrammatic Language}

\author[1]{Júlia de Araújo Mota}[orcid=0000-0002-9547-9629]
\cormark[1]
\ead{juliamota@ufrj.br}
\credit{}

\author{Iago {Leal de Freitas}}[orcid=0009-0001-6813-5863]
\ead{iago@iagoleal.com}
\credit{}

\author[2]{Lucas Rufino}[orcid=0009-0009-2402-7949]
\ead{rufino.martelotte.lucas.c5@s.mail.nagoya-u.ac.jp}
\credit{}

\author[1]{João Paixão}[orcid=0000-0002-1610-9022]
\ead{jpaixao@ic.ufrj.br}

\credit{}

\affiliation[1]{organization={Universidade Federal do Rio de Janeiro},
            addressline={Instituto de Computação, Avenida Athos da Silveira Ramos, 274},
            city={Rio de Janeiro},
           postcode={21941-916},
            state={RJ},
            country={Brazil}}

\affiliation[2]{organization={Nagoya University},
            addressline={Furocho, Chikusaku},
            city={Nagoya},
            postcode={464-8602},
            state={Aichi},
            country={Japan}}

\cortext[1]{Corresponding author}

\begin{abstract}
We employ the machinery of linear relations to the study of optimization problems in linear algebra. We first show that the relational version of the pseudo-inverse can be realized through a generalization of the least-squares problem. This allows one to prove that the pseudo-inverse realizes the solution of certain relational optimization problems. Our main result is showing that a certain truncation of this pseudo-inverse defines a solution to a relational version of the classical low-rank approximation problem which recovers both the Eckart-Young Theorem and several optimization problems involving pairs of matrices and vector spaces.
\end{abstract}


\maketitle

\section{Introduction}\label{sec:introduction}

Low-rank approximation is the name given to the problem of approximating a given real matrix $A$ by a matrix $S$ of lower rank. Solving this problem is one of the most standard and fundamental routines in linear algebra. More concretely, for $0 \leq k \leq {\rm rank} \; A$, one is interested in finding a solution to the optimization problem
\begin{equation}\label{eq:intro_low_rank_approximation} \underset{S}{\rm argmin} \left[  \| A - S \| \;:\; {\rm rank}\;S = k \right],
\end{equation}
where $\| \cdot \|$ is the Frobenius norm (i.e. finding a minimizer $S$ of the function $\|A-S\|$ over the set of all matrices of rank $k$). A solution is constructed as follows. First, one computes the Singular Value Decomposition (SVD) of $A$, written as $A = U\Sigma V^T$. The matrix $\Sigma$ is diagonal with non-decreasing diagonal entries called the singular values of $A$. Let $\Sigma_k$ be the $k$-th truncation of $\Sigma$, i.e. the diagonal matrix whose $(i, i)$-th entries are zero if $i > k$, otherwise coincide with those of $\Sigma$. Then, we have the Eckart-Young Theorem (EYT).

\begin{theorem}[Eckart-Young]\label{thm:intro_eckart_young} The optimization problem~\eqref{eq:intro_low_rank_approximation} has a (unique) solution given by the $k$-th truncation of $A$, i.e.
\begin{equation}\label{eq:intro_truncation_of_A}A_k := U \Sigma_k V^T.\end{equation}
\end{theorem}

\noindent This essentially answers all questions one might have about~\eqref{eq:intro_low_rank_approximation}. For some natural variants of the problem, however, the solution is not so easy. For example, let $W$ be a subspace of the domain of $A$. The \emph{restricted} low-rank approximation problem \begin{equation}\label{eq:intro_restricted_low_rank_approximation} \underset{S}{\rm argmin} \left[ \| A - S \| \;:\; {\rm rank}\; \left. S \right|_W = k \right]
\end{equation}
is very hard to solve using solely the SVD. One way to solve it is via an important work by De Moor and Golub~\cite[Theorem 6]{Golub:1991} which describes a solution (as well as sufficient conditions for the existence of a solution) to the more general optimization problem
\begin{equation}\label{eq:intro_golubs_triple_matrix_optimization_problem}
\underset{S}{\rm argmin} \left[ \|S\| \;:\; {\rm rank}(A - BSC)\right].
\end{equation}
More concretely, a solution is given in terms of the so-called Restricted Singular Value Decomposition (RSVD) of the triplet $(A,B,C)$. Going back to the restricted low-rank approximation problem~\eqref{eq:intro_restricted_low_rank_approximation}, if one first computes a matrix $M_W$ whose image is $W$, then
\begin{align*}
&\;\;\;\;\;\;\; T \in \underset{S}{\rm argmin} \; \left[\|A-S\| \;:\; {\rm rank}\;\left. S \right|_W = k \right] \\
&\Leftrightarrow\; T \in \underset{S}{\rm argmin} \; \left[\|A-S\| \;:\; {\rm rank} \; (SM_W) = k \right] \\
&\Leftrightarrow\; A - T \in \underset{S}{\rm argmin} \; \left[\|S\| \;:\; {\rm rank} \; ((A - S)M_W) = k \right] \\
&\Leftrightarrow\; A - T \in \underset{S}{\rm argmin} \; \left[\|S\| \;:\; {\rm rank} \; (AM_W - SM_W) = k \right].
\end{align*}
This is now in the format of~\eqref{eq:intro_golubs_triple_matrix_optimization_problem}, and so the problem is solved by applying the RSVD for the matrix triplet $(AM_W, I, M_W)$.

This solution, even though implementable, is arguably not completely satisfactory from a conceptual standpoint for a number of reasons, some examples given as follows.
\begin{itemize}
\item One first has to compute a basis or a generating set $M_W$ of the space $W$.
\item The existence of a solution is hard to describe. Indeed, a solution exists if and only if $r_{ab} + r_{ac} - r_{abc} \leq k \leq \min({r_{ab}, r_{ac}})$, where the terms $r_{ab}, r_{ac}, r_{abc}$ are ranks of different blocked matrices involving the triplet $(A, B, C)$ (see Appendix~\ref{ap:functionalizer} for a precise definition).
\item The solution itself is delicate to describe: even defining the RSVD requires chopping up a blocked matrix into another ten-by-ten blocked matrix (100 blocks!), and the size of each block is given by a certain combination of the numbers in the previous item (see~\cite[Section 2]{Golub:1991}).
\end{itemize}

\noindent These conceptual issues often occur in linear algebra. Sometimes they are an inevitable reflex of the sheer difficulty of the problem. Other times, however, they are avoidable, which is often the case for problems mixing linear maps and vector spaces, just as~\eqref{eq:intro_restricted_low_rank_approximation}. The key obstacle usually lies in the fact that the SVD, the RSVD and other famous matrix decompositions cannot receive vector spaces as input. This forced us to first compute a basis of $W$; in other cases one might have to resort to set-theoretical notation. A possible remedy is to change the language: instead of working with linear maps and vector spaces as if they were objects of different nature, one instead declares linear relations as the primary object of study.

\subsection{The workflow of linear relations}


Let $V, W$ be two vector spaces. A linear relation $R$ from $V$ to $W$ is a subspace of the product $V \times W$, and we write $R \subseteq V \times W$. It is also standard to write $xRy$ to mean $(x, y) \in R$. Any vector space $V$ can be seen as a linear relation via the trivial inclusion $V \subseteq \{0\} \times V$. Moreover, any linear transformation $M : V \rightarrow W$ can be seen as a linear relation via its graph ${\rm Gr}(M) := \{(x, y) \;:\; Mx = y\} \subseteq V \times W$. We say a linear relation $R$ is (or corresponds to) a linear transformation $M$ if it is equal to ${\rm Gr}(M)$.


We now want to illustrate what can be called the "workflow" of linear relations, that is, how can one use them to sometimes put theorems of seemingly different natures into the same context. Let us consider, as an example, perhaps the most famous theorem in linear algebra: the Invertible Matrix Theorem (IMT). In order to be precise (since the IMT is often presented as an equivalence of more than 10 statements), we mean the following.

\begin{theorem}[Invertible Matrix Theorem] For a square matrix $M$, the following are equivalent.
\[\begin{tabular}{l l l l}
    (i) & $M$ is surjective, & (iv) & $M$ is injective, \\
    (ii) & $M$ has a right-inverse, & (v) & $M$ has a left-inverse, \\
    (iii) & $M$ has full rank, & (vi) & $M$ has an inverse.
\end{tabular}\]
\end{theorem}

\noindent The IMT therefore "connects" together a number of statements concerning a single matrix $M$. To prove it, it is costumary to make use of a matrix decomposition, e.g. the canonical decomposition
\begin{equation}\label{eq:intro_canonical_decomposition}
M = U \begin{bmatrix}
I & 0 \\ 0 & 0
\end{bmatrix} V,
\end{equation}
where $U,V$ are invertible matrices. It turns out that this decomposition can be generalized in a natural way to linear relations and, as a result, one can prove an analogous characterization theorem about linear relations which reduces to the IMT. Moreover, linear relations being more general than linear maps, this new theorem encodes much more information than the original IMT, which is not limited to statements about a single matrix anymore. This includes results about the existence of solutions to linear systems, and also the Exchange Lemma, which is a foundational lemma about vector spaces present in any proof-focused linear algebra textbook. For a more detailed account of this story, see~\cite{freitas2025imt}. Figure~\ref{fig:intro_relations_connecting_theorems} illustrates the usual "workflow" of linear relations and how they can be used to put results mixing linear maps and vector spaces in the same context as results in which only one of the two appear.

\begin{figure}[pos=htb]
    \centering
    \begin{tikzcd}
\begin{matrix} \text{relational} \\ \text{version of} \\ \text{decomposition }D \end{matrix} \arrow[d, "{\rm specializes} \; {\rm to}"] \arrow[rr, "\text{used to prove}"] &  & \begin{matrix} \text{relational} \\ \text{version of} \\ \text{theorem }T \end{matrix} \arrow[d, "{\rm specializes}\;{\rm to}"] \arrow[rr, "\text{specializes to}"] &  & \begin{matrix} \text{theorems} \\  \text{mixing} \\ \text{matrices and} \\ \text{vector spaces} \end{matrix} \\
{\rm decomposition} \; D \arrow[rr, "\text{used to prove}"] &  & {\rm theorem}\; T  &  &
\end{tikzcd}
    \caption{How linear relations are used to "connect" theorems together.}
    \label{fig:intro_relations_connecting_theorems}
\end{figure}
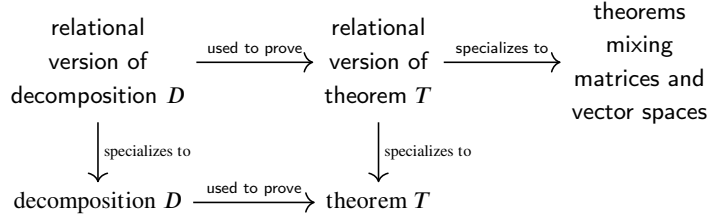

\subsection{Objectives}

Following the notation of Figure~\ref{fig:intro_relations_connecting_theorems}, the goal of this paper, roughly speaking, is to apply this workflow when $D$ is the SVD and $T$ is the Eckart-Young Theorem. This is not so straightforward as the canonical decomposition of matrices mentioned earlier, the reasons being as follows.
\begin{itemize}
    \item First, one needs to make sense of what the word "rank" means for linear relations. This is not a trivial question, since, for example, in the canonical decomposition of a linear relation one needs three numbers to describe the middle dimensions (as opposed to just one, in the case of linear maps). In this paper, we adopt a simpler approach and define the rank of a linear relation as the dimension of its image, i.e. \[{\rm rank}\; R := {\rm dim} \{y \;:\; \exists x, \; xRy\}.\] This of course matches the usual definition of rank when $R$ is a linear transformation.
    \item The second problem is that the term $\| R - S \|$ in~\eqref{eq:intro_low_rank_approximation} is not well-defined since there is not a good candidate for what the norm of a linear relation should be. If $R$ corresponded to a linear transformation, we could make the (involutive) change of variables $S \mapsto R - S$ and make it so only $S$ appears inside of the norm. This "modified" version of the problem is the one we will consider in this paper.
\end{itemize}
Putting these two considerations together, the main result of this paper Theorem~\ref{thm:Eckart-Young-gsvd}, which describes a solution (together with a sufficient condition for its existence) to the relational low-rank approximation problem
\begin{equation}\label{eq:intro_relational_low_rank_optimization} \underset{S}{\rm argmin} \left[ \|S\| \;:\; {\rm rank}(R - S) \right],
\end{equation}
where $R - S := \{(x, y_1 - y_2) \;:\; xRy_1 \text{ and } xSy_2 \}$ and $S$ is a cofunction (i.e. it equals $\{(x, y) \;:\; x = \widetilde{S}y\}$ for some linear transformation $\widetilde{S}$; here we define $\|S\| := \|\widetilde{S}\|$). When $R$ itself is a cofunction, this reduces to the usual low-rank approximation problem~\eqref{eq:intro_low_rank_approximation} as we recover the Eckart-Young Theorem up to the uniqueness statement. Therefore, for the purposes of this explanation, let us call our main result the Generalized Eckart-Young Theorem (GEYT). What about when $R = A$ is a linear transformation? If $k = 0$, the relational low-rank approximation problem reduces to
\begin{equation}\label{eq:intro_pseudo_inverse_optimization_problem}
\underset{S}{\rm argmin} \left[\|S\| \;:\; SA = I\right],
\end{equation}
where now $S$ is a linear transformation. In other words, it is the problem of finding the smallest left-inverse of $A$. The solution is known to be the pseudo-inverse $A^+$ of $A$; it is the unique linear transformation satisfying the four Moore-Penrose equations~\citep{Penrose:1955}, and can be characterized by the fact that it is a solution to the least-squares problem
\[A^+y \in \underset{x}{\rm argmin} \; \|Ax - y\|.\]
If $k > 0$, we get a truncation (in the same sense of~\eqref{eq:intro_truncation_of_A}) of the pseudo-inverse. This is the first conceptual advantage of the relational method: the Eckart-Young Theorem and the solution to the pseudo-inverse optimization problem~\eqref{eq:intro_pseudo_inverse_optimization_problem} are particular cases of the same result.

In the same way that the proof of EYT uses the SVD, the proof of the GEYT uses a relational version of the SVD called the Generalized Singular Value Decomposition (GSVD), first introduced by~\cite{vanloan1976} and later developed by many others~\citep{paige1981,sun1983,vanloan1985,Betcke:2008,Gockenbach:2016,Edelman:2020}. We therefore have completed the "workflow" diagram from Figure~\ref{fig:intro_relations_connecting_theorems} with the exception of the mixed theorems, as illustrated in Figure~\ref{fig:lift_of_svd}.

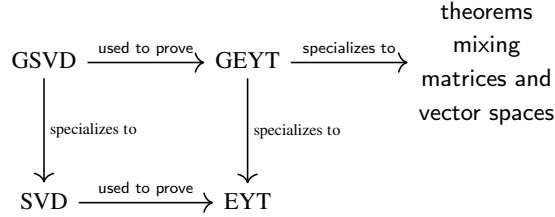
\begin{figure}[pos=htb]
    \centering
    \begin{tikzcd}
{\rm GSVD} \arrow[d, "{\rm specializes} \; {\rm to}"] \arrow[rr, "\text{used to prove}"] &  & {\rm GEYT} \arrow[d, "{\rm specializes}\;{\rm to}"] \arrow[rr, "\text{specializes to}"] &  & \begin{matrix} \text{theorems} \\  \text{mixing} \\ \text{matrices and} \\ \text{vector spaces} \end{matrix} \\
{\rm SVD} \arrow[rr, "\text{used to prove}"] &  & {\rm EYT}  &  &
\end{tikzcd}
    \caption{The linear relation workflow applied to the Eckart-Young Theorem.}
    \label{fig:lift_of_svd}
\end{figure}

\noindent The remaining question is what fits in the remaining node of this diagram? In other words, what theorems we get by instantiating $R$ as mixtures of linear maps and vector spaces? It turns out that, given a linear transformation $A$ and a subspace $W$ of the domain of $A$, if we instantiate $R$ as the linear relation
\begin{equation}\label{eq:intro_restricted_low_rank_approximation_instance}
\{(x, y) \;:\; x = Ay, \; y \in V\},
\end{equation}
we specialize to the restricted low-rank approximation problem~\eqref{eq:intro_restricted_low_rank_approximation}, thus showing that both the restricted and non-restricted low-rank approximation problems are solved by applying the same theorem to different inputs (notice that we did not have to compute a basis for $W$!). This, of course, is not the only approximation problem obtained from the GEYT. There are many other variants, seven of which can be seen, in a sense, dual versions of the restricted low-rank approximation problem; these are presented in Table~\ref{tab:instantiations_of_geyt}.

\begin{table}[pos=htb]
  \centering
  \begin{tabular}{c l r}
    {$R$ in Graphical Syntax} & {$R$ in Point-full Syntax} & {Optimization Condition} \\ \hline
    \scalebox{0.8}{\input{subspace_intro_7.tikz}} & $\{(x, y) \;:\; x = Ay, \; y \in W\}$ & ${\rm rank} \left.(A - S)\right|_W = k$ \\ \hline
    \scalebox{0.8}{\input{subspace_intro_6.tikz}} & $\{(x+v, y) \;:\; x = Ay, \; v \in V\}$ & ${\rm dim}( {\rm Img}(A - S)  + V) = k$ \\ \hline
    \scalebox{0.8}{\input{subspace_intro_1.tikz}} & $\{(x, y) \;:\; Ax = y, \; x \in W\}$ & ${\rm rank} \left. (I - SA) \right|_W = k$ \\ \hline
    \scalebox{0.8}{\input{subspace_intro_2.tikz}} & $\{(x+w, y) \;:\; Ax = y, \; w \in W\}$ & ${\rm dim}({\rm Img}(I - SA) + W) = k$ \\ \hline
    \scalebox{0.8}{\input{subspace_intro_3.tikz}} & $\{(x, y) \;:\; Ax = y, \; y \in V\}$ & ${\rm rank} \left. (I - SA) \right|_{A^{-1}[V]} = k$ \\ \hline
    \scalebox{0.8}{\input{subspace_intro_4.tikz}} & $\{(x, y+v) \;:\; Ax = y, \; v \in V\}$  & ${\rm dim}( {\rm Img}(I - SA) + SV) = k$ \\ \hline
    \scalebox{0.8}{\input{subspace_intro_5.tikz}} & $\{(x, y) \;:\; x = Ay, \; x \in V\}$ & ${\rm rank} \left. (A - S) \right|_{A^{-1}[V]} = k$ \\ \hline
    \scalebox{0.8}{\input{subspace_intro_8.tikz}} & $\{(x, y+w) \;:\; x = Ay, \; w \in W\}$  & ${\rm dim}({\rm Img}(A - S) + AW) = k$
  \end{tabular}
  \caption{Instantiations of GEYT, where $A$ is a linear transformation, $W$ is a subspace of the domain of $A$ and $V$ is a subspace of the codomain of $A$. Each line shows $R$ in two different syntaxes and, on the last column, the condition on the comap $S$ in the resulting optimization problem, i.e. $(\text{optimization problem}) = \underset{S}{\rm argmin} \left[ \|S\| \;:\; (\text{optimization condition}) \right]$.}
  \label{tab:instantiations_of_geyt}
\end{table}

Before proceeding, let us clarify some aspects of Table~\ref{tab:instantiations_of_geyt}. The reader may be a little confused about the first column; it contains exactly the same information as the second column, but the relation $R$ is now written using a syntax of string diagrams~\citep{selinger2010survey,BaezErbele-CategoriesInControl} for Linear Algebra developed by \citet{ZanasiThesis} called Graphical Linear Algebra (GLA)~\citep{bonchi2014categorical,Bonchi2015,bonchi2017refinement,Bonchi2019a, PAIXAO2022, blogpawel}. Of course, the reader is not expected to understand the meaning of these diagrams yet. We, however, chose to include this column here to allow the reader to appreciate the fact that, even though the instances in the second column may feel arbitrary and/or taken "out-of-the-hat", the graphical syntax reveals how strikingly similar they truly are. According to \citet{hinze2023introducing}, string diagrams are a good tool for point-free reasoning and type-checking, and, in the authors' experience, they often greatly improve the readability of proofs. For that, and many other reasons, almost all of the proofs in this paper are done using the graphical syntax. Moreover, all proofs follow the calculational style~\citep{Dijkstra1989PredicateCA}, structured as sequences of diagram manipulations. Each step includes a ``proof hint'' in brackets, placed to the right of a diagram or statement, explaining which rule or result was used to arrive at it from the previous step.

\subsection{Text Structure}

Section~\ref{sec:dictionary} presents the theoretical foundations of this work, that is, the \emph{syntax} and \emph{decompositions} used throughout the text.
Some concepts of Linear Algebra and Relations are revisited. Diagrammatic notation is presented, along with an explanation of how to reason with it. The decompositions used in the study of linear relations are introduced. The main objective of this chapter is to provide the reader with a ``user manual'' for understanding the concepts, notation, and results that will be developed in the following chapters.

Section~\ref{sec:LQ} introduces the linear relations operations, \emph{generalizing the Least Squares problem} to the context of linear relations. Building upon these characterizations, we revisit the definition of the pseudoinverse of relations and propose algorithms to calculate it. Instantiating $R$ as a linear transformation, we obtain the Moore-Penrose pseudoinverse as a particular case. Finally, we present some applications of the problem in optimization problems.

Section~\ref{sec:LR}, \emph{generalizes the Low-Rank Approximation problem} to the context of linear relations. We introduce a low-rank approximation of a linear relation obtained from the truncation of the GSVD  and prove that this is the cofunction with the lowest rank and smallest norm that approximates the linear relation. Several instances of the problem are explored, showing how this approach recovers classical results.

\section{Theoretical Foundations}\label{sec:dictionary}

This section presents the \emph{syntax} and \emph{decompositions} used in this work. Its main objective is to provide the reader with a ``user manual'' for understanding the concepts, notations, and results that will be developed in the following chapters. Whenever a proof is omitted, it is because it can be found in one of the published works \citep{MOTA2025,freitas2025imt}.

Throughout this work, we represent linear relations with \emph{string diagrams syntax}.
As noted by~\cite{hinze2023introducing}, these are well-known objects in category theory that allow equational reasoning without loss of type information.
They characterize the arrows of free strict monoidal categories~\citep{selinger2010survey} and are, therefore, rigorous mathematical objects.
\citet{PAIXAO2022} apply string diagrams to develop a graphical syntax for Linear Algebra,
known as Graphical Linear Algebra (GLA), which presents the laws of linear relations in an equational manner.

\subsection{Graphical Notation for Linear Relations}
\label{sec:syntax}

We now present an introduction to how diagrams are constructed and outline the structures of how to manipulate and reason with them. Most of the results presented here have been published in \cite{MOTA2025}. For a more thorough account of Graphical Linear Algebra (GLA), we recommend the blog series by Pawe{\l}  Soboci\'{n}ski~\citep{blogpawel}.

An $n$-by-$m$ \emph{linear relation} can be written as a two-dimensional diagram with $m$ dangling wires on the left and $n$ on the right (the chamfered edge on the right is necessary because eventually we will need to start ``flipping diagrams horizontally”, so it is important to retain directional information).
\begin{equation}
  R \subset \mathbb{R}^n \times \mathbb{R}^m \mapsto \quad \input{Introdution_GLA_n_wired_dots_relation.tikz}.
\end{equation}
For example, a $2$-by-$1$ linear relation $R$ can be written as a diagram \input{Introdution_GLA_1_by_2_relation.tikz} with one wire on the left and two wires on the right. Generally, when $R$ is a $n$-by-$m$ linear relation, it is more convenient to write $\TypedRel{R}{m}{n}$, and when there is no confusion about the type information, we may omit the numbers $m, n$ and write it as $\Rel{R}$. Lastly, when one of the numbers is zero, say, $m=0$,  then $R$ is called a \emph{linear subspace} and we omit the dangling wire on the left, writing $\begin{tikzpicture}
	\begin{pgfonlayer}{nodelayer}
		\node [style=Rel] (0) at (0, 0) {$R$};
		\node [style=none] (1) at (0.5, 0) {};
		\node [style=none] (5) at (1.75, 0) {};
	\end{pgfonlayer}
	\begin{pgfonlayer}{edgelayer}
		\draw (1.center) to (5.center);
	\end{pgfonlayer}
\end{tikzpicture}
$.

We also reserve different symbols based on how specialized the linear relation is: while the box symbol $\Rel{R}$ denotes a general linear relation, the curvy symbol $\Map{A}$ is used when we know $A$ is a linear transformation. Moreover, the blue curvy symbol $\Inv{A}$ is used when $A$ is known to have an inverse (usually denoted by $A^{-1}$), and the red curved symbol $\Ortho{U}$ denotes orthogonal linear transformations.

In diagrammatic notation, the relational composition and Cartesian product operations become, respectively, the action of connecting diagrams horizontally and stacking them vertically, as shown in Definition~\ref{def:composing_diagrams}. The opposite relation becomes simply the action of flipping the diagram to the opposite side.

\begin{definition}[Operations on Linear Relations in Diagrams]\label{def:composing_diagrams} The main operations on linear relations in diagrams are the following.
  \begin{center}
    \begin{enumerate}
      \item \label{def:composing_diagrams:seq}
            $\input{Introdution_GLA_composing_diagrams_1.tikz}$,
      \item \label{def:composing_diagrams:par}
            $\input{Introdution_GLA_composing_diagrams_2.tikz}$,
      \item \label{def:composing_diagrams:op}
            $\input{Introdution_GLA_composing_diagrams_3.tikz}$,
      \item \label{def:composing_diagrams:perp}
        $\TypedRel{A^\top}{m}{n} = \TypedCoRelGray{A}{m}{n}$.
    \end{enumerate}
  \end{center}
\end{definition}

Notice that, similar to relational composition, to compute $\Rel{R} ; \Rel{S}$, $R$ and $S$ have to be compatible: the number of right-wires of $R$ must be the same as the number of left-wires of $S$.

There are specific symbols commonly used for some linear relations, called identity, twist, zero, sum, copy, and discard. All of these graphs represent linear transformations and act as ``building blocks'' for more complex diagrams. Table~\ref{tab:generators} shows their diagrammatic symbols and corresponding linear transformations. Semantically, each diagram canonically represents a linear relation. We leave the details of how to manipulate such diagrams to Appendix~\ref{subsec:translating}.

\begin{table}[pos=htb]
  \centering
  \renewcommand{\arraystretch}{2.5}
  \begin{tabular}{r l l l}
    (Identity) & \TypedId{n} & $x \mapsto x$ & $: \mathbb{R}^n \to \mathbb{R}^n$  \\
    (Twist) & \input{Introdution_GLA_n_wired_twist.tikz} & $(y \in \mathbb{R}^m, x \in \mathbb{R}^n) \mapsto (y, x)$ & $: \mathbb{R}^{m+n} \to \mathbb{R}^{m+n}$ \\
    (Zero) & \TypedZero{n} & $y \mapsto 0$ & $: \mathbb{R}^0 \to \mathbb{R}^n$ \\
    (Sum) & \input{Introdution_GLA_n_wired_sum.tikz} & $(x \in \mathbb{R}^n, y \in \mathbb{R}^n) \mapsto x + y$ & $: \mathbb{R}^{2n} \to \mathbb{R}^n$  \\
    (Discard) & \TypedDiscard{n} & $x \mapsto 0$ & $: \mathbb{R}^n \to \mathbb{R}^0$ \\
    (Copy) & \input{Introdution_GLA_n_wired_copy.tikz} & $x \mapsto (x, x)$ & $: \mathbb{R}^n \to \mathbb{R}^{2n}$ \\
  \end{tabular}
  \caption{Special relations together with their names and graphical syntax.}
  \label{tab:generators}
\end{table}

Subspaces can be thought of as relations $R \subseteq V \times W$
where $V = \{*\}$ is the zero-dimensional space.
In diagrammatic notation, these correspond to the diagrams with 0 wires on the left.
In particular, the zero diagram \TypedZero{n} can be thought of as the zero subspace $\{0\} \subseteq \mathbb{R}^n$, whereas the opposite of the discard, \TypedCoDiscard{n}, can be thought of as the full space $\mathbb{R}^n \subseteq \mathbb{R}^n$.

\begin{definition}[Fundamental Subspaces of a Linear Relation]\label{def:fund_subspaces} For every linear relation $\Rel{R}$, we define:
  \begin{center}
    \input{Introdution_Dictionary_fund_subespaces.tikz}
  \end{center}
\end{definition}

With this definition, we can define surjectivity and injectivity graphically as
\begin{center}
\renewcommand{\arraystretch}{2.0}
\begin{tabular}{l l l l l}
    $A \text{ is injective}$ & $\iff$  &
    $\Ker(A) \subseteq \{0\}$ &
    $\iff$ &
    \input{Introdution_Maps_A_is_injective.tikz},\\

    $A \text{ is surjective}$ & $\iff$ &
    $\Img(A) \supseteq \mathbb{R}^n$ &
    $\iff$ &
    \input{Introdution_Maps_A_is_surjective.tikz}.
\end{tabular}
\end{center}
In fact, as a side note, linear transformations can be defined as linear relations satisfying $\input{Introdution_Maps_R_is_det.tikz}$ and $\input{Introdution_Maps_R_is_tot.tikz}$. If these equations hold, we say the relation $R$ is, respectively, deterministic and total.  More concretely, all four properties (injectivity, surjectivity, determinism and totality) can be written in terms of the white and black structures (see Figure~\ref{fig:typerelations}).

The diagrammatic manipulation rules, as well as useful theorems are presented in Appendix~\ref{ap:gla}. The appendix also showcases rules for converting between the graphical and the usual syntax for matrices and linear relations.

\subsection{Orthogonal Projections}
\label{sec:orthogonal-projections}

For a fixed subspace $\Subspace{W}$, we denote by $\Map{P_W}$ the \emph{orthogonal projection} onto $\Subspace{W}$, defined as follows via linear relations:

\begin{definition}[Orthogonal Projection]
  \label{def:proj-rel}
  The orthogonal projection $\Map{P_W}$ into a subspace $\Subspace{W}$
  is the linear relation:
  \[\Map{P_W} = \input{orthogonal-projection.tikz}.\]
\end{definition}

In the following, we briefly recall the properties of an orthogonal projection in terms of linear relations.

\begin{lemma}
  \label{lem:proj-props}
  An orthogonal projection satisfies
  \begin{enumerate}
    \item \label{lem:proj-props:img} $\Image{P_W}  = \Subspace{W}$,
    \item \label{lem:proj-props:ker} $\Kernel{P_W} = \begin{tikzpicture}
	\begin{pgfonlayer}{nodelayer}
		\node [style=SubspaceGray] (4) at (0, 0) {$W$};
		\node [style=none] (5) at (1.5, 0) {};
	\end{pgfonlayer}
	\begin{pgfonlayer}{edgelayer}
		\draw (4) to (5.center);
	\end{pgfonlayer}
\end{tikzpicture}
$,
    \item \label{lem:proj-props:idem} $\Map{P_W} = \Comp{P_W/Map, P_W/Map} = \CoMapGray{P_W}$.
  \end{enumerate}
\end{lemma}

\subsection{Generalized Singular Value Decomposition  for Linear Relations}\label{sec:gsvd}

Decomposition of linear transformations and linear subspaces are widely used tools in Linear Algebra,
as they reveal fundamental properties of these objects.
For linear relations, decomposition plays a similar role.
Results such as the canonical decomposition of linear relations~\citep{hassi2006canonical} and the characterization of canonical forms equivalent to Kronecker's~\citep{berger2015} allow us to separate regular and singular components, providing a clear description of the kernels, images, and dimensions involved.
More recently, we have shown how relational decompositions directly connect classical concepts of invertibility and rank to the general case of linear relations~\citep{freitas2025imt}.

The study of linear relations allows the application of a classical result: every linear relation can be expressed both in the \emph{span form} and in the \emph{cospan form}. This is described in Proposition~\ref{prop:span-cospan}.

\begin{proposition}[Span/Cospan form]\label{prop:span-cospan}
For every linear relation $\TypedRel{R}{m}{n}$, there exist $\TypedMap{A}{p}{n}$, $\TypedMap{B}{p}{m}$, $\TypedMap{\tilde{A}}{m}{q}$ and $\TypedMap{\tilde{B}}{n}{q}$ such that:
\begin{center}
\begin{tabular}{r r r l}
  \textbf{(Span)} & $\TypedRel{R}{m}{n}$ &$=$&$\TypedSpan{B}{A}{m}{p}{n}$; \\[1em]
  \textbf{(Cospan)} & $\TypedRel{R}{m}{n}$ &$=$&$\TypedCoSpan{\tilde{A}}{\tilde{B}}{m}{q}{n}$.
\end{tabular}
\end{center}
\end{proposition}

It is always possible to switch between the \emph{span} and \emph{cospan} forms. Throughout the text, both representations will be employed according to the convenience. Expressing a linear relation in the \emph{span} form is particularly useful for constructing visualizations, as illustrated in Example~\ref{ex:visual-span}, whereas the \emph{cospan} form is advantageous for verification: a pair $(x,y)$ belongs to the relation if and only if the equality $\tilde{A}x=\tilde{B}y$ holds.

\begin{example}[Span Visualization]\label{ex:visual-span}
Let $R \subseteq \mathbb{R}^{3} \times \mathbb{R}^3$ be a linear relation in span form, as described in Proposition~\ref{prop:span-cospan}. We choose $A=US$ and $B=VC$, where $U,V$ are orthogonal linear transformations, and \[S=\begin{bmatrix}1 & 0 & 0 \\0 & s & 0 \\0 & 0 & 0 \\\end{bmatrix},\quad C=\begin{bmatrix}0 & 0 & 0 \\0 & c & 0 \\0 & 0 & 1 \\\end{bmatrix}\] with $s,c \in \mathbb{R}$ such that $0< s,c< 1$. To illustrate how we can construct visualizations using \emph{span}, we apply the linear transformations induced by $A$ and $B$ to the vectors \[w_1=\begin{bmatrix}1 \\ 0 \\ 0\end{bmatrix},\quad w_2=\begin{bmatrix}0 \\ 1 \\ 0 \end{bmatrix},\quad w_3=\begin{bmatrix}0 \\ 0 \\ 1 \end{bmatrix}\quad \text{and}\quad w_4=\begin{bmatrix}1 \\ 1 \\ 1 \end{bmatrix}.\]
\begin{figure}[pos=h]
    \centering
    \input{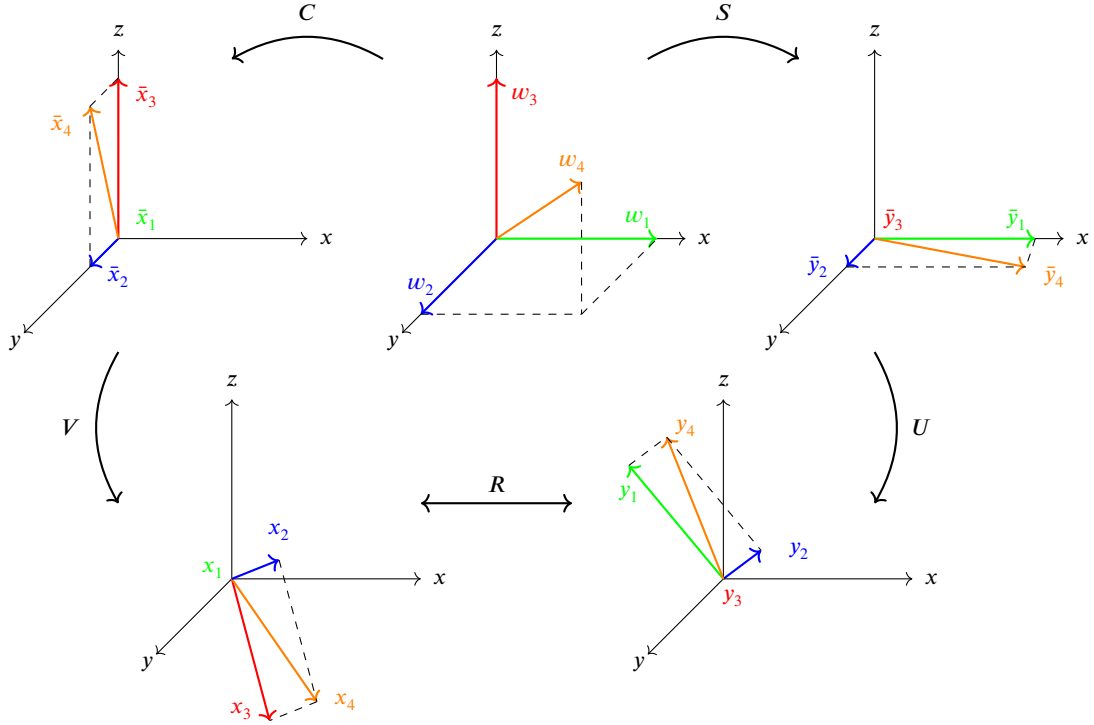}
  \caption{Visualization of a linear relation $R \subseteq \mathbb{R}^{3} \times \mathbb{R}^3$ in span form.}
  \label{fig:visual-span}
\end{figure}
In Figure~\ref{fig:visual-span}, the vectors $w_i$ are shown in the upper central illustration. On the right, their images $\bar{y}_i = Sw_i$ are represented, and on the left, the corresponding images $\bar{x}_i = Cw_i$. Applying the orthogonal linear transformations $U$ and $V$ produces the lower illustrations, where $y_i = U\bar{y}_i$ and $x_i = V\bar{x}_i$. For the linear relation $R$, however, only the two lower figures are observable: we know the pairs $(x_i, y_i)$ but not the intermediate representations.
\end{example}

The choices of linear transformations $A,B$ in the Example~\ref{ex:visual-span} were not random. Through a tool known as \emph{Generalized Singular Value Decomposition} (GSVD), we can simultaneously decompose a pair of linear transformations into linear transformations with non-negative diagonal entries, structured by blocks as in $S$ and $C$, plus two rotations. The GSVD was initially introduced by \citet{vanloan1976} and later developed by many others~\citep{paige1981,sun1983,vanloan1985,Betcke:2008,Gockenbach:2016,Edelman:2020}.

\begin{theorem}[Generalized Singular Value Decomposition]\label{thm:gsvd}
  For every pair of linear transformations $\TypedMap{A}{p}{n}$ and $\TypedMap{B}{p}{m}$
  there exist numbers $q$, $k_I$, $k_S$, $k_T$, $k_D$, $d$ $\in \mathbb{N}$,
  orthogonal linear transformations $\TypedOrtho{U}{n}{n}$ and $\TypedOrtho{V}{m}{m}$,
  invertible diagonal linear transformations $\TypedInv{S}{d}{d}$ and $\TypedInv{C}{d}{d}$,
  and a surjective linear transformation $\TypedMap{H}{p}{q}$
  such that
  \[
    \TypedMap{A}{p}{n}=\input{GSVD_GSVD_A.tikz},
    \quad\quad
    \TypedMap{B}{p}{m}=\input{GSVD_GSVD_B.tikz}.
  \]
\end{theorem}

Note that the linear transformation $\Map{H}$ does not possess orthogonality or diagonality properties --- we can only guarantee its surjectivity.
Nevertheless, it does not stand a problem, since we will see below that $\Map{H}$ vanishes when we apply Theorem~\ref{thm:gsvd} to a linear relation in span form as in Proposition~\ref{prop:span-cospan}.

\begin{theorem}[GSVD for Linear Relations]
  \label{thm:gsvd-relations}
For every linear relation $\TypedRel{R}{m}{n}$, there exist numbers $k_T$, $k_D$, $k_I$, $k_S$, $d \in \mathbb{N}$, orthogonal linear transformations $\TypedOrtho{U}{n}{n}$ and $\TypedOrtho{V}{m}{m}$, and invertible linear transformations $\TypedInv{S}{d}{d}$ and $\TypedInv{C}{d}{d}$, such that
  \[\TypedRel{R}{m}{n}=\input{GSVD_GSVD.tikz} \]
\end{theorem}
\begin{proof}
  \begin{hcalculation}[=]{\scalebox{0.8}{\TypedRel{R}{m}{n}}}
    \hstep{\scalebox{0.8}{\TypedSpan{B}{A}{m}{p}{n}}}{Prop.\ref{prop:span-cospan}: Span}
    \hstep{\scalebox{0.8}{\input{GSVD_step1.tikz}}}{Thm.\ref{thm:gsvd}: GSVD}
    \hstep{\scalebox{0.8}{\input{GSVD_GSVD.tikz}}}{Fig.\ref{fig:typerelations}-DET $+$ SUR: $\scalebox{0.7}{\Span{H}{H}}=\scalebox{0.7}{\Id}$ }
  \end{hcalculation}
\end{proof}

As this result will be used several times throughout this work, in order to avoid repetition, Definition~\ref{def:diagonalrelation} defines the ``center'' part of Theorem~\ref{thm:gsvd-relations} as a \emph{diagonal linear relation}.

\begin{definition}[Diagonal Linear Relation in Graphical Version]\label{def:diagonalrelation}
Let $m, n, k_I, k_S, k_T, k_D, d \in \mathbb{N}$ such that \(m = k_T + d + k_I\) and \(n = k_D + d + k_S\), and $\TypedInv{S}{d}{d}$ and $\TypedInv{C}{d}{d}$ be invertible diagonal linear transformations. The diagonal linear relation is defined by
\[\TypedRel{D}{m}{n}=\input{GSVD_GSVD_D.tikz}.\]
\end{definition}

\begin{remark}
From now on, for any linear relation $\Rel{R}$, the notations $\Ortho{U}$, $\Ortho{V}$, and $\Rel{D}$ will always denote, respectively, the orthogonal linear transformations and the diagonal linear relation given by the GSVD (Theorem~\ref{thm:gsvd-relations} and Definition~\ref{def:diagonalrelation}).
\end{remark}

\begin{remark} Definition~\ref{def:diagonalrelation} presents the diagonal linear relation in span form. Note that, in graphical syntax, changing to cospan form is the same as inverting the invertible linear transformations $\Inv{C}$ and $\Inv{S}$ and sliding wires:
\[
\begin{tabular}{r r}
  \TypedRel{D}{m}{n}=\input{GSVD_GSVD_D_Cospan.tikz} & \textbf{(Cospan)}
\end{tabular}
\]
\end{remark}

Theorem~\ref{thm:gsvd-relations} provides a GSVD for any linear relation $\Rel{R}$. Since, by Proposition~\ref{prop:span-cospan}, every linear relation can be described in terms of two linear transformations, the following question arises: is it always possible to turn this GSVD into the GSVD of a pair of linear transformations by separately and compatibly decomposing the constituent factors $\Map{A}$ and $\Map{B}$? The answer is yes, but it requires the introduction of a linear transformation $\Map{H}$, which connects the two factorizations, as proved in \cite[Section~3.1]{freitas2025imt}.

\subsection{Fundamental properties and subspaces with the GSVD}

The GSVD is powerful enough to characterize all fundamental properties and subspaces of a linear relation.
We show in the next two lemmas that
the vanishing of its middle wires determine the four fundamental properties,
while the orthogonal transformations produce the fundamental subspaces.
In both Lemma~\ref{lem:no_lolipop_GSVD} and~\ref{lem:subespaces_GSVD},
let \Rel{R} be decomposed as in Theorem~\ref{thm:gsvd-relations}, i.e.,
\[ \TypedRel{R}{m}{n} = \input{GSVD_GSVD.tikz}.\]

\begin{lemma}[Fundamental properties via GSVD]
  \label{lem:no_lolipop_GSVD}
  For every linear relation $\Rel{R}$, the vanishing of central wires of its GSVD determine the fundamental properties of $\Rel{R}$:
  \begin{center}
    \renewcommand{\arraystretch}{2.0}
    \begin{tabular}{lcr|lcr}
      \Rel{R} is INJ   & $\iff$ & $k_I = 0$
      & \Rel{R} is SUR & $\iff$ & $k_S = 0$  \\
      \Rel{R} is TOT   & $\iff$ & $k_T = 0$
      & \Rel{R} is DET & $\iff$ & $k_D = 0$
    \end{tabular}
  \end{center}
\end{lemma}
\begin{proof}
  We prove the TOT case. The other three cases are proven similarly.
  \begin{hcalculation}[\Leftrightarrow]{\scalebox{0.8}{\input{GSVD_Fund.Prop_tot.tikz}}}
    \hstep{\scalebox{0.8}{\input{GSVD_Fund.Prop_tot_step3.tikz}}}{Lem~\ref{lem:subespaces_GSVD}}
    \hstep{\scalebox{0.8}{\input{GSVD_Fund.Prop_tot_step4.tikz}}}{Multiply $\scalebox{0.7}{\input{GSVD_Fund.Prop_tot_v.tikz}}$}
    \hstep{\scalebox{0.8}{\input{GSVD_Fund.Prop_tot_step5.tikz}}}{$\scalebox{0.7}{\input{GSVD_Fund.Prop_tot_v.tikz}}$ is orthogonal}
    \hstep{\scalebox{0.8}{\input{GSVD_Fund.Prop_tot_step6.tikz}}}{Remove $\scalebox{0.7}{\Discard}$ from both sides}
    \hstep{k_T=0}{Fig.\ref{fig:Inequality}}
  \end{hcalculation}
\end{proof}

\begin{lemma}[Fundamental subspaces via GSVD]
  \label{lem:subespaces_GSVD} For every linear relation $\Rel{R}$, we have that:
  \begin{center}
    \input{GSVD_Subespaces_fund_sub_gsvd.tikz}
  \end{center}
\end{lemma}
\begin{proof}
  We prove the case $\scalebox{0.8}{\begin{tikzpicture}
	\begin{pgfonlayer}{nodelayer}
		\node [style=Subspace] (0) at (0, 0) {$\Ker(R)$};
		\node [style=none] (1) at (-2.25, 0) {};
	\end{pgfonlayer}
	\begin{pgfonlayer}{edgelayer}
		\draw (1.center) to (0);
	\end{pgfonlayer}
\end{tikzpicture}
}$. The other three cases are proven similarly.
  \begin{hcalculation}[=]{\scalebox{0.8}{\begin{tikzpicture}
	\begin{pgfonlayer}{nodelayer}
		\node [style=Subspace] (0) at (0, 0) {$\Ker(R)$};
		\node [style=none] (1) at (-2.25, 0) {};
	\end{pgfonlayer}
	\begin{pgfonlayer}{edgelayer}
		\draw (1.center) to (0);
	\end{pgfonlayer}
\end{tikzpicture}
}}
    \hstep{\scalebox{0.8}{\begin{tikzpicture}
	\begin{pgfonlayer}{nodelayer}
		\node [style=White] (0) at (2, 0) {};
		\node [style=Rel] (1) at (0, 0) {$R$};
		\node [style=none] (2) at (-1.75, 0) {};
	\end{pgfonlayer}
	\begin{pgfonlayer}{edgelayer}
		\draw (0) to (1);
		\draw (1) to (2.center);
	\end{pgfonlayer}
\end{tikzpicture}
}}{Def.\ref{def:fund_subspaces}}
    \hstep{\scalebox{0.8}{\input{GSVD_Subespaces_kernel_step1.tikz}}}{Thm.\ref{thm:gsvd-relations}: GSVD}
    \hstep{\scalebox{0.8}{\input{GSVD_Subespaces_kernel_step2.tikz}}}{$\scalebox{0.7}{\Ortho{U}}$ is orthogonal}
    \hstep{\scalebox{0.8}{\input{GSVD_Subespaces_kernel_step3.tikz}}}{Fig.\ref{fig:Frobenius_Algebra} + Fig.\ref{fig:Bialgebra} + \\ $C$ and $S$ invertible}
  \end{hcalculation}
\end{proof}

\section{Relational Least-Squares}\label{sec:LQ}

The least-squares method has been a very important tool for data fitting since its original formulation by \citet{legendre1805}, and its developments extend across all areas of mathematical science. The method solves inconsistent linear systems of the form $Ax = b$, where $A$ is a linear transformation with more rows than columns. This section generalizes this method to the context of linear relations.

For illustrative purposes, consider the following motivating example: a natural disaster is about to occur, and the government needs to make quick decisions to shelter the city's population. A diagram linking each person to the nearby shelters they have access to is displayed in Figure~\ref{fig:example-relation}.
Some people, of course, may not have access to any shelter, in which case a natural measure would be to associate them to those accessible by the closest person in their vicinity, as shown in Figure~\ref{fig:illustration-sur}.

Roughly speaking, the problem receives as input a finite relation between shelters and people and returns a \emph{reasonable} (which in mathematical terms would mean minimizing a certain norm) cofunction, i.e. an injective and surjective relations, from shelters to people. If the input relation were an injective function, the output would be a particular choice of left-inverse. We now make this discussion precise by defining operations which take a linear relation as input and return, respectively, something injective, surjective, total or deterministic.

\begin{remark} From here on we often use the compact notation $\underset{x}{\argmin}{\left[F(x) \;:\; P(x) \right]}$. This means taking the argmin with respect to the function $F(x)$ over all elements $x$ satisfying the property $P(x)$.
\end{remark}

\subsection{Least-Squares Operators}


\begin{figure}[pos=htb]
  \centering
  \begin{subfigure}[b]{0.48\textwidth}
    \centering
    \includegraphics[width=0.85\linewidth]{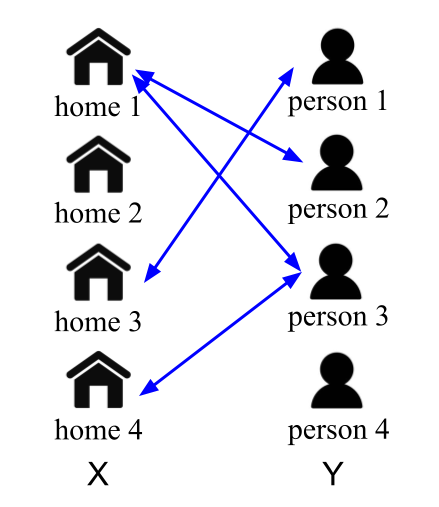}
    \caption{Example of finite relation between houses and people in a city.}
    \label{fig:example-relation}
  \end{subfigure}
  \hfill
  \begin{subfigure}[b]{0.48\textwidth}
    \centering

    \begin{subfigure}[b]{\linewidth}
      \centering
      \includegraphics[width=\linewidth]{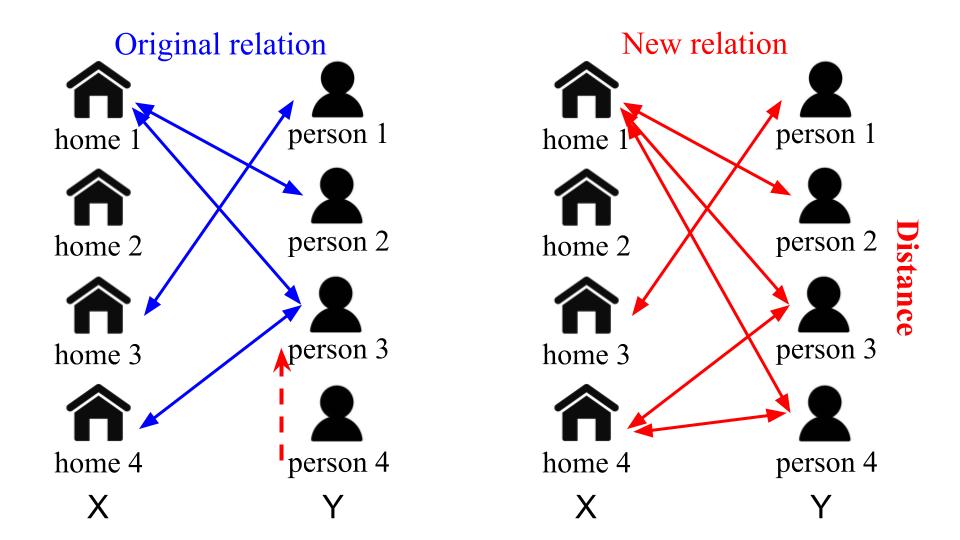}
      \caption{Example of new finite relation between houses and people in a city, ensuring that all people are sheltered.}
      \label{fig:illustration-sur}
    \end{subfigure}
    \vspace{1em}
    \begin{subfigure}[b]{\linewidth}
      \centering
      \includegraphics[width=\linewidth]{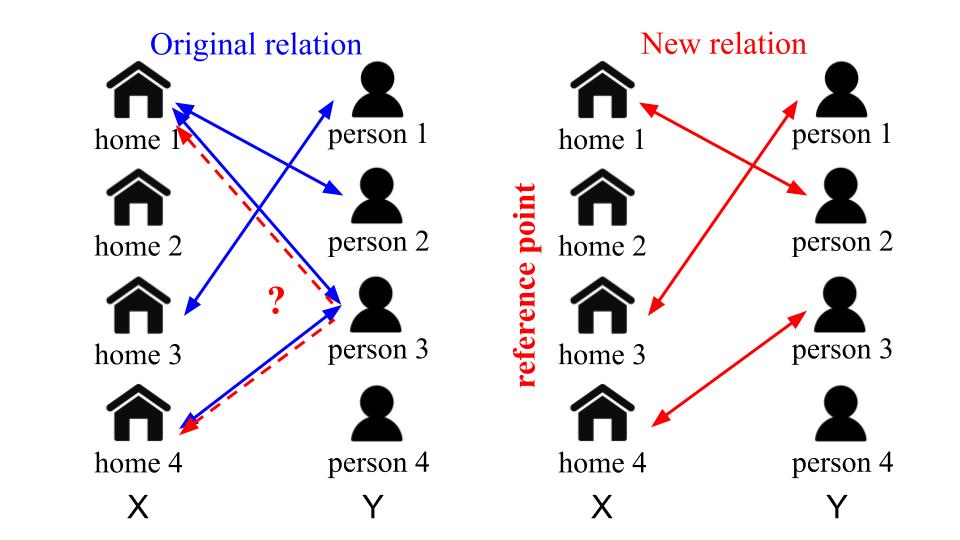}
      \caption{Example of the new finite relation between houses and people in a city, ensuring that each person is assigned to only one house.}
      \label{fig:illustration-inj}
    \end{subfigure}
  \end{subfigure}
  \caption{These examples illustrate how a linear relation can be turned injective or surjective.}
  \label{fig:illustration-properties}
\end{figure}

\begin{definition}[Generalized Least-Squares]\label{def:LQ-geometric}
  For every linear relation $R$, define
  \begin{enumerate}
    \item $\Sur(R) \coloneqq \left\{(x,\bar{y})\mid \exists y;\; xRy\;\;\text{and}\;\; y\in \underset{\hat{y}}{\argmin}\left [\left\| \bar{y}-\hat{y}\right\|_2^2\;:\;\hat{y}\in \Img(R) \right ] \right\}$,
    \item $\Inj(R) \coloneqq \left\{(\bar{x},y)\mid \exists x;\; xRy\;\;\text{and}\;\; \bar{x}\in \underset{\hat{x}}{\argmin}\left [ \left\|\hat{x}\right\|_2^2\;:\; \hat{x}Ry \right ] \right\}$,
    \item $\Tot(R) \coloneqq \left\{(\bar{x},y)\mid \exists x;\; xRy\;\;\text{and}\;\; x\in \underset{\hat{x}}{\argmin}\left [ \left\| \bar{x}-\hat{x}\right\|_2^2\;:\;\hat{x} \in \Dom(R) \right ] \right\}$,
    \item $\Det(R) \coloneqq \left\{(x,\bar{y})\mid \exists y;\; xRy\;\;\text{and}\;\; \bar{y}\in \underset{\hat{y}}{\argmin}\left [ \left\|\hat{y}\right\|_2^2\;:\; xR\hat{y}\right ] \right\}$.
  \end{enumerate}
\end{definition}

\begin{example}[Least-Squares with Equality Constraints~\citep{golub2013matrix}]\label{ex:LSE}
If we set \[\Rel{R}=\input{Connect_disconnect_lse.tikz}=\input{Connect_disconnect_lse-A1A2.tikz},\]
\noindent the relation represents the solutions to $A_1x = y$ that satisfy the equality constraint $A_2x=0$. The relation $\Sur(R)$, from Definition~\ref{def:LQ-geometric}, extends $R$ by associating $x$ with $\bar{y}$ whenever there exists a $y$ that satisfies the equality constraints and minimizes the Euclidean distance to $\bar{y}$, which is originally associated with $x$.
\end{example}

\begin{example}[Minimum-Norm Solution Restricted to a Subspace]
If we set \[\Rel{R}=\input{subspace_V-copy-A.tikz},\]
\noindent the relation represents the solutions to $Ax = y$ restricted to $x \in V$. The relation $\Inj(R)$ from Definition~\ref{def:LQ-geometric} selects only the $\bar{x}$ with minimum norm among all possible solutions while maintaining the restrictions.

\end{example}

\begin{example}[Orthogonal Projection onto an Affine Space]
If we set \[\Rel{R}=\input{subspace_V-mais-A.tikz},\]
\noindent the relation represents the solutions to the affine equation $Ax + b = y$, for all $b \in V$. The relation $\Sur(R)$ from Definition~\ref{def:LQ-geometric} expands $R$ by associating $x$ with $\bar{y}$ whenever the orthogonal projection of $\bar{y}$ into the affine subspace $\{\,Ax+b : b\in V\,\}$ is originally related to $x$.

\end{example}

Simply by looking at Definition~\ref{def:LQ-geometric}, it is not clear, for example, that $\Inj(R)$ is still a linear relation. This is a consequence of Theorem~\ref{thm:LQ-relational} below, which shows that these operations can be written as certain linear projections of $R$. For example, a linear relation $R$ can be canonically made injective by restricting it to ${\rm Ker}(R)^\perp$. Similarly, it can be made total by composing it with the projection onto ${\rm Dom}(R)$. Figure~\ref{fig:projection-example} presents a geometric visualization of the projections associated with surjectivity and injectivity.

\begin{figure}[pos=htb]
    \centering
    \begin{subfigure}[b]{0.48\textwidth}
        \centering
        \input{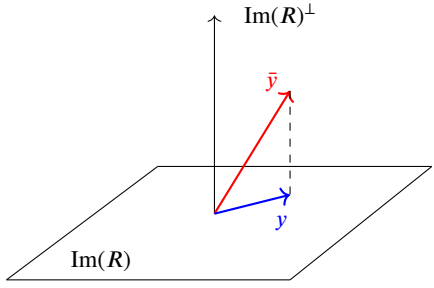}
        \caption{Given a linear relation $R \subseteq U \times V$ and a vector $\bar{y} \in V$, we project $\bar{y}$ onto $\Img(R)$ to obtain the nearest vector $y$ such that $xRy$.}
        \label{fig:projection-surjectivity}
    \end{subfigure}
    \hfill
    \begin{subfigure}[b]{0.48\textwidth}
        \centering
        \input{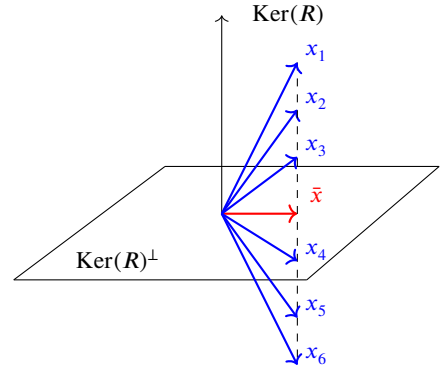}
        \caption{Given a linear relation $R \subseteq U \times V$ and a set of vectors $x_1, \dots, x_6 \in U$, we select the projection $\bar{x}$ with the smallest norm (obtained in $\Ker(R)^\perp$) such that $\bar{x}Ry$.}
        \label{fig:projection-injectivity}
    \end{subfigure}
    \caption{Geometric visualization of the projections associated with surjectivity (left) and injectivity (right).}
    \label{fig:projection-example}
\end{figure}


\begin{theorem}[Least-Squares as Projection]\label{thm:LQ-relational}
  For every linear relation $\Rel{R}$, the following statements hold:
  \begin{enumerate}
    \item \label{thm:LQ-relational:sur} $\input{Connect_disconnect_def_R_sur.tikz}$,
    \item \label{thm:LQ-relational:inj} $\input{Connect_disconnect_def_R_inj.tikz}$,
    \item \label{thm:LQ-relational:tot} $\input{Connect_disconnect_def_R_tot.tikz}$,
    \item \label{thm:LQ-relational:det} $\input{Connect_disconnect_def_R_det.tikz}$.
  \end{enumerate}
\end{theorem}

\begin{proof}
Appendix~\ref{ap:relationalLQ}.
\end{proof}

Although the relation operations of Definition~\ref{def:LQ-geometric} were formulated to satisfy each of the four fundamental properties, we can formalize these guarantees rigorously, proving them directly from Theorem~\ref{thm:LQ-relational}.
The following lemma establishes this result.

\begin{lemma}[Fundamental properties in Relational Least-Squares]\label{lem:prop_fund_Rbar}
  For every linear relation $\Rel{R}$, the following statements hold:
  \begin{multicols}{2}
    \begin{enumerate}
      \item \label{lem:prop_fund_Rbar:sur}
            $\Rel{\Sur(R)}$ is surjective,
      \item $\Rel{\Inj(R)}$ is injective,
      \item $\Rel{\Tot(R)}$ is total and
      \item $\Rel{\Det(R)}$ is deterministic.
    \end{enumerate}
  \end{multicols}
\end{lemma}
\begin{proof}
  Appendix~\ref{ap:relationalLQ}.
\end{proof}

\subsubsection{Compositions of Least-Squares Operators}

By composing the relation operations of Theorem~\ref{thm:LQ-relational}, it is possible to obtain new linear relations with specific properties. For example, the four possible compositions of the $\Tot$ operation take the form:
\begin{enumerate}
  \item  $\Rel{\Tot\left( \Sur(R) \right)} = \Comp{P_{\Dom(R)}/Map, R/Rel, P_{\Img(R)}/CoMap}$,
  \item  $\Rel{\Tot\left( \Inj(R) \right)} = \Comp{P_{\Dom(R)}/Map, P_{\Ker(R)}/CoMapGray, R/Rel}$,
  \item  $\Rel{\Tot\left( \Tot(R) \right)} = \Comp{P_{\Dom(R)}/Map, P_{\Dom(R)}/Map, R/Rel}$,
  \item  $\Rel{\Tot\left( \Det(R) \right)} = \Comp{P_{\Dom(R)}/Map, R/Rel, P_{\CoKer(R)}/MapGray}$.
\end{enumerate}
The properties of these linear relations obtained through the composition of the relation operations of Theorem~\ref{thm:LQ-relational} are described in Lemma~\ref{lem:comp-idempotent}.

\begin{lemma}[Idempotency and Commutativity]\label{lem:comp-idempotent}
  For every linear relation $\Rel{R}$, the following statements hold:
  \begin{center}
    \normalfont
    \begin{multicols}{2}
      \textbf{Idempotency}
      \begin{enumerate}[nolistsep]
        \item \label{lem:comp-idempotent:sur}
              $\Rel{\Sur\left( \Sur(R) \right)} = \Rel{\Sur(R)}$,
        \item $\Rel{\Inj\left( \Inj(R) \right)} = \Rel{\Inj(R)}$,
        \item $\Rel{\Tot\left( \Tot(R) \right)} = \Rel{\Tot(R)}$,
        \item $\Rel{\Det\left( \Det(R) \right)} = \Rel{\Det(R)}$.
      \end{enumerate}

      \textbf{Commutativity}
      \begin{enumerate}[nolistsep]
        \setcounter{enumi}{5}
        \item \label{lem:comp-commutative:sur}
              $\Rel{\Sur\left( \Inj(R) \right)} = \Rel{\Inj\left( \Sur(R) \right)}$,
        \item $\Rel{\Sur\left( \Tot(R) \right)} = \Rel{\Tot\left( \Sur(R) \right)}$,
        \item $\Rel{\Det\left( \Inj(R) \right)} = \Rel{\Inj\left( \Det(R) \right)}$,
        \item $\Rel{\Tot\left( \Det(R) \right)} = \Rel{\Det\left( \Tot(R) \right)}$.
      \end{enumerate}
    \end{multicols}
  \end{center}
\end{lemma}
\begin{proof}
  Appendix~\ref{ap:comp}.
\end{proof}

With the results of idempotency and commutativity established, we are now ready to state Theorem~\ref{lem:comp-fund-prop}, which synthesizes the fundamental properties of all compositions between pairs of the relational operations $\Sur$, $\Inj$, $\Tot$ and $\Det$.

\begin{theorem}[Linear Compositional Relations: Fundamental Properties]
  \label{lem:comp-fund-prop}
  For every linear relation $\Rel{R}$ and for any pair of relational operations
  \[
    F\in\left\{ \Sur,\Inj,\Tot,\Det\right\}\quad\text{and}\quad G\in\left\{ \Sur,\Inj,\Tot,\Det\right\},
  \]
  the composed $\Rel{F(G(R))}$ relation satisfies the fundamental property indicated at the intersection of the row labeled $F$ and the column labeled $G$ in the following table:
  \begin{center}
    \normalfont
    \begin{tabular}{r|cccc}
      \diagbox{F}{G} & $\Sur$                    & $\Inj$                       & $\Tot$                   & $\Det$                       \\ \hline
      $\Sur$                   & surjective               & \makecell{surjective and\\injective}    & \makecell{surjective\\and total}    & surjective                  \\ \hline
      $\Inj$                   & \makecell{surjective and\\injective} & injective                   & injective               & \makecell{injective and\\deterministic} \\ \hline
      $\Tot$                   & \makecell{surjective\\and total}     & total                       & total                   & \makecell{total and\\deterministic}     \\ \hline
      $\Det$                   & deterministic            & \makecell{injective and\\deterministic} & \makecell{total and\\deterministic} & deterministic
    \end{tabular}
  \end{center}
\end{theorem}
\begin{proof}
  Direct application of Lemmas~\ref{lem:prop_fund_Rbar} and~\ref{lem:comp-idempotent}.
\end{proof}

The commutative compositions established in Lemma~\ref{lem:comp-idempotent} exhibit, by Theorem~\ref{lem:comp-fund-prop}, particularly relevant properties. In particular, the opposite of the composition $\scalebox{0.8}{\Rel{\Sur\left ( \Inj(R)\right)}}$ coincides with the well-known \emph{pseudoinverse of relations}, as defined in~\cite{alvarez2014},
which will be further explored in Section~\ref{sec:Pseudoinverse}.
The remaining compositions, although still requiring a more in-depth investigation of their properties, also present themselves as mathematically promising operations, which we have not yet found in the literature. The composition $\scalebox{0.8}{\Rel{\Sur\left ( \Tot(R)\right)}}$ defines a surjective and total relation, that is, it ensures that all elements are related; $\scalebox{0.8}{\Rel{\Det\left ( \Inj(R)\right)}}$, being deterministic and total, results in a one-to-one relation; and $\scalebox{0.8}{\Rel{\Tot\left ( \Det(R)\right)}}$ produces a linear relation that is also a function, providing a functional approximation of the original linear relation.

\subsection{Relational Least-Squares via GSVD}\label{sec:algoritmos-leastsquares}

Besides the optimization and projection formulations presented above, the operations $\Sur$, $\Inj$, $\Tot$, and $\Det$ admit an additional description in terms of the GSVD, in which case, perhaps surprisingly, they correspond to just swapping the colors of the black and white dots.

\begin{theorem}[GSVD for Relational Least Squares]\label{thm:gsvd-compositional}
  Let $\Rel{R}$ be a linear relation whose GSVD is
  \[\TypedRel{R}{m}{n}=\input{GSVD_GSVD.tikz}. \]
  Then,
  \begin{enumerate}
    \item $\Rel{\Sur(R)}=\input{Connect_disconnect_GSVD_sur.tikz}$,
    \item $\Rel{\Inj(R)}=\input{Connect_disconnect_GSVD_inj.tikz}$,
    \item $\Rel{\Tot(R)}=\input{Connect_disconnect_GSVD_tot.tikz}$,
    \item $\Rel{\Det(R)}=\input{Connect_disconnect_GSVD_det.tikz}$.
  \end{enumerate}
\end{theorem}
\begin{proof}
  We will prove item $1.$ The remaining items are proved analogously.
  \begin{hcalculation}[=]{\scalebox{0.8}{\Rel{\Sur(R)}}}
    \hstep{\scalebox{0.8}{\input{Connect_disconnect_LQ-GSVD_step1.tikz}}}{Thm.\ref{thm:LQ-relational}}
    \hstep{\scalebox{0.8}{\input{Connect_disconnect_LQ-GSVD_step2.tikz}}}{Def.~\ref{def:proj-rel}}
    \hstep{\scalebox{0.8}{\input{Connect_disconnect_LQ-GSVD_step3.tikz}}}{GSVD + Lem.\ref{lem:subespaces_GSVD}}
    \hstep{\scalebox{0.8}{\input{Connect_disconnect_LQ-GSVD_step4.tikz}}}{$U$ is orthogonal + Fig.\ref{fig:Inequality}}
    \hstep{\scalebox{0.8}{\input{Connect_disconnect_LQ-GSVD_step5.tikz}}}{Fig.\ref{fig:Commutative_Comonoid} + Fig.\ref{fig:Bialgebra}}.
  \end{hcalculation}
  This completes the proof.
\end{proof}

It may be useful to think of the theorem above as the following: calculating the relational least-squares of a linear relation is equivalent to ``switching on or off'' the wire associated with that property in the GSVD, if we wish the relation to be surjective, we simply change the color of the node corresponding to the wire $k_S$.

\begin{corollary}
  Theorem~\ref{thm:gsvd-compositional} coincides, except for two orthogonal linear transformations, with the application of the operations $\Sur$, $\Inj$, $\Tot$, and $\Det$ to the GSVD of a diagonal linear relation. In other words, if we write $\Rel{D}$ for the relation in the middle of the GSVD of $\Rel{R}$, we get
  \begin{enumerate}
    \item $\Rel{\Sur(R)}=\input{Connect_disconnect_GSVD_sur_D_withD.tikz}$,
    \item $\Rel{\Inj(R)}=\input{Connect_disconnect_GSVD_inj_D_withD.tikz}$,
    \item $\Rel{\Tot(R)}=\input{Connect_disconnect_GSVD_tot_D_withD.tikz}$,
    \item $\Rel{\Det(R)}=\input{Connect_disconnect_GSVD_det_D_withD.tikz}$.
  \end{enumerate}
\end{corollary}

\subsection{Pseudoinverse}\label{sec:Pseudoinverse}

\begin{definition}[Pseudoinverse of Linear Relation]\label{def:pseudo-rel}
  For every linear relation $\TypedRel{R}{m}{n}$, its pseudoinverse is
  \[ \TypedRel{R^+}{m}{n} \coloneqq \TypedRel{\Sur\left ( \Inj(R) \right )}{m}{n}.\]
\end{definition}

\begin{remark}
  By Theorem~\ref{thm:gsvd-compositional}, $\TypedRel{R^+}{m}{n}$ is a cofunction, so we henceforth denote it as $\TypedCoMap{R^+}{m}{n}$.
\end{remark}

As a cofunction, $R^+$ can be interpreted as a relation that associates elements of the target set to elements of the starting set, and not the reverse. This can be interpreted geometrically through
Definition~\ref{def:LQ-geometric}, as illustrated in Figure~\ref{fig:visualize_projection_pseudo}. Moreover, it corresponds precisely to the pseudoinverse, originally formalized by~\citet{alvarez2014}, which generalizes the classic Moore-Penrose pseudoinverse~\citep{Penrose:1955} to the broader context of linear relations.

\begin{figure}[pos=htb]
    \centering
    \input{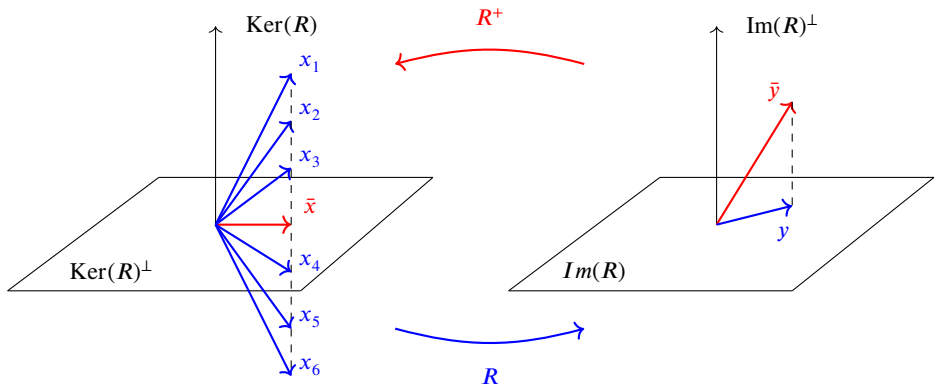}
    \caption{Geometric visualization of the pseudoinverse. Given a linear relation $R \subseteq U \times V$ and a vector $\bar{y} \in V$, we project $\bar{y}$ onto $\Img(R)$ to obtain the nearest point $xRy$. The vector $\bar{y}$ is then associated with the set of infinitely many least squares solutions $x$ of the relational equation $x R \bar{y}$. Among these solutions, the projection $\bar{x}$ with the smallest norm (obtained in $\Ker(R)^\perp$) is selected, characterizing the pseudoinverse $R^+$ such that $ \bar{y} R^+ \bar{x}$.}
    \label{fig:visualize_projection_pseudo}
\end{figure}

By opening the orthogonal projections in Definition~\ref{def:pseudo-rel},
we arrive at the following characterization of the pseudoinverse
in terms of the fundamental spaces of $\Rel{R}$.

\begin{lemma}\label{lem:pseudo-relation}
  For every linear relation $\Rel{R}$, its pseudoinverse can be written as
  \[\CoMap{R^+}=\input{MQ_Pseudo_pseudoinverse.tikz}.\]
\end{lemma}
\begin{proof}
  Appendix~\ref{ap:pseudoinverse}.
\end{proof}

Furthermore, for linear maps, Lemma~\ref{lem:pseudo-relation}
simplifies to an expression in terms of it and its transpose.

\begin{corollary}[Pseudoinverse of linear transformation]\label{lem:pseudoinverse}
  For every linear transformation $\Map{A}$, its pseudoinverse can be written as
  \[\CoMap{A^+}=\input{MQ_Pseudo_Pseudo-function_pseudo-function.tikz}.\]
\end{corollary}
\begin{proof}
  Appendix~\ref{ap:pseudoinverse}.
\end{proof}

According to \citet{BenIsrael:2003}, the classic formulation of the Moore-Penrose pseudoinverse is traditionally presented in terms of the Penrose equations, introduced by \citet{Penrose:1955}. The following theorem says that the relational definition of the pseudoinverse presented in Definition~\ref{def:pseudo-rel}, when applied to the particular case where $\Rel{R}$ is a linear transformation $\Map{A}$, satisfies the four Penrose equations.

\begin{theorem}[Moore-Penrose pseudoinverse]\label{thm:moore-penrose-pseudo}
For every linear transformation $\Map{A}$, its pseudoinverse $\CoMap{A^+}$
is the unique cofunction satisfying the following system of equations:
\begin{align}
  \begin{tikzpicture}
	\begin{pgfonlayer}{nodelayer}
		\node [style=none] (0) at (-2.5, 0) {};
		\node [style=Map] (1) at (-1, 0) {$A$};
		\node [style=Map] (2) at (1, 0) {$A^+$};
		\node [style=Map] (3) at (3, 0) {$A$};
		\node [style=none] (4) at (4.5, 0) {};
	\end{pgfonlayer}
	\begin{pgfonlayer}{edgelayer}
		\draw (0.center) to (4.center);
	\end{pgfonlayer}
\end{tikzpicture}
 &= \Map{A},\label{thm:moore-penrose-pseudo1}\\
  \begin{tikzpicture}
	\begin{pgfonlayer}{nodelayer}
		\node [style=none] (0) at (-2.5, 0) {};
		\node [style=Map] (1) at (-1, 0) {$A^+$};
		\node [style=Map] (2) at (1, 0) {$A$};
		\node [style=Map] (3) at (3, 0) {$A^+$};
		\node [style=none] (4) at (4.5, 0) {};
	\end{pgfonlayer}
	\begin{pgfonlayer}{edgelayer}
		\draw (0.center) to (4.center);
	\end{pgfonlayer}
\end{tikzpicture}
 &= \Map{A^+},\label{thm:moore-penrose-pseudo2}\\
  \begin{tikzpicture}
	\begin{pgfonlayer}{nodelayer}
		\node [style=none] (0) at (-2.5, 0) {};
		\node [style=CoMapGray] (1) at (-1, 0) {$A$};
		\node [style=CoMapGray] (2) at (1, 0) {$A^{+}$};
		\node [style=none] (3) at (2.5, 0) {};
	\end{pgfonlayer}
	\begin{pgfonlayer}{edgelayer}
		\draw (0.center) to (3.center);
	\end{pgfonlayer}
\end{tikzpicture}
 &= \begin{tikzpicture}
	\begin{pgfonlayer}{nodelayer}
		\node [style=none] (0) at (-2.5, 0) {};
		\node [style=Map] (1) at (-1, 0) {$A^{+}$};
		\node [style=Map] (2) at (1, 0) {$A$};
		\node [style=none] (3) at (2.5, 0) {};
	\end{pgfonlayer}
	\begin{pgfonlayer}{edgelayer}
		\draw (0.center) to (3.center);
	\end{pgfonlayer}
\end{tikzpicture}
,\label{thm:moore-penrose-pseudo3}\\
  \begin{tikzpicture}
	\begin{pgfonlayer}{nodelayer}
		\node [style=none] (0) at (-2.5, 0) {};
		\node [style=CoMapGray] (1) at (-1, 0) {$A^{+}$};
		\node [style=CoMapGray] (2) at (1, 0) {$A$};
		\node [style=none] (3) at (2.5, 0) {};
	\end{pgfonlayer}
	\begin{pgfonlayer}{edgelayer}
		\draw (0.center) to (3.center);
	\end{pgfonlayer}
\end{tikzpicture}
 &= \begin{tikzpicture}
	\begin{pgfonlayer}{nodelayer}
		\node [style=none] (0) at (-2.5, 0) {};
		\node [style=Map] (1) at (-1, 0) {$A$};
		\node [style=Map] (2) at (1, 0) {$A^{+}$};
		\node [style=none] (3) at (2.5, 0) {};
	\end{pgfonlayer}
	\begin{pgfonlayer}{edgelayer}
		\draw (0.center) to (3.center);
	\end{pgfonlayer}
\end{tikzpicture}
.\label{thm:moore-penrose-pseudo4}
\end{align}
\end{theorem}
\begin{proof}
  Appendix~\ref{ap:pseudoinverse}.
\end{proof}

\subsubsection{Pseudoinverse via GSVD}

The composition of the injectivity and surjectivity operations, which defines the pseudoinverse $\scalebox{0.8}{\CoMap{R^+}}$, not only extends the classical Moore–Penrose theory to the context of linear relations but also enables the derivation of algorithms. We now present a formulation of the pseudoinverse of relations in terms of the GSVD.

In the graphical syntax, the procedure essentially consists of turning on and off the ``wires'' corresponding to the injectivity and surjectivity properties in the GSVD of $\scalebox{0.8}{\Rel{R}}$.

\begin{theorem} \label{thm:GSVD-Pseudoinverse}
  Let $\TypedRel{R}{m}{n}$ be a linear relation with GSVD as in Theorem~\ref{thm:gsvd-compositional}.
  Then its pseudoinverse is
\begin{equation}
  \TypedCoMap{R^+}{m}{n} = \input{GSVD_GSVD_pseudoinversa2.tikz}.
\end{equation}
\end{theorem}
\begin{proof}
  Theorem~\ref{thm:gsvd-compositional} for $\Inj$ and $\Sur$.
\end{proof}

For a diagonal linear relation, Theorem~\ref{thm:GSVD-Pseudoinverse} maintains the central structure, without the orthogonal linear transformations.

\begin{corollary}[Pseudoinverse of a Diagonal Linear Relation]\label{cor:GSVD_diagonal_relation}
  The pseudoinverse of a diagonal linear relation
  \[\TypedRel{D}{m}{n}=\input{GSVD_GSVD_D.tikz},\quad\quad \text{is} \quad\quad \TypedCoMap{D^+}{m}{n}=\input{GSVD_GSVD_pseudoinversa.tikz}.\]
\end{corollary}

\begin{remark}\label{rem:pseudo}
  By Corollary~\ref{cor:GSVD_diagonal_relation}, we can rewrite Theorem~\ref{thm:GSVD-Pseudoinverse} in the simplified form
\[\CoMap{R^+}=\input{GSVD_GSVD_pseudoinversa_withD2.tikz}.\]
\end{remark}

In the particular case in which the linear relation $\scalebox{0.8}{\Rel{R}}$ is a linear transformation $\scalebox{0.8}{\Map{A}}$, it is still possible to employ the GSVD as in Theorem~\ref{thm:gsvd-relations}. To this end, it suffices to consider the linear transformation $\scalebox{0.8}{\Map{B}}=\scalebox{0.8}{\TypedId{m}}$ in the representation of the linear relation in span form (Proposition~\ref{prop:span-cospan}), in which case the GSVD coincides with the SVD~\citep{golub2013matrix,trefethen97}. In this way, we recover the classical solution of the Moore–Penrose pseudoinverse, introduced by~\cite{Penrose:SVD-pseudo}, as established in Theorem~\ref{thm:SVDpseudoinverse}.

\begin{theorem}[SVD for Pseudoinverse]\label{thm:SVDpseudoinverse} For every linear transformation $\TypedMap{A}{m}{n}$, there exist orthogonal linear transformations $\TypedOrtho{U}{n}{n}$ and $\TypedOrtho{V}{m}{m}$ and a linear transformation with nonnegative diagonal entries $\TypedMap{\Sigma}{m}{n}$ such that, if
\[\Map{A}=\input{GSVD_SVD_withD.tikz} \quad\quad \text{then} \quad\quad \CoMap{A^+}=\input{GSVD_SVD_pseudoinversa_withD.tikz},\]
where $\CoMap{\Sigma^+}\coloneqq\CoMap{\Sur\left ( \Inj(\Sigma) \right )}.$
\end{theorem}
\begin{proof}
Appendix~\ref{ap:algopseudoinverse}
\end{proof}

\subsubsection{Optimization problems}\label{ap:optimizationpseudo}

Now, we show how two optimization problems presented in the article \emph{The Restricted Singular Value Decomposition: Properties and Applications} \citep{Golub:1991} can be reformulated and reinterpreted when expressed in the semantics of linear relations.

The original paper uses the Restricted Singular Value Decomposition (RSVD) to analyze a series of applications involving the decomposition of a linear transformation $A$ relative to two other linear transformations $B$ and $C$, exploring how restrictions on column and row spaces influence decompositions and geometric relations between kernels and images. The results obtained in the context of the RSVD generally depend on the complete structure of the triple $(A,B,C)$ and the interaction between the generated subspaces. From the perspective of linear relations, we can view the triple of linear transformations $(A,B,C)$ as a decomposition of linear relations $\scalebox{0.8}{\Comp{C/Map, A/CoMap, B/Map}}$.
This is similar to the span and cospan forms in Proposition~\ref{prop:span-cospan}.
A proof that every linear relation may by decomposed this way is left to Proposition~\ref{prop:tripla}.

As proven in~\cite{Golub:1991}, by setting $\scalebox{0.8}{\Map{B}}=\scalebox{0.8}{\Id}$ or $\scalebox{0.8}{\Map{C}}=\scalebox{0.8}{\Id}$, the RSVD reduces to the QSVD --- a decomposition equivalent to the GSVD.
In this paper, we revisit some of the problems presented by the authors for the case $\scalebox{0.8}{\Map{C}}=\scalebox{0.8}{\Id}$.

\begin{remark}
To provide context for the reader and facilitate comparison between version,
Appendix~\ref{ap:functionalizer}
contains all relevant theorems of~\cite{Golub:1991} ``as-is'', i.e.,
in terms of triplet of matrices.
Throughout this section we prove equivalent results using linear relations instead.
\end{remark}

We begin by~\cite[Thm.~4]{Golub:1991}, transcribed as Theorem~\ref{thm:teo4golub}.
In that theorem, the authors describe the conditions for the linear transformation $\scalebox{0.8}{\Map{S}}$ that minimizes $\rank\left(\scalebox{0.6}{\input{Golub_a-dc.tikz}}\right)$ to be unique.
We present the same result in Theorem~\ref{thm:teo4-nosso},
but first prove a useful lemma.

\begin{lemma}[Image of \scalebox{0.7}{\input{Golub_r-d.tikz}}]\label{lem:rank_R-S} For every linear relation $\TypedRel{R}{m}{n} = \TypedSpan{A}{B}{m}{}{n}$ and every cofunction $\TypedCoMap{S}{m}{n}$ ,
\[\input{Golub_im_r-d.tikz}=\input{Golub_im_a-dc.tikz}.\]
\end{lemma}
\begin{proof}
\begin{hcalculation}[=]{\scalebox{0.8}{\input{Golub_im_r-d.tikz}}}
\hstep{\scalebox{0.8}{\input{Golub_im_ca-d.tikz}}}{$\scalebox{0.7}{\Rel{R}}=\scalebox{0.7}{\Span{A}{B}}$}
\hstep[=]{\scalebox{0.8}{\input{Golub_im_a-dc.tikz}}}{Fig.\ref{fig:Inequality} + Fig.\ref{fig:typerelations}}
\end{hcalculation}
\end{proof}

\begin{theorem}[Uniqueness of $\scalebox{0.7}{\CoMap{S}}$ that minimizes $\rank\left(\scalebox{0.7}{\input{Golub_r-d.tikz}}\right)$]\label{thm:teo4-nosso}
For every linear relation $\TypedRel{R}{m}{n}$,
\[\left\{\TypedCoMap{R^+}{m}{n}\right\} = \underset{\scalebox{0.7}{\TypedCoMap{S}{m}{n}}}{\argmin} \left [\rank\left(\scalebox{0.9}{\input{Golub_r-d.tikz}}\right)\right ]\]
\noindent if and only if $\Rel{R}$ is injective and surjective.
\end{theorem}
\begin{proof}
Let $\scalebox{0.8}{\Rel{R}} = \scalebox{0.8}{\Span{A}{B}}$.
\begin{hcalculation}[\Leftrightarrow ]{\scalebox{0.8}{\Rel{R}}\; \text{is injective and surjective}}
\hstep{\left\{\scalebox{0.8}{\input{GSVD_GSVD_pseudoinversa_withD2.tikz}}\right\} = \underset{\scalebox{0.7}{\CoMap{S}}}{\argmin}\left[\rank\left(\scalebox{0.8}{\input{Golub_a-dc.tikz}}\right)\right]}{Cor.\ref{cor:golub4-nosso}}
\hstep{\left\{\scalebox{0.8}{\input{GSVD_GSVD_pseudoinversa_withD2.tikz}}\right\} = \underset{\scalebox{0.7}{\CoMap{S}}}{\argmin}\left[\rank\left(\scalebox{0.8}{\input{Golub_r-d.tikz}}\right)\right]}{Lem.\ref{lem:rank_R-S}}
\hstep{\left\{\scalebox{0.8}{\CoMap{R^+}}\right\} = \underset{\scalebox{0.7}{\CoMap{S}}}{\argmin}\left[\rank\left(\scalebox{0.8}{\input{Golub_r-d.tikz}}\right)\right]}{Thm.\ref{thm:GSVD-Pseudoinverse}: Pseudoinverse}
\end{hcalculation}
\end{proof}

In~\cite[Thm.~7]{Golub:1991}, transcribed as Theorem~\ref{thm:teo7},
the authors seek to characterize the minimum norm solution of the matrix equation $\scalebox{0.8}{\begin{tikzpicture}
	\begin{pgfonlayer}{nodelayer}
		\node [style=Map] (0) at (-1, 0) {$B$};
		\node [style=Map] (1) at (1, 0) {$S$};
		\node [style=none] (2) at (-2.5, 0) {};
		\node [style=none] (3) at (2.5, 0) {};
	\end{pgfonlayer}
	\begin{pgfonlayer}{edgelayer}
		\draw (2.center) to (3.center);
	\end{pgfonlayer}
\end{tikzpicture}
}=\scalebox{0.8}{\Map{A}}$ with respect to the unknown linear transformation $\scalebox{0.8}{\Map{S}}$. They also note that this equation has particular historical relevance, as it led Penrose to the formulation of the pseudoinverse that now bears his name.
We reformulate these statements into Theorems~\ref{thm:teo7-nosso} and~\ref{thm:teo7-nosso2}, both in terms of linear relations.
Lemmas~\ref{lem:thm7} and~\ref{lem:thm7.2} provide the necessary tools for that.

\begin{lemma}[Conditions for the existence of an upper limit for a linear relation]\label{lem:thm7} Let $\Rel{R}=\Span{A}{B}$ be a linear relation. The following statements are equivalent:
\begin{enumerate}
\item $\Rel{R}$ is injective;
\item $\exists\CoMap{S}$ such that $\Rel{R}\subseteq \CoMap{S}$;
\item $\Map{A}=\begin{tikzpicture}
	\begin{pgfonlayer}{nodelayer}
		\node [style=Map] (0) at (-1, 0) {$B$};
		\node [style=Map] (1) at (1, 0) {$S$};
		\node [style=none] (2) at (-2.5, 0) {};
		\node [style=none] (3) at (2.5, 0) {};
	\end{pgfonlayer}
	\begin{pgfonlayer}{edgelayer}
		\draw (2.center) to (3.center);
	\end{pgfonlayer}
\end{tikzpicture}
$;
\item $\input{Golub_im_r-d.tikz}=\Zero$.
\end{enumerate}
\end{lemma}
\begin{proof}($1\implies 2$)

  \begin{hcalculation}[\Leftrightarrow ]{\scalebox{0.8}{\Rel{R}}\; \text{is injective}}
    \hstep{\scalebox{0.8}{\Rel{R}}=\scalebox{0.8}{\input{Golub_R_inj.tikz}}}{Lem.\ref{lem:no_lolipop_GSVD}}
    \hstep[\Rightarrow ]{\scalebox{0.8}{\Rel{R}}\subseteq\scalebox{0.8}{\input{Golub_R_inj_black.tikz}}}{Fig.\ref{fig:min_max}}
    \hstep{\scalebox{0.8}{\Rel{R}}\subseteq \scalebox{0.8}{\CoMap{S}}}{$\scalebox{0.7}{\CoMap{S}}\coloneqq\scalebox{0.7}{\input{Golub_R_inj_black.tikz}}$}
  \end{hcalculation}

($2\implies 1$)

\begin{hcalculation}[\subseteq ]{\scalebox{0.8}{\begin{tikzpicture}
	\begin{pgfonlayer}{nodelayer}
		\node [style=Rel] (0) at (0, 0) {$R$};
		\node [style=White] (1) at (1.5, 0) {};
		\node [style=none] (2) at (-1.5, 0) {};
	\end{pgfonlayer}
	\begin{pgfonlayer}{edgelayer}
		\draw (2.center) to (1);
	\end{pgfonlayer}
\end{tikzpicture}
}}
  \hstep{\scalebox{0.8}{\begin{tikzpicture}
	\begin{pgfonlayer}{nodelayer}
		\node [style=CoMap] (0) at (0, 0) {$S$};
		\node [style=White] (1) at (1.5, 0) {};
		\node [style=none] (2) at (-1.5, 0) {};
	\end{pgfonlayer}
	\begin{pgfonlayer}{edgelayer}
		\draw (2.center) to (1);
	\end{pgfonlayer}
\end{tikzpicture}
}}{Hypothesis}
  \hstep{\scalebox{0.8}{\CoZero}}{Fig.~\ref{fig:typerelations}}
\end{hcalculation}

($2 \iff 3$)

\begin{hcalculation}[\Leftrightarrow]{\scalebox{0.8}{\Rel{R}}\subseteq \scalebox{0.8}{\CoMap{S}}}
  \hstep{\scalebox{0.8}{\Span{A}{B}}\subseteq\scalebox{0.8}{\CoMap{S}}}{$\scalebox{0.7}{\Rel{R}}=\scalebox{0.7}{\Span{A}{B}}$}
  \hstep{\scalebox{0.8}{\Map{A}}\subseteq\scalebox{0.8}{\begin{tikzpicture}
	\begin{pgfonlayer}{nodelayer}
		\node [style=Map] (0) at (-1, 0) {$B$};
		\node [style=Map] (1) at (1, 0) {$S$};
		\node [style=none] (2) at (-2.5, 0) {};
		\node [style=none] (3) at (2.5, 0) {};
	\end{pgfonlayer}
	\begin{pgfonlayer}{edgelayer}
		\draw (2.center) to (3.center);
	\end{pgfonlayer}
\end{tikzpicture}
}}{Lem.\ref{lem:pa_pum_igual}--item $1$}
  \hstep{\scalebox{0.8}{\Map{A}}=\scalebox{0.8}{\begin{tikzpicture}
	\begin{pgfonlayer}{nodelayer}
		\node [style=Map] (0) at (-1, 0) {$B$};
		\node [style=Map] (1) at (1, 0) {$S$};
		\node [style=none] (2) at (-2.5, 0) {};
		\node [style=none] (3) at (2.5, 0) {};
	\end{pgfonlayer}
	\begin{pgfonlayer}{edgelayer}
		\draw (2.center) to (3.center);
	\end{pgfonlayer}
\end{tikzpicture}
}}{Lem.\ref{lem:pair_of_matrices_1}}
\end{hcalculation}

($3 \iff 4$)

\begin{hcalculation}[\Leftrightarrow ]{\scalebox{0.8}{\Map{A}}=\scalebox{0.8}{\begin{tikzpicture}
	\begin{pgfonlayer}{nodelayer}
		\node [style=Map] (0) at (-1, 0) {$B$};
		\node [style=Map] (1) at (1, 0) {$S$};
		\node [style=none] (2) at (-2.5, 0) {};
		\node [style=none] (3) at (2.5, 0) {};
	\end{pgfonlayer}
	\begin{pgfonlayer}{edgelayer}
		\draw (2.center) to (3.center);
	\end{pgfonlayer}
\end{tikzpicture}
}}
  \hstep{\scalebox{0.8}{\input{Golub_a-dc.tikz}} = \scalebox{0.8}{\Discard} \; \scalebox{0.8}{\Zero}}{Algebra}
  \hstep{\scalebox{0.8}{\input{Golub_im_a-dc.tikz}} = \scalebox{0.8}{\Zero}}{Fig.\ref{fig:min_max}}
\end{hcalculation}

\end{proof}

\begin{lemma}[Conditions for the existence of a lower limit for a linear relation]\label{lem:thm7.2} Let $\Rel{R}=\Span{C}{A}$ be a linear relation. The following statements are equivalent:
\begin{enumerate}
\item $\Rel{R}$ is surjective;
\item $\exists\CoMap{S}$ such that $\CoMap{S}\subseteq \Rel{R}$;
\item $\Map{A}=\begin{tikzpicture}
	\begin{pgfonlayer}{nodelayer}
		\node [style=Map] (0) at (-1, 0) {$S$};
		\node [style=Map] (1) at (1, 0) {$C$};
		\node [style=none] (2) at (-2.5, 0) {};
		\node [style=none] (3) at (2.5, 0) {};
	\end{pgfonlayer}
	\begin{pgfonlayer}{edgelayer}
		\draw (2.center) to (3.center);
	\end{pgfonlayer}
\end{tikzpicture}
$;
\item $\input{Golub_ker_r-d.tikz}=\Discard$.
\end{enumerate}
\end{lemma}
\begin{proof}
Analogous to the proof of Lemma~\ref{lem:thm7}.
\end{proof}

\begin{theorem}[Upper limit with least norm of a linear relation]\label{thm:teo7-nosso}
For every injective linear relation $\TypedRel{R}{m}{n}$,
\[\TypedCoMap{R^+}{m}{n}\in\underset{\scalebox{0.7}{\TypedCoMap{S}{m}{n}}}{\argmin}\left[\left\|\scalebox{0.9}{\Map{S}} \right\|_F\;:\;\scalebox{0.9}{\Rel{R}}\subseteq \scalebox{0.9}{\CoMap{S}}\right].\]
\end{theorem}
\begin{proof}
Let $\scalebox{0.8}{\Rel{R}} = \scalebox{0.8}{\Span{A}{B}}$.
\begin{hcalculation}[\Leftrightarrow ]{\scalebox{0.8}{\Rel{R}}\; \text{is injective}}
\hstep{\scalebox{0.8}{\input{GSVD_GSVD_pseudoinversa_withD2.tikz}}  \in \underset{\scalebox{0.7}{\CoMap{S}}}{\argmin}\left[\left\|\scalebox{0.8}{\Map{S}} \right\|_F\;:\;\scalebox{0.8}{\Map{A}}= \scalebox{0.8}{\begin{tikzpicture}
	\begin{pgfonlayer}{nodelayer}
		\node [style=Map] (0) at (-1, 0) {$B$};
		\node [style=Map] (1) at (1, 0) {$S$};
		\node [style=none] (2) at (-2.5, 0) {};
		\node [style=none] (3) at (2.5, 0) {};
	\end{pgfonlayer}
	\begin{pgfonlayer}{edgelayer}
		\draw (2.center) to (3.center);
	\end{pgfonlayer}
\end{tikzpicture}
} \right]}{Cor.\ref{cor:golub7-nosso}}
\hstep{\scalebox{0.8}{\input{GSVD_GSVD_pseudoinversa_withD2.tikz}}  \in \underset{\scalebox{0.7}{\CoMap{S}}}{\argmin}\left[\left\|\scalebox{0.8}{\Map{S}} \right\|_F\;:\;\scalebox{0.8}{\Rel{R}}\subseteq \scalebox{0.8}{\CoMap{S}} \right]}{Lem.\ref{lem:thm7}}
\hstep{\scalebox{0.8}{\CoMap{R^+}}  \in \underset{\scalebox{0.7}{\CoMap{S}}}{\argmin}\left[\left\|\scalebox{0.8}{\Map{S}} \right\|_F\;:\;\scalebox{0.8}{\Rel{R}}\subseteq \scalebox{0.8}{\CoMap{S}} \right]}{Thm.\ref{thm:GSVD-Pseudoinverse}: Pseudoinverse}
\end{hcalculation}
\end{proof}

\begin{theorem}[Lower limit with least norm of a linear relation]\label{thm:teo7-nosso2}
For every surjective linear relation $\TypedRel{R}{m}{n}$,
\[\TypedCoMap{R^+}{m}{n}\in\underset{\scalebox{0.7}{\TypedCoMap{S}{m}{n}}}{\argmin}\left[\left\|\scalebox{0.9}{\Map{S}} \right\|_F\;:\;\scalebox{0.9}{\CoMap{S}}\subseteq \scalebox{0.9}{\Rel{R}}\right].\]
\end{theorem}
\begin{proof}
Analogous to the proof of Theorem~\ref{thm:teo7-nosso}.
\end{proof}

The results of this section establish an interesting connection between the pseudoinverse of linear relations and optimization problems. In particular, Theorems~\ref{thm:teo4-nosso},~\ref{thm:teo7-nosso} and~\ref{thm:teo7-nosso2} characterize solutions according to two distinct criteria: minimizing the rank of the difference between relations and minimizing the Frobenius norm of a cofunction that approximates a given linear relation.

In the next section, we revisit these criteria by exploring a relational version of the low-rank approximation problem. This formulation brings together the rank constraint and the minimization of the Frobenius norm within a unified optimization framework, extending classical results, such as the Eckart–Young theorem, to the context of linear relations.

\section{Relational Low-Rank Approximation}\label{sec:LR}

In Linear Algebra, low-rank approximation consists of an optimization problem that seeks to represent a linear transformation in the best possible way by another of lower rank. This approach has numerous applications, as it enables compact representations with controlled information loss. Applications include image processing and compression, data mining and analysis, noise reduction, regularization of ill-posed problems, and statistical modeling, among many others \citep{lowrank1,lowrank2,lowrank3,lowrank4,lowrank5}.
This chapter presents a generalization of low-rank approximation to linear relations.

Throughout this work we generalize the notion of ``rank''
to linear relations according to the following definition.

\begin{definition}[Rank of a Linear Relation]\label{def:rank-linear-relation}
  For every linear relation $\Rel{R}$, we define its rank as
  \[
    \rank\left( \scalebox{0.7}{\Rel{R}} \right)
    \coloneqq \dim\left( \scalebox{0.7}{\Subspace{\Img(R)}} \right)
    = \dim\left( \scalebox{0.7}{\Diagram{Rel}{R}{Black}{}} \right).
  \]
\end{definition}

\subsection{Generalizing Low-Rank Approximation for Linear Relations}\label{sec:g-lowrank}

We begin by using the GSVD to define a low-rank cofunction that approximates a linear relation.

\begin{definition}[Relational Low-Rank Approximation]\label{def:lowrank-gsvd}
For every linear relation  \(\Rel{R}=\input{GSVD_GSVD_withD.tikz}\) with rank $r$ and for any \(\ell \le r\), we define the cofunctions
\[\CoMap{R_{\overline{\ell}}^{+}} \coloneqq \input{GSVD_GSVD_pseudoinversa_withD_ell.tikz}\qquad\text{and}\qquad
\CoMap{R_{\ell}^{+}} \coloneqq \input{GSVD_GSVD_pseudoinversa_withD_ell2.tikz},\]
\noindent where
\begin{itemize}
    \item \(\CoMap{D_{\overline{\ell}}^{+}}\) is the cofunction containing the $\ell$ \textbf{smallest} elements of the diagonal of $\CoMap{D^+}$;
    \item \(\CoMap{D_{\ell}^{+}}\) as the cofunction containing the $\ell$ \textbf{largest} elements of the diagonal of $\CoMap{D^+}$.
\end{itemize}
\end{definition}

\begin{remark}\label{rmk:simple_identities_on_truncations}
  For any linear relation $\Rel{R}$ with rank $r$,
  \begin{align*}
    \CoMap{R_{\overline{r}}^{+}}&=\CoMap{R_{r}^{+}}=\CoMap{R^{+}}, \\
    \CoMap{R^+_{\overline{r - k}}} &= \input{GSVD_R_mais-R_k.tikz}.
  \end{align*}
\end{remark}

Theorem~\ref{thm:Eckart-Young-gsvd} states that, for every linear relation $\scalebox{0.8}{\Rel{R}}$ of rank $r$, $\scalebox{0.8}{\CoMap{R_{\overline{\ell}}^{+}}}$, as defined in Definition~\ref{def:lowrank-gsvd}, is the minimum norm cofunction with rank $\ell$ that approximates $\scalebox{0.8}{\Rel{R}}$. A similar result was proved in~\cite[Thm.~6]{Golub:1991} (stated in Theorem~\ref{thm:teo6}), although in a different context: given three linear transformations $\scalebox{0.8}{\Map{A}}$, $\scalebox{0.8}{\Map{B}}$, and $\scalebox{0.8}{\Map{C}}$, the authors characterize the minimum norm function $\scalebox{0.8}{\Map{S}}$ that satisfies a given value for $\rank\left(\scalebox{0.6}{\input{Golub_a-bdc.tikz}}\right)$.

\begin{theorem}[Relational Low-Rank Approximation with Least Norm]\label{thm:Eckart-Young-gsvd}
  Let $\TypedCoMap{R_{\overline{\ell}}^{+}}{m}{n}$ be defined as in Definition~\ref{def:lowrank-gsvd}. For every linear relation $\TypedRel{R}{m}{n}$ of rank $r$ and for every linear transformation $\TypedMap{S}{n}{m}$ satisfying \[k_I\leq \rank\left(\input{Golub_r-d.tikz}\right)< k_I +d,\] we have:
\[\TypedCoMap{R_{\overline{\ell}}^{+}}{m}{n}\in\underset{\scalebox{0.7}{\TypedCoMap{S}{m}{n}}}{\argmin}\left[\left\|\scalebox{0.9}{\Map{S}} \right\|_F\;:\;\rank\left(\scalebox{0.9}{\input{Golub_r-d.tikz}}\right)=r-\ell\right].\]
\end{theorem}
\begin{proof}
\begin{hcalculation}[\Leftrightarrow ]{\scalebox{0.8}{\Rel{R}}=\scalebox{0.8}{\Span{A}{B}}}
\hstep{\scalebox{0.7}{\input{GSVD_GSVD_pseudoinversa_withD_r-k.tikz}}  \in \underset{\scalebox{0.7}{\CoMap{S}}}{\argmin}\left[\left\|\scalebox{0.7}{\Map{S}} \right\|_F\;:\;\rank\left(\scalebox{0.7}{\input{Golub_a-dc.tikz}}\right)=k\right]}{Cor.\ref{cor:golub6-nosso}}
\hstep{\scalebox{0.7}{\input{GSVD_GSVD_pseudoinversa_withD_r-k.tikz}}  \in \underset{\scalebox{0.7}{\CoMap{S}}}{\argmin}\left[\left\|\scalebox{0.7}{\Map{S}} \right\|_F\;:\;\rank\left(\scalebox{0.7}{\input{Golub_r-d.tikz}}\right)=k\right]}{Lem.\ref{lem:rank_R-S}}
\hstep{\scalebox{0.7}{\input{GSVD_GSVD_pseudoinversa_withD_ell3.tikz}}  \in \underset{\scalebox{0.7}{\CoMap{S}}}{\argmin}\left[\left\|\scalebox{0.7}{\Map{S}} \right\|_F\;:\;\rank\left(\scalebox{0.7}{\input{Golub_r-d.tikz}}\right)=r-\ell\right]}{\makecell{Variable change:\\$\ell\coloneqq r-k$}}
\hstep{\scalebox{0.7}{\CoMap{R_{\overline{\ell}}^{+}}}  \in \underset{\scalebox{0.7}{\CoMap{S}}}{\argmin}\left[\left\|\scalebox{0.7}{\Map{S}} \right\|_F\;:\;\rank\left(\scalebox{0.7}{\input{Golub_r-d.tikz}}\right)=r-\ell\right]}{\makecell{Def.\ref{def:lowrank-gsvd}: Relational\\Low-Rank}}
\end{hcalculation}
\end{proof}

\subsection{Instances}\label{sec:lowrank-instances}

From the formulation presented in Theorem~\ref{thm:Eckart-Young-gsvd}, it is possible to retrieve a series of results from the literature through instances of the linear relation $\scalebox{0.8}{\Rel{R}}$ and the number $\ell$. We begin with a more familiar instance, namely, when the linear relation reduces to a linear transformation. This case serves as a consistency check for the relational formulation and provides a direct bridge to classical theory, since the result obtained corresponds to the Eckart-Young theorem~\citep{EckartYoung1936,YoungHouseholder1936}.

\begin{theorem}[Eckart-Young Theorem via Relational Low-Rank Approximation]
  \label{thm:eckart-young}
  Let $\TypedMap{A}{m}{n}$ be a linear transformation with rank $r$. Then, for every $0 \leq k < r$ we have
\[\TypedCoMap{A_k}{m}{n}\in\underset{\scalebox{0.7}{\TypedCoMap{S}{m}{n}}}{\argmin}\left[\left\|\scalebox{0.9}{\input{Golub_a-s.tikz}} \right\|_F\;:\;\rank\left(\scalebox{0.9}{\Map{S}}\right)=k\right].\]
\end{theorem}

\begin{proof} Let $k=r-\ell$. We instantiate Theorem~\ref{thm:Eckart-Young-gsvd} with $\scalebox{0.8}{\Rel{R}}=\scalebox{0.8}{\CoMap{A}}$.

\begin{hcalculation}[\iff]{\scalebox{0.7}{\CoMap{A_{\overline{r-k}}^{+}}}  \in \underset{\scalebox{0.7}{\CoMap{S}}}{\argmin}\left[\left\|\scalebox{0.7}{\Map{S}} \right\|_F\;:\;\rank\left(\scalebox{0.7}{\input{Golub_a-s.tikz}}\right)=k\right]}
\hstep{\scalebox{0.7}{\CoMap{A_{\overline{r-k}}}}  \in \underset{\scalebox{0.7}{\CoMap{S}}}{\argmin}\left[\left\|\scalebox{0.7}{\Map{S}} \right\|_F\;:\;\rank\left(\scalebox{0.7}{\input{Golub_a-s.tikz}}\right)=k\right]}{$k_I = k_S = 0$}
\hstep{\scalebox{0.7}{\input{GSVD_A-A_k.tikz}}\in\underset{\scalebox{0.6}{\input{Golub_a-s_op.tikz}}}{\argmin}\left[\left\|\scalebox{0.6}{\input{Golub_a-s.tikz}} \right\|_F\;:\;\rank\left(\scalebox{0.7}{\Map{S}}\right)=k\right]}{\makecell{Variable change:\\$\scalebox{0.6}{\Map{S}}\coloneqq \scalebox{0.6}{\input{Golub_a-s.tikz}}$}}
\hstep{\scalebox{0.7}{\CoMap{A_k}}\in\underset{\scalebox{0.7}{\CoMap{S}}}{\argmin}\left[\left\|\scalebox{0.7}{\input{Golub_a-s.tikz}} \right\|_F\;:\;\rank\left(\scalebox{0.7}{\Map{S}}\right)=k\right]}{Simplify}
\end{hcalculation}
\end{proof}

Lemma~\ref{lem:lowrank-to-upperlimit} shows that Theorem~\ref{thm:Eckart-Young-gsvd} recovers Theorem~\ref{thm:teo7-nosso} when $\ell=r$, that is, it finds the smallest upper limit of a linear relation.

\begin{lemma}[Upper limit with least norm of a linear relation via Relational Low-Rank Approximation]\label{lem:lowrank-to-upperlimit}
Let $\TypedRel{R}{m}{n}$ be an injective linear relation and $\TypedCoMap{R^+}{m}{n}$ its pseudoinverse defined as in Definition~\ref{def:pseudo-rel}. Then
\[\TypedCoMap{R^+}{m}{n}\in\underset{\scalebox{0.7}{\TypedCoMap{S}{m}{n}}}{\argmin}\left[\left\|\scalebox{0.9}{\Map{S}} \right\|_F\;:\;\scalebox{0.9}{\Rel{R}}\subseteq \scalebox{0.9}{\CoMap{S}}\right],\]
\end{lemma}
\begin{proof} We instantiate Theorem~\ref{thm:Eckart-Young-gsvd} with $r = \ell$.
\begin{hcalculation}[\Leftrightarrow ]{\scalebox{0.8}{\CoMap{R_{\overline{r}}^{+}}}\in\underset{\scalebox{0.7}{\CoMap{S}}}{\argmin}\left[\left\|\scalebox{0.8}{\Map{S}} \right\|_F\;:\;\rank\left(\scalebox{0.8}{\input{Golub_r-d.tikz}}\right)=0\right]}
\hstep{\scalebox{0.8}{\CoMap{R^+}}\in\underset{\scalebox{0.7}{\CoMap{S}}}{\argmin}\left[\left\|\scalebox{0.8}{\Map{S}} \right\|_F\;:\;\scalebox{0.8}{\Rel{R}}\subseteq \scalebox{0.8}{\CoMap{S}}\right]}{$R_{\bar{r}}^{+}=R^{+}$ and Lem.\ref{lem:thm7}}
\end{hcalculation}
\end{proof}

By instantiating $\ell=r$, we obtain another classic result when considering the linear relation $\scalebox{0.8}{\Rel{R}}$ as a linear transformation $\scalebox{0.8}{\Map{A}}$: the pseudoinverse is the left-inverse of least norm of a linear transformation. Lemma~\ref{lem:left_inverse} shows this result.

\begin{lemma}[Left-inverse with Least Norm]\label{lem:left_inverse}
  The pseudoinverse of an \emph{injective} linear transformation $\TypedMap{A}{m}{n}$ satisfies
  \[\TypedCoMap{A^+}{m}{n}\in\underset{\scalebox{0.7}{\TypedCoMap{S}{m}{n}}}{\argmin}\left[\left\|\scalebox{0.9}{\Map{S}} \right\|_F\;:\;\scalebox{0.9}{\input{EY_left-pseudo.tikz}}\right],\]
\end{lemma}
\begin{proof} We instantiate Theorem~\ref{thm:Eckart-Young-gsvd} with $r = \ell$ and $\scalebox{0.8}{\Rel{R}}=\scalebox{0.8}{\Map{A}}$.
  \begin{hcalculation}[\Leftrightarrow ]{\scalebox{0.8}{\CoMap{A_{\overline{r}}^{+}}}\in\underset{\scalebox{0.7}{\CoMap{S}}}{\argmin}\left[\left\|\scalebox{0.8}{\Map{S}} \right\|_F\;:\;\rank\left(\scalebox{0.8}{\input{Golub_a-d.tikz}}\right)=0\right]}
  \hstep{\scalebox{0.8}{\CoMap{A_{\overline{r}}^{+}}}\in\underset{\scalebox{0.7}{\CoMap{S}}}{\argmin}\left[\left\|\scalebox{0.8}{\Map{S}} \right\|_F\;:\;\scalebox{0.8}{\input{Golub_im_a-d.tikz}}=\scalebox{0.8}{\Zero}\right]}{Def.~\ref{def:rank-linear-relation}}
    \hstep{\scalebox{0.8}{\CoMap{A^+}}\in\underset{\scalebox{0.7}{\CoMap{S}}}{\argmin}\left[\left\|\scalebox{0.8}{\Map{S}} \right\|_F\;:\;\scalebox{0.8}{\input{EY_left-pseudo.tikz}}\right]}{Lem.\ref{lem:thm7}}
  \end{hcalculation}
\end{proof}

Instances of Theorem~\ref{thm:Eckart-Young-gsvd} also allow us to recover some of the optimization problems proposed in~\cite{Golub:1991}, already cited throughout this work. To do this, it suffices to consider the linear relation $\scalebox{0.8}{\Rel{R}}$ in its cospan form (Proposition ~\ref{prop:span-cospan}).

\begin{lemma}[{\cite[Thm.~6]{Golub:1991}} via Relational Low-Rank Approximation]\label{lem:lowrank-to-golub6}
For every linear relation $\Rel{R}=\Span{A}{B}$, there exist orthogonal linear transformations \(\Ortho{U}\) and \(\Ortho{V}\), and a diagonal linear relation \(\Rel{D}\), defined as in Definition~\ref{def:diagonalrelation}, such that
\[\scalebox{0.9}{\input{GSVD_GSVD_pseudoinversa_withD_r-k_op.tikz}}  \in \underset{\scalebox{0.7}{\Map{S}}}{\argmin}\left[\left\|\Map{S} \right\|_F\;:\;\rank\left(\scalebox{0.9}{\input{Golub_a-dc.tikz}}\right)=k\right].\]
\end{lemma}
\begin{proof}
Apply Theorem~\ref{thm:Eckart-Young-gsvd} with $\scalebox{0.8}{\Rel{R}}=\scalebox{0.8}{\Span{A}{B}}$.
\end{proof}

\begin{lemma}[{\cite[Thm.~7]{Golub:1991}} via Relational Low-Rank Approximation]\label{lem:lowrank-to-golub7}
For every linear relation $\Rel{R}=\Span{A}{B}$, we have that:
\[\input{GSVD_GSVD_pseudoinversa_withD_op.tikz}  \in \underset{\scalebox{0.7}{\Map{S}}}{\argmin}\left[\left\|\Map{S} \right\|_F\;:\;\Map{A}=\begin{tikzpicture}
	\begin{pgfonlayer}{nodelayer}
		\node [style=Map] (0) at (-1, 0) {$B$};
		\node [style=Map] (1) at (1, 0) {$S$};
		\node [style=none] (2) at (-2.5, 0) {};
		\node [style=none] (3) at (2.5, 0) {};
	\end{pgfonlayer}
	\begin{pgfonlayer}{edgelayer}
		\draw (2.center) to (3.center);
	\end{pgfonlayer}
\end{tikzpicture}
 \right]\]
\end{lemma}
\begin{proof} By substituting $\CoMap{R^+} = \input{GSVD_GSVD_pseudoinversa_withD2.tikz}$ in Lemma~\ref{lem:lowrank-to-upperlimit} we obtain
\begin{hcalculation}[\Leftrightarrow]{\scalebox{0.8}{\input{GSVD_GSVD_pseudoinversa_withD2.tikz}}\in\underset{\scalebox{0.7}{\CoMap{S}}}{\argmin}\left[\left\|\scalebox{0.8}{\Map{S}} \right\|_F\;:\;\scalebox{0.8}{\Rel{R}}\subseteq \scalebox{0.8}{\CoMap{S}}\right]}
\hstep{\scalebox{0.8}{\input{GSVD_GSVD_pseudoinversa_withD_op.tikz}}  \in \underset{\scalebox{0.7}{\Map{S}}}{\argmin}\left[\left\|\scalebox{0.8}{\Map{S}} \right\|_F\;:\;\scalebox{0.8}{\Map{A}}=\scalebox{0.8}{\begin{tikzpicture}
	\begin{pgfonlayer}{nodelayer}
		\node [style=Map] (0) at (-1, 0) {$B$};
		\node [style=Map] (1) at (1, 0) {$S$};
		\node [style=none] (2) at (-2.5, 0) {};
		\node [style=none] (3) at (2.5, 0) {};
	\end{pgfonlayer}
	\begin{pgfonlayer}{edgelayer}
		\draw (2.center) to (3.center);
	\end{pgfonlayer}
\end{tikzpicture}
} \right]}{Algebra +  Lem.\ref{lem:thm7}}
\end{hcalculation}
\end{proof}

One can also recover the main result from~\cite{Golub1987}, which describes a solution to the problem of low-rank approximation with a specified number of fixed columns.

\begin{lemma}[\cite{Golub1987} via Relational Low-Rank Approximation]\label{lem:lowrank-to-golub87} Let $\Rel{R}=\input{Golub_lowrank_golub.tikz}$ be a linear relation. Then there exist $\Map{A_1}$ and $\Map{A_2}$ such that $\Map{A}=\input{Golub_im_r-d-golub3.tikz}$ and
\[\scalebox{0.9}{\input{Golub_golub87_golub1.tikz}}\in\underset{\scalebox{0.7}{\CoMap{\bar{A}_2}}}{\argmin}\left[\left\|\scalebox{0.9}{\input{Golub_golub87_golub2.tikz}} \right\|_F\;:\;\rank\left(\scalebox{0.9}{\input{Golub_golub87_golub3.tikz}}\right)=r-\ell\right],\]
\end{lemma}
\begin{proof}
We instantiate Theorem~\ref{thm:Eckart-Young-gsvd} with $\scalebox{0.8}{\Rel{R}}=\scalebox{0.8}{\input{Golub_lowrank_golub.tikz}}$.
\begin{hcalculation}[\Leftrightarrow]{\scalebox{0.7}{\CoMap{R_{\overline{\ell}}^{+}}}  \in \underset{\scalebox{0.7}{\CoMap{S}}}{\argmin}\left[\left\|\scalebox{0.7}{\Map{S}} \right\|_F\;:\;\rank\left(\scalebox{0.7}{\input{Golub_golub87_golub4.tikz}}\right)=r-\ell\right]}
\hstep{\scalebox{0.7}{\CoMap{R_{\overline{\ell}}^{+}}}  \in \underset{\scalebox{0.7}{\CoMap{S}}}{\argmin}\left[\left\|\scalebox{0.7}{\Map{S}} \right\|_F\;:\;\rank\left(\scalebox{0.7}{\input{Golub_golub87_golub5.tikz}}\right)=r-\ell\right]}{Lem.\ref{lem:golub87}}
\hstep{\scalebox{0.7}{\CoMap{R_{\overline{\ell}}^{+}}}  \in \underset{\scalebox{0.7}{\input{Golub_golub87_golub7.tikz}}}{\argmin}\left[\left\|\scalebox{0.7}{\input{Golub_golub87_golub8.tikz}} \right\|_F\;:\;\rank\left(\scalebox{0.7}{\input{Golub_golub87_golub9.tikz}}\right)=r-\ell\right]}{\makecell{Variable change:\\\scalebox{0.7}{\input{Golub_golub87_golub6.tikz}}}}
\hstep{\scalebox{0.7}{\input{Golub_golub87_golub10.tikz}}  \in \underset{\scalebox{0.7}{\CoMap{\bar{A}_2}}}{\argmin}\left[\left\|\scalebox{0.7}{\input{Golub_golub87_golub8.tikz}} \right\|_F\;:\;\rank\left(\scalebox{0.7}{\input{Golub_golub87_golub9.tikz}}\right)=r-\ell\right]}{Algebra}
\end{hcalculation}
\end{proof}

Finally, Theorem~\ref{thm:Eckart-Young-gsvd} can also be instantiated using subspaces, arriving at all of the results shown in Table~\ref{tab:instantiations_of_geyt}. Below, as an example, we prove two of them.

\begin{lemma}\label{lem:subspace-copy} Let $\Rel{R}=\input{subspace_A-imB.tikz}$ be a linear relation where $\Subspace{V}$ is a subspace and $\Map{A}$ is a linear transformation. Then
\[\input{subspace_a-pseudo.tikz}\in\underset{\scalebox{0.7}{\TypedCoMap{\tilde{S}}{m}{n}}}{\argmin}\left[\left\|\scalebox{0.9}{\input{subspace_a-Stildeop.tikz}} \right\|_F\;:\;\rank\left(\scalebox{0.9}{\input{subspace_B-Stilde.tikz}}\right)=r-\ell\right].\]
\end{lemma}
\begin{proof}We instantiate Theorem~\ref{thm:Eckart-Young-gsvd} with $\scalebox{0.8}{\Rel{R}}=\scalebox{0.8}{\input{subspace_A-imB.tikz}}$.
\begin{hcalculation}[\Leftrightarrow]{\scalebox{0.7}{\CoMap{R_{\overline{\ell}}^{+}}}  \in \underset{\scalebox{0.7}{\CoMap{S}}}{\argmin}\left[\left\|\scalebox{0.7}{\Map{S}} \right\|_F\;:\;\rank\left(\scalebox{0.7}{\input{subspace_step1.tikz}}\right)=r-\ell\right]}
\hstep{\scalebox{0.7}{\CoMap{R_{\overline{\ell}}^{+}}}  \in \underset{\scalebox{0.7}{\CoMap{S}}}{\argmin}\left[\left\|\scalebox{0.7}{\Map{S}} \right\|_F\;:\;\rank\left(\scalebox{0.7}{\input{subspace_step2.tikz}}\right)=r-\ell\right]}{Fig.~\ref{fig:Commutative_Comonoid}}
\hstep{\scalebox{0.7}{\CoMap{R_{\overline{\ell}}^{+}}}  \in \underset{\scalebox{0.6}{\input{subspace_a-Stilde.tikz}}}{\argmin}\left[\left\|\scalebox{0.7}{\input{subspace_a-Stildeop.tikz}} \right\|_F\;:\;\rank\left(\scalebox{0.7}{\input{subspace_B-Stilde.tikz}}\right)=r-\ell\right]}{\makecell{Variable change\\$\scalebox{0.7}{\Map{\tilde{S}}}:=\scalebox{0.7}{\input{subspace_A-Sop.tikz}}$}}
\hstep{\scalebox{0.7}{\input{subspace_a-pseudo.tikz}}  \in \underset{\scalebox{0.7}{\CoMap{\tilde{S}}}}{\argmin}\left[ \left\|\scalebox{0.7}{\input{subspace_a-Stildeop.tikz}} \right\|_F\;:\;\rank\left(\scalebox{0.7}{\input{subspace_B-Stilde.tikz}}\right)=r-\ell\right]}{Algebra}
\end{hcalculation}
\end{proof}

\begin{lemma}\label{lem:subspace-sum} Let $\Rel{R}=\input{subspace_V-mais-Aop.tikz}$ be a linear relation where $\Subspace{V}$ is a subspace and $\Map{A}$ is a linear transformation. Then
\[\input{subspace_a-pseudo.tikz}\in\underset{\scalebox{0.7}{\TypedCoMap{\tilde{S}}{m}{n}}}{\argmin}\left[\left\|\scalebox{0.9}{\input{subspace_a-Stildeop.tikz}} \right\|_F\;:\;\rank\left(\scalebox{0.9}{\input{subspace_V-mais-S.tikz}}\right)=r-\ell\right].\]
\end{lemma}
\begin{proof}We instantiate Theorem~\ref{thm:Eckart-Young-gsvd} with $\scalebox{0.8}{\Rel{R}}=\scalebox{0.8}{\input{subspace_V-mais-Aop.tikz}}$, and the proof proceeds analogously to the proof of Lemma~\ref{lem:subspace-copy}.
\end{proof}

\begin{lemma}\label{lem:2subspaces} Let $\Rel{R}=\input{subspace_2subspaces.tikz}$ be a linear relation where $\Subspace{U}$ and $\Subspace{V}$ are subspaces and $\Map{A}$ is a linear transformation. Then
\[\input{subspace_a-pseudo.tikz}\in\underset{\scalebox{0.7}{\TypedCoMap{\tilde{S}}{m}{n}}}{\argmin}\left[\left\|\scalebox{0.9}{\input{subspace_a-Stildeop.tikz}} \right\|_F\;:\;\rank\left(\scalebox{0.9}{\input{subspace_2subspacesop.tikz}}\right)=r-\ell\right].\]
\end{lemma}
\begin{proof}We instantiate Theorem~\ref{thm:Eckart-Young-gsvd} with $\scalebox{0.8}{\Rel{R}}=\scalebox{0.8}{\input{subspace_2subspaces.tikz}}$, and the proof proceeds analogously to the proof of Lemma~\ref{lem:subspace-copy}.
\end{proof}

\appendix
\section{Axioms and Theorems in Graphical Notation}
\label{ap:gla}

Most of the proofs in this text consist of sequences of diagram manipulations. This subsection contains all the necessary rules in their diagrammatic form. We present the theorems of linear relations in graphical syntax, which will be used later. Only the theorems necessary for the presented proofs are recorded here. A complete presentation can be found at~\cite{PAIXAO2022}. In order to maintain the focus of this work, most of the proofs in this section will be omitted, as they can be found in the same reference.

Raw terms are quotiented with respect to the laws of symmetric strict monoidal (SSM) categories, summarized in Fig.\ref{fig:lawsSMC}. We omit the (well-known) details~\cite{selinger2010survey} here and mention only that this amounts to eschewing the need for ``dotted line boxes''
and ensuring that diagrams with the same topological connectivity are equated. For the purposes of this work, a detailed knowledge of strict symmetric monoidal categories is not required, since we use this structure only as a convenient tool for graphical reasoning.

\begin{figure}[pos=h]
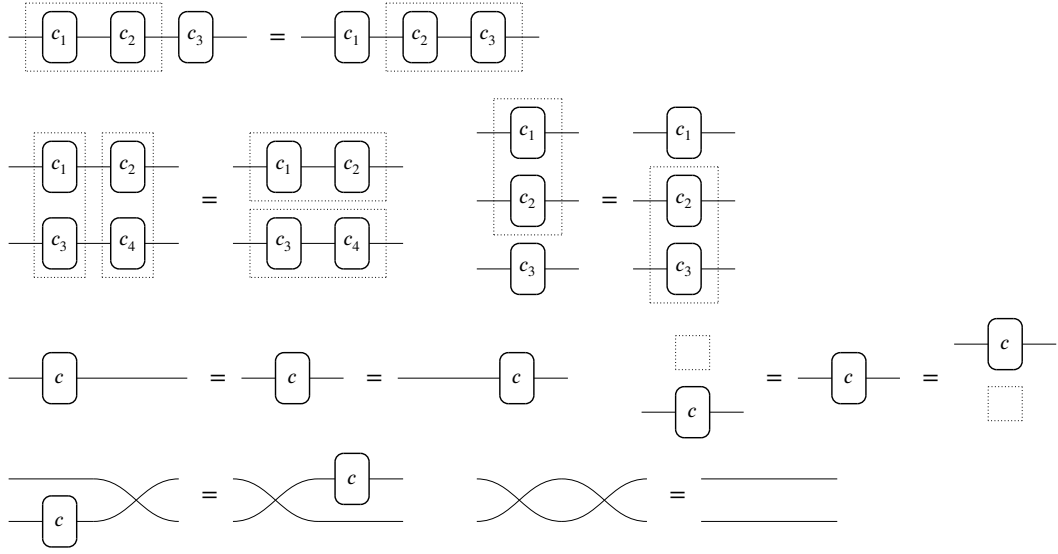

  \begin{center}
    \[
      \scalebox{0.9}{\input{SSM_sequential-associativity.tikz}} = \scalebox{0.9}{\input{SSM_sequential-associativity-1.tikz}}
    \]
    \[
      \scalebox{0.9}{\input{SSM_interchange-law.tikz}} = \scalebox{0.9}{\input{Introdution_SSM_interchange-law-1.tikz}}
      \qquad
      \scalebox{0.9}{\input{SSM_parallel-associativity.tikz}} = \scalebox{0.9}{\input{SSM_parallel-associativity-1.tikz}}
    \]
    \[
      \scalebox{0.9}{\input{SSM_unit-right2.tikz}} = \scalebox{0.9}{\input{SSM_c.tikz}}
      = \scalebox{0.9}{\input{SSM_unit-left2.tikz}}
      \quad\quad
      \scalebox{0.9}{\input{SSM_parallel-unit-above.tikz}} = \scalebox{0.9}{\input{SSM_c.tikz}} =  \scalebox{0.9}{\input{SSM_parallel-unit-below.tikz}}
    \]
    \[
      \scalebox{0.9}{\input{SSM_sym-natural.tikz}} = \scalebox{0.9}{\input{SSM_sym-natural-1.tikz}}
      \qquad
      \scalebox{0.9}{\input{SSM_sym-iso.tikz}} = \scalebox{0.9}{\input{SSM_id2.tikz}}
    \]
    \caption{Laws of Symmetric Strict Monoidal (SSM) Categories.}
    \label{fig:lawsSMC}
  \end{center}
\end{figure}

In graphical notation, the diagrams are closed under two symmetries: \emph{Mirror-Image} and \emph{Color-Swap}.
The theorems and their corresponding symmetries are summarized in Figures~\ref{fig:Commutative_Comonoid} to~\ref{fig:min_max}.
Finally, Figure~\ref{fig:typerelations} outlines statements determining whether a relation is classified as total (TOT), deterministic (DET), surjective (SUR), or injective (INJ).

Lemmas~\ref{lem:weak-inverse} to~\ref{lem:pa_pum_igual} show how compositions of linear relations and linear transformations behave in equalities and inequalities.

\begin{lemma}[Opposite is Weak Inverse]
  \label{lem:weak-inverse}
  For every linear transformation $\Map{A}$,
  \[ \Comp{A/Map,A/CoMap,A/Map} = \Map{A}. \]
\end{lemma}

\begin{lemma}
  \label{lem:pair_of_matrices_1}
  Let $\Map{A}$ and $\Map{B}$ be two linear transformations, then
  \[\input{Papum_Lemma1_lemma.tikz}.\]
\end{lemma}
\begin{proof}
  \begin{hcalculation}[\subseteq]{\scalebox{0.8}{\Map{B}}}
    \hstep{\scalebox{0.8}{\input{Papum_Lemma1_proof1.tikz}}}{Fig.\ref{fig:typerelations}-TOT: \scalebox{0.7}{\input{theorems_totalid.tikz}}}
    \hstep{\scalebox{0.8}{\input{Papum_Lemma1_proof2.tikz}}}{Hyp: $\scalebox{0.7}{\Map{A}}\subseteq \scalebox{0.7}{\Map{B}}$}
    \hstep{\scalebox{0.8}{\Map{A}}}{Fig.\ref{fig:typerelations}-DET: \scalebox{0.7}{\input{theorems_detid.tikz}}}
  \end{hcalculation}
\end{proof}

\begin{lemma}\label{lem:pa_pum}
For every linear relations $\Rel{A_2}$, $\Rel{B_1}$ and linear transformations $\Map{A_1}$, $\Map{B_2}$,
\[\input{Papum_lemma.tikz}.\]
\end{lemma}
\begin{proof}
($\Rightarrow$)
\begin{hcalculation}[\subseteq]{\scalebox{0.8}{\begin{tikzpicture}
	\begin{pgfonlayer}{nodelayer}
		\node [style=none] (0) at (3, 0) {};
		\node [style=Map] (1) at (1.5, 0) {$B_2$};
		\node [style=Rel] (2) at (-0.75, 0) {$A_2$};
		\node [style=none] (3) at (-2.25, 0) {};
	\end{pgfonlayer}
	\begin{pgfonlayer}{edgelayer}
		\draw (3.center) to (2);
		\draw (2) to (1);
		\draw (1) to (0.center);
	\end{pgfonlayer}
\end{tikzpicture}
}}
  \hstep{\scalebox{0.8}{\input{Papum_idaprova2.tikz}}}{Fig.\ref{fig:typerelations}-TOT: \scalebox{0.7}{\input{theorems_totalid.tikz}}}
  \hstep{\scalebox{0.8}{\input{Papum_idaprova3.tikz}}}{Hyp: \scalebox{0.7}{\input{Papum_hip1.tikz}}}
  \hstep{\scalebox{0.8}{\begin{tikzpicture}
	\begin{pgfonlayer}{nodelayer}
		\node [style=none] (2) at (3.25, 0) {};
		\node [style=Rel] (3) at (1.5, 0) {$B_1$};
		\node [style=Map] (4) at (-0.75, 0) {$A_1$};
		\node [style=none] (5) at (-2.25, 0) {};
	\end{pgfonlayer}
	\begin{pgfonlayer}{edgelayer}
		\draw (3) to (2.center);
		\draw (5.center) to (4);
		\draw (4) to (3);
	\end{pgfonlayer}
\end{tikzpicture}
}}{Fig.\ref{fig:typerelations}-DET: \scalebox{0.7}{\input{theorems_detid.tikz}}}
\end{hcalculation}
($\Leftarrow$)
\begin{hcalculation}[\subseteq]{\scalebox{0.8}{\begin{tikzpicture}
	\begin{pgfonlayer}{nodelayer}
		\node [style=none] (0) at (3.25, 0) {};
		\node [style=Rel] (1) at (1.5, 0) {$A_2$};
		\node [style=CoMap] (2) at (-0.75, 0) {$A_1$};
		\node [style=none] (3) at (-2.25, 0) {};
	\end{pgfonlayer}
	\begin{pgfonlayer}{edgelayer}
		\draw (3.center) to (2);
		\draw (2) to (1);
		\draw (1) to (0.center);
	\end{pgfonlayer}
\end{tikzpicture}
}}
  \hstep{\scalebox{0.8}{\input{Papum_voltaprova2.tikz}}}{Fig.\ref{fig:typerelations}-TOT: \scalebox{0.7}{\input{theorems_totalid.tikz}}}
  \hstep{\scalebox{0.8}{\input{Papum_voltaprova3.tikz}}}{Hyp: \scalebox{0.7}{\input{Papum_hip2.tikz}}}
  \hstep{\scalebox{0.8}{\begin{tikzpicture}
	\begin{pgfonlayer}{nodelayer}
		\node [style=none] (2) at (3, 0) {};
		\node [style=CoMap] (3) at (1.5, 0) {$B_2$};
		\node [style=Rel] (4) at (-1, 0) {$B_1$};
		\node [style=none] (5) at (-2.75, 0) {};
	\end{pgfonlayer}
	\begin{pgfonlayer}{edgelayer}
		\draw (3) to (2.center);
		\draw (5.center) to (4);
		\draw (4) to (3);
	\end{pgfonlayer}
\end{tikzpicture}
}}{Fig.\ref{fig:typerelations}-DET: \scalebox{0.7}{\input{theorems_detid.tikz}}}
\end{hcalculation}
\end{proof}

\begin{lemma}\label{lem:pa_pum_igual}
  For every linear transformations $\Map{A_1}$, $\Map{A_2}$, $\Map{B_1}$, $\Map{B_2}$,
  \[\def\arraystretch{2.0}
  \begin{array}{lrclcrcl}
   1. & \Span{A_1}{A_2} &\subseteq& \CoSpan{B_1}{B_2}
      &\iff&
      \Comp{A_1/Map, B_1/Map} &=& \Comp{A_2/Map, B_2/Map} \\
    2.& \Image{A_2} &\subseteq& \Kernel{B_2}
      &\iff&
      \Discard \; \Zero &=& \Comp{A_2/Map, B_2/Map} \\
    3.& \Map{A_2} &\subseteq& \CoMap{B_2}
      &\iff&
      \Id &=& \Comp{A_2/Map, B_2/Map} \\
    4.& \CoMap{A_1} &\subseteq& \Map{B_1}
      &\iff&
      \Comp{A_1/Map, B_1/Map} &=& \Id \\
    5.& \Map{A_2} &\subseteq& \CoSpan{B_1}{B_2}
      &\iff&
      \Map{B_1} &=& \Comp{A_2/Map, B_2/Map} \\
    6.& \Span{A_1}{A_2} &\subseteq& \Map{B_1}
        &\iff&
        \Comp{A_1/Map, B_1/Map} &=& \Map{A_2}
  \end{array}
  \]
\end{lemma}
\begin{proof}
We first prove item (i).
By Lemma~\ref{lem:pa_pum}, $\scalebox{0.8}{\CompTwo{.,Map,A_2,Map,B_2,.}} \subseteq \scalebox{0.8}{\CompTwo{.,Map,A_1,Map,B_1,.}}$.
The equality follows by Lemma~\ref{lem:pair_of_matrices_1}.
The remaining items are instances of the first one.
\end{proof}

\begin{lemma}\label{lem:golub87} For every linear relation $\Rel{R}=\input{Golub_lowrank_golub.tikz}$ and every cofunction $\CoMap{S}$, there exist linear transformations $\Map{A_1}$ and $\Map{A_2}$ such that
  \[\input{Golub_im_r-d.tikz}=\input{Golub_im_r-d-golub.tikz}.\]
\end{lemma}
\begin{proof}
  \begin{hcalculation}[=]{\scalebox{0.8}{\input{Golub_im_r-d.tikz}}}
    \hstep{\scalebox{0.8}{\input{Golub_im_r-d-golub1.tikz}}}{Def. $\scalebox{0.7}{\Rel{R}}$}
    \hstep{\scalebox{0.8}{\input{Golub_im_r-d-golub2.tikz}}}{Fig.\ref{fig:Split}: $\scalebox{0.7}{\Map{A}}=\scalebox{0.7}{\input{Golub_im_r-d-golub3.tikz}}$}
    \hstep{\scalebox{0.8}{\input{Golub_im_r-d-golub.tikz}}}{Fig.\ref{fig:Commutative_Comonoid}}
  \end{hcalculation}
\end{proof}

\begin{figure}[pos=htb]
    \centering
    \scalebox{0.9}{\input{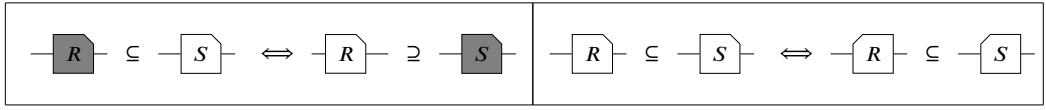}}
    \caption{The Symmetry Theorems}
    \label{fig:Symmetry}
\end{figure}

\begin{figure}[pos=htb]
    \centering
    \scalebox{0.9}{\input{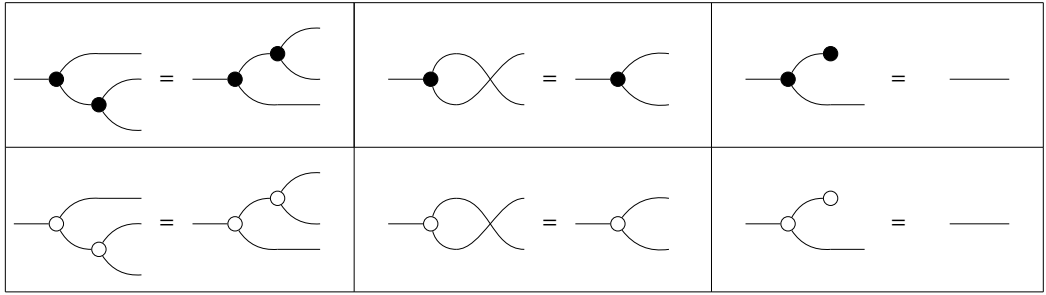}}
    \caption{Commutative Comonoid}
    \label{fig:Commutative_Comonoid}
\end{figure}

\begin{figure}[pos=htb]
    \centering
    \scalebox{0.9}{\input{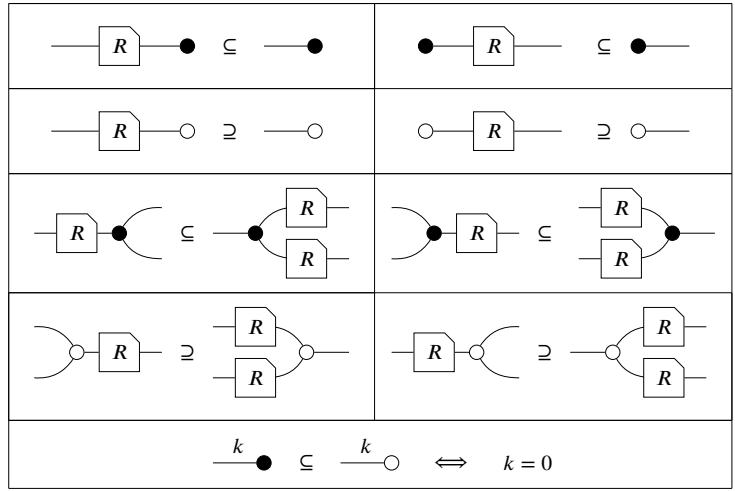}}
    \caption{The Inequality Theorems}
    \label{fig:Inequality}
\end{figure}

\begin{figure}[pos=htb]
    \centering
    \scalebox{0.9}{\input{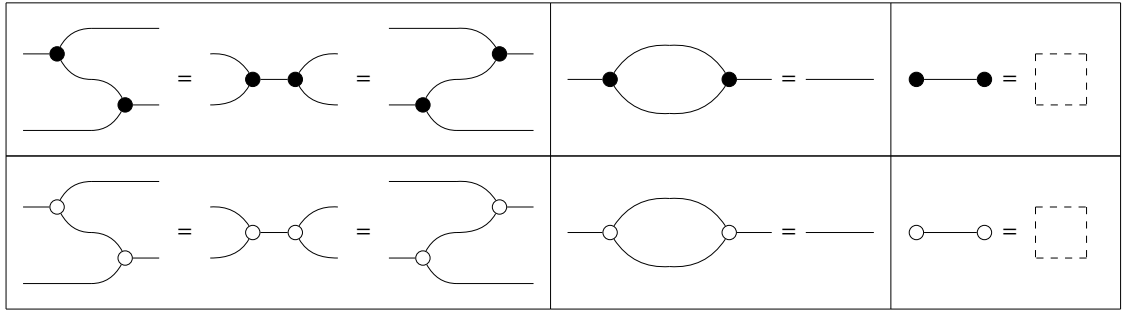}}
    \caption{Frobenius Algebra}
    \label{fig:Frobenius_Algebra}
\end{figure}

\begin{figure}[pos=htb]
    \centering
    \scalebox{0.9}{\input{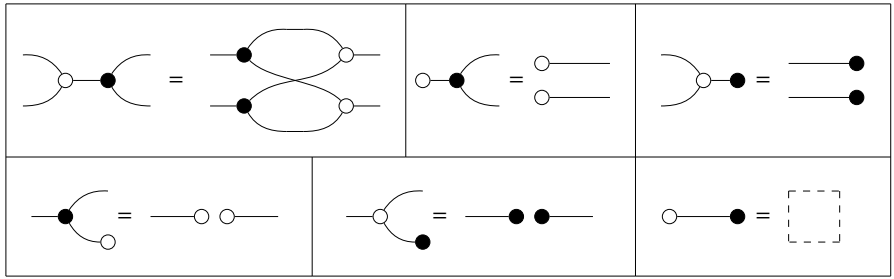}}
    \caption{Bialgebra}
    \label{fig:Bialgebra}
\end{figure}

\begin{figure}[pos=htb]
    \centering
    \scalebox{0.9}{\input{Appendix_Flip.tikz}}
    \caption{Flip}
    \label{fig:Flip}
\end{figure}

\begin{figure}[pos=htb]
    \centering
    \scalebox{0.9}{\input{Appendix_Snake.tikz}}
    \caption{Snake}
    \label{fig:Snake}
\end{figure}

\begin{figure}[pos=htb]
  \centering
  \begin{minipage}[t]{.6\textwidth}
    \centering
    \scalebox{0.9}{\input{Appendix_Connect_sum.tikz}}
    \begingroup
      \ExplSyntaxOn
      \dim_set:Nn \l_fig_width_dim { \linewidth }
      \ExplSyntaxOff
      \caption{Connect sum}
      \label{fig:Connect_sum}
    \endgroup
  \end{minipage}\hfill
  \begin{minipage}[t]{.4\textwidth}
    \centering
    \scalebox{0.9}{\input{Appendix_Split.tikz}}
    \begingroup
      \ExplSyntaxOn
      \dim_set:Nn \l_fig_width_dim { \linewidth }
      \ExplSyntaxOff
      \caption{Split}
      \label{fig:Split}
    \endgroup
  \end{minipage}

  \medskip

  \begin{minipage}[t]{.6\textwidth}
    \centering
    \scalebox{0.9}{\input{Appendix_Compare.tikz}}
    \begingroup
      \ExplSyntaxOn
      \dim_set:Nn \l_fig_width_dim { \linewidth }
      \ExplSyntaxOff
      \caption{Compare}
      \label{fig:Compare}
    \endgroup
  \end{minipage}\hfill
  \begin{minipage}[t]{.4\textwidth}
    \centering
    \scalebox{0.9}{\input{Appendix_Minimum_and_maximum.tikz}}
    \begingroup
      \ExplSyntaxOn
      \dim_set:Nn \l_fig_width_dim { \linewidth }
      \ExplSyntaxOff
      \caption{Minimum and maximum relations}
      \label{fig:min_max}
    \endgroup
  \end{minipage}
\end{figure}

\begin{figure}[pos=htb]
    \centering
    \resizebox{\textwidth}{!}{\scalebox{0.9}{\input{Appendix_typesofrelations.tikz}}}
    \caption{Each column of the table is a series of equivalent statements.
      Furthermore, the opposite inequalities for the second and third row hold for any relation (See Figure~\ref{fig:Inequality}).
      We are, thus, allowed to treat these characterizations as equalities.
    }
    \label{fig:typerelations}
\end{figure}

\subsection{Translating graphical notation to classical notation}\label{subsec:translating}

In the graphical syntax, each diagram semantically corresponds to a linear relation. The translation of diagrams to classical notation is compositional and translates their semantic meaning; that is, the meaning of a compound diagram is calculated from the meanings of its sub-diagrams. Here there is a connection with relational algebra, the two operations of diagram composition are mapped to the standard ways of composing relations: relational composition and cartesian product.
\begin{align*}
  \input{Introdution_GLA_Copy.tikz} &\longmapsto \left\{\left(x,\,\begin{bmatrix}y_1\\y_2 \end{bmatrix}\right) \;\mid\; x,y_1,y_2\in \mathbb{R}\;\text{and}\;y_1=y_2=x\right\} \\
  \Discard &\longmapsto \left\{(x,*)\;|\; x\in \mathbb{R} \right\} \\
  \input{Introdution_GLA_sum.tikz}  &\longmapsto \left\{\left(\begin{bmatrix}x_1\\x_2 \end{bmatrix},\,y\right) \;\mid\; x_1,x_2,y\in \mathbb{R}\;\text{and}\;y=x_1+x_2\right\} \\
  \Zero &\longmapsto \{(*,y)\;|\; y\in \mathbb{R}\;\text{and}\;y=0\}
\end{align*}
The generators $\input{Introdution_GLA_CoCopy.tikz},\CoDiscard,\input{Introdution_GLA_CoSum.tikz}$ and $\CoZero$ are, respectively, the opposite of the relations above, as hinted by the symmetric graphical syntax. Its interpretation in classical notation is therefore trivial.

\begin{example}
Here's an example of how to translate a more complicated diagram. The intuitive way to think about diagrams is like a series of logical gates, where the input starts on the left and flows to the right via the wires. In the image below we include some variable names $x_1, x_2, x_3, y_1, y_2$ (which are not wire numberings) to aid comprehension.
    \[\scalebox{0.9}{\input{Introdution_GLA_example2.tikz}}\]
By visualizing the $x_i$'s flowing to the right, we can see $x_2$ is discarded and $x_3$ is copied. One of the copies is passed to $y_1$ while the other is added with $x_1$ to form $y_2$. So $y_1 = x_3$ and $y_2 = x_1 + x_3$. Thus, written in classical notation this diagram corresponds to the relation
    $$\left\{\left(\begin{bmatrix}x_1\\x_2\\x_3 \end{bmatrix},\,
    \begin{bmatrix}y_1\\y_2 \end{bmatrix}\right) \in \mathbb{R}^3 \times \mathbb{R}^2 \;|\; y_1=x_3 \text{ and }  y_2=x_1+x_3 \right\}.$$
   For the purpose of simplifying the notation, sometimes in this paper, the inclusion ``$\in \mathbb{R}^3 \times \mathbb{R}^2$” will be omitted and diagrams such as the one presented above will be written in classical notation as follows.
    $$\left\{\left(\begin{bmatrix}x_1\\x_2\\x_3 \end{bmatrix},\,
    \begin{bmatrix}y_1\\y_2 \end{bmatrix}\right) \;|\; y_1=x_3 \text{ and } y_2=x_1+x_3\right\}.$$
\end{example}

In our case, string diagrams also implicitly handle some non-trivial rules in the usual syntax, just as is commonly done with associativity of functions by ignoring parentheses. Table~\ref{tab:freetheoremsintro} exemplifies how some rules are handled implicitly in the diagrammatic language, with the dotted lines simulating parentheses. It is worth noting that the dashed lines are represented here for illustrative purposes only and are not part of the notation.

\begin{table}[pos=htb]
  \centering
  \begin{tabular}{c c}
    Usual syntax & Graphical syntax \\ \hline
    $A(BC) = (AB)C$ & $\scalebox{0.8}{\input{Introdution_Relations_composition_assoc.tikz}}$\\ \hline
    $A(\mathrm{Im}(B)) = \mathrm{Im}(AB)$ & $\scalebox{0.8}{\input{Introdution_Relations_im-composition.tikz}}$ \\ \hline
    $\mathrm{Im}(A) + \mathrm{Im}(B) = \mathrm{Im}(\begin{bmatrix} A & B \end{bmatrix})$ & $\scalebox{0.8}{\input{Introdution_Relations_sumimage.tikz}}$ \\ \hline
    $(C \times D) (A \times B) = (CA) \times (DB)$ & $\scalebox{0.8}{\input{Introdution_Relations_mixed_composition.tikz}}$ \\ \hline
    $\begin{bmatrix} A & B \end{bmatrix} \begin{bmatrix} C \\ D \end{bmatrix}=AC+BD$ & $\scalebox{0.8}{\input{Introdution_Relations_divideconquer.tikz}}$
  \end{tabular}
  \caption{Implicit rules. The dotted lines represent parenthesization.}
  \label{tab:freetheoremsintro}
\end{table}

Table~\ref{tab:dicconcepts1} showcases how classical constructions are formulated on either graphical or usual notation.

\begin{table}[pos=htb]
  \centering
  \begin{tabular}{c c c}
    \textbf{Graphical Syntax} & \textbf{Name} & \textbf{Classical Notation} \\ \hline
    \scalebox{0.8}{\input{Introdution_Dictionary_relation.tikz}} & Relation & $R$
    \\ \hline
    \scalebox{0.8}{\input{Introdution_Dictionary_relation_op.tikz}} &  Opposite & $R^{op}$
    \\ \hline
    \scalebox{0.8}{\input{Introdution_Dictionary_color1.tikz}} &  \makecell{Orthogonal complement\\\citep{Stein2024}} & $R^{\perp}$
    \\ \hline
    \scalebox{0.8}{\input{Introdution_Dictionary_comp-serie.tikz}} & Composition & $R \,;\, S=SR$
    \\ \hline
    \scalebox{0.8}{\input{Introdution_Dictionary_comp-paralelo.tikz}} & Cartesian product & $R \times S$
    \\ \hline
    \Id & Identity & $I$
    \\ \hline
    \scalebox{0.8}{\input{Introdution_Dictionary_lineartransformation.tikz}} & Linear transformation & $A$
    \\ \hline
    \scalebox{0.8}{\Inv{A}} & \makecell{Linear transformation\\known to have an inverse} & $A^{-1}$
    \\ \hline
    \scalebox{0.8}{\Ortho{U}} & \makecell{Orthogonal linear\\transformation} & $U$
    \\ \hline
    \input{Introdution_Dictionary_subspace.tikz} &  Linear Subspace & $V$
    \\ \hline
    \input{Introdution_Dictionary_matrixAB.tikz} & \makecell{Linear transformation\\composition} & $A;B = BA$
    \\ \hline
    \scalebox{0.8}{\input{Introdution_Dictionary_imagem.tikz}}  & Image & $\Img{(A)}$
    \\ \hline
    \scalebox{0.8}{\input{Introdution_Dictionary_kernel.tikz}} & Kernel & $\Ker{(A)}$
    \\ \hline
    \scalebox{0.8}{\input{Introdution_Dictionary_coluna.tikz}} & \makecell{Two-by-one block\\linear transformation} & $\begin{bmatrix}A\\B \end{bmatrix}$
    \\ \hline
    \scalebox{0.8}{\input{Introdution_Dictionary_linha.tikz}} & \makecell{One-by-two block\\linear transformation} & $\begin{bmatrix}A&B \end{bmatrix}$
    \\ \hline
    \scalebox{0.8}{\input{Introdution_Dictionary_defmatrix.tikz}} & \makecell{Two-by-two\\linear transformation} & $\begin{bmatrix}a&b\\c&d \end{bmatrix}$
    \\ \hline
    \scalebox{0.8}{\input{Introdution_Dictionary_defmatrixbloco.tikz}} & \makecell{Two-by-two block\\linear transformation} & $\begin{bmatrix}A&B\\C&D \end{bmatrix}$
    \\
  \end{tabular}
  \caption{Dictionary of Notations.
    $R,S$ denote relations; $A,B,C,D, U$ represent linear transformations; and $a,b,c,d$ are real numbers.
  }
  \label{tab:dicconcepts1}
\end{table}

\FloatBarrier

\section{Auxiliary results and omitted proofs}

\subsection{Relational Least-Squares}\label{ap:relationalLQ}

\begin{lemma}[Properties of the orthogonal projection]\label{lem:prop-projection}
  Let $V$ be a finite-dimensional subspace of $U$, $u\in U$ and $v\in V$. Then
  \begin{enumerate}
    \item \label{lem:prop-projection:norm}
          $\left\|P_V{(u)} \right\|_2\leq \left\|u \right\|_2$;
    \item \label{lem:prop-projection:in}
          $u-P_V{(u)}\in V^\perp$;
    \item \label{lem:prop-projection:dist}
          $\left\|u-P_V{(u)} \right\|_2\leq \left\|u-v \right\|_2$ and the inequality is an equality iff $v=P_V{(u)}$.
  \end{enumerate}
\end{lemma}
\begin{proof}
  \begin{enumerate}
    \item If $u=v+w$ with $w \in V^\perp$ then
      \begin{hcalculation}[=]{\left\|P_V{(u)} \right\|_2^2}
        \hstep{\left\| v\right\|_2^2}{$u=v+w$ and $w \in V^\perp$}
        \hstep[\leq]{\left\| v\right\|_2^2+\left\| w\right\|_2^2}{$\left\| w\right\|_2^2\geq 0$}
        \hstep[=]{\left\| u\right\|_2^2}{Pythagorean Theorem}
      \end{hcalculation}
    \item  If $u=v+w$ with $w \in V^\perp$ then $u-P_V{(u)}=u-v = w \in V^\perp$.
    \item \begin{hcalculation}[=]{\left\|u-P_V{(u)} \right\|_2^2}
        \hstep[\leq]{\left\|u-P_V{(u)} \right\|_2^2+\left\|P_V{(u)}-v \right\|_2^2}{$\left\|P_V{(u)}-v \right\|_2^2\geq 0$}
        \hstep{\left\|\left(u-P_V{(u)}\right)+\left(P_V{(u)}-v\right) \right\|_2^2}{Item~\ref{lem:prop-projection:in} + Pythagorean Theorem}
        \hstep{\left\|u-v \right\|_2^2}{Algebra}
      \end{hcalculation} Furthermore, inequality is an equality if and only if $\left\|P_V{(u)}-v \right\|_2^2=0\Leftrightarrow v=P_V{(u)}$.
  \end{enumerate}
\end{proof}

\begin{proof}[Proof of Theorem~\ref{thm:LQ-relational} (Relational Least-Squares)]
  We prove items~\ref{thm:LQ-relational:sur} and~\ref{thm:LQ-relational:inj}.
  The proofs of the other items are analogous.
\begin{enumerate}
\item
\begin{hcalculation}[=]{\Sur(R)}
  \hstep{\left\{(x,\bar{y})\mid \exists y;\; xRy\;\;\text{and}\;\; y\in \underset{\hat{y}}{\argmin}\left [\left\| \bar{y}-\hat{y}\right\|_2^2\;:\;\hat{y}\in \Img(R) \right ] \right\}}{Def.\ref{def:LQ-geometric}}
  \hstep{\left\{(x,\bar{y})\mid \exists y;\; xRy\;\;\text{and}\;\; \forall \hat{y}\in \Img{(R)};\; \left\| \bar{y}-y\right\|_2^2\leq \left\| \bar{y}-\hat{y}\right\|_2^2 \right\}}{Def. argmin}
  \hstep{\left\{(x,\bar{y})\mid \exists y;\; xRy\;\;\text{and}\;\; y=P_{\Img(R)}{\bar{y}} \right\}}{Lem.\ref{lem:prop-projection} -- item~\ref{lem:prop-projection:dist}}
  \hstep{P_{\Img(R)}^{op}\;R}{Def.\ref{def:composing_diagrams} -- items~\ref{def:composing_diagrams:seq} and~\ref{def:composing_diagrams:par}}.
\end{hcalculation}

\item By definition, $(x,y) \in R$ and $(\bar{x},y) \in R$. Since the linear relation $R$ is a subspace of $U \times V$, it follows that
  \[(x,y)-(\bar{x},y)=(x-\bar{x},0) \in R\]
  Consequently, $x-\bar{x} \in \Ker(R)$. Thus, any $\tilde{x}$ such that $y \in R\tilde{x}$ can be expressed as $\tilde{x} = x + e$, for some $e \in \Ker(R)$.
  \begin{hcalculation}[\Leftrightarrow]{\bar{x}\in \underset{\tilde{x}:\;\tilde{x}Ry}{\argmin}\left\|\tilde{x}\right\|_2^2}
    \hstep{\tilde{x}\in \underset{e\in \Ker(R)}{\argmin}\left\|x+e\right\|_2^2}{$\hat{x}=x+e$}
    \hstep{\bar{x}\in \underset{e\in \Ker(R)}{\argmin}\left\|P_{\Ker(R)}x + P_{\Ker(R)^\perp}x+e\right\|_2^2}{$x=P_{\Ker(R)}x + P_{\Ker(R)^\perp}x$}
    \hstep{\bar{x}\in \underset{\bar{e}\in \Ker(R)}{\argmin}\left\|P_{\Ker(R)^\perp}x + \bar{e}\right\|_2^2}{$\bar{e}\coloneqq + P_{\Ker(R)}x\in \Ker(R)$}
    \hstep{\bar{x}\in \underset{\bar{e}\in \Ker(R)}{\argmin}\left(\left\|P_{\Ker(R)^\perp}x\right\|_2^2+\left\|\bar{e}\right\|_2^2\right)}{Pythagorean Theorem}
    \hstep{\left\|\bar{e}\right\|_2^2=0}{$\bar{e}\in \Ker(R)$ and $\left\|\bar{e}\right\|_2^2\geq 0$}
  \end{hcalculation}
  Therefore
  \begin{equation}\label{eq:proof_relationalLQ}
    \bar{x}\;=\;P_{\Ker(R)^\perp}x + \bar{e}\;=\;P_{\Ker(R)^\perp}x + 0\;=\;P_{\Ker(R)^\perp}x
  \end{equation}
  So
  \begin{hcalculation}[=]{\Inj(R)}
    \hstep{\left\{(\bar{x},y)\mid \exists x;\; xRy\;\;\text{and}\;\; \bar{x}\in \underset{\hat{x}:\;\hat{x}Ry}{\argmin}\left\|\hat{x}\right\|_2^2 \right\}}{Def.\ref{def:LQ-geometric}}
    \hstep{\left\{(\bar{x},y)\mid \exists x;\; xRy\;\;\text{and}\;\; \bar{x}=P_{\Ker(R)^\perp}x \right\}}{Eq.\eqref{eq:proof_relationalLQ}}
    \hstep{R\;P_{\Ker(R)^{\perp}}^{op}}{Def.\ref{def:composing_diagrams} -- items~\ref{def:composing_diagrams:seq} and~\ref{def:composing_diagrams:op}}
  \end{hcalculation}
\end{enumerate}
\end{proof}

\begin{proof}[Proof of Lemma~\ref{lem:prop_fund_Rbar} (Fundamental properties)]
  We prove item~\ref{lem:prop_fund_Rbar:sur} The remaining items are proved analogously.
  \begin{hcalculation}[=]{\scalebox{0.8}{\begin{tikzpicture}
	\begin{pgfonlayer}{nodelayer}
		\node [style=Black] (0) at (-2.5, 0) {};
		\node [style=Rel] (1) at (0, 0) {$\Sur(R)$};
		\node [style=none] (2) at (2.25, 0) {};
	\end{pgfonlayer}
	\begin{pgfonlayer}{edgelayer}
		\draw (1) to (2.center);
		\draw (0) to (1);
	\end{pgfonlayer}
\end{tikzpicture}
}}
    \hstep{\scalebox{0.8}{\input{Connect_disconnect_FundPropLQ_step2.tikz}}}{Thm.\ref{thm:LQ-relational}}
    \hstep{\scalebox{0.8}{\input{Connect_disconnect_FundPropLQ_step3.tikz}}}{Def.~\ref{def:proj-rel}}
    \hstep{\scalebox{0.8}{\input{Connect_disconnect_FundPropLQ_step4.tikz}}}{Def.\ref{def:fund_subspaces}: Fundamental subspace}
    \hstep[\supseteq]{\scalebox{0.8}{\input{Connect_disconnect_FundPropLQ_step5.tikz}}}{Fig.\ref{fig:Inequality}}
    \hstep{\scalebox{0.8}{\input{Connect_disconnect_FundPropLQ_step6.tikz}}}{Fig.\ref{fig:Commutative_Comonoid}}
    \hstep{\scalebox{0.8}{\CoDiscard}}{$V+V^{\perp}= \mathbb{R}^n$}
  \end{hcalculation}
\end{proof}

\subsubsection{Compositions of Least-Squares Operators}
\label{ap:comp}

\begin{proof}[Proof of Lemma~\ref{lem:comp-idempotent} (Idempotency)]
  We prove item~\ref{lem:comp-idempotent:sur}.
  The remaining items are proved analogously.
  \begin{hcalculation}[=]{\scalebox{0.8}{\Rel{\Sur\left ( \Sur(R) \right )}}}
    \hstep{\Sur \left( \scalebox{0.8}{\Comp{R/Rel,  P_{\Img(R)}/CoMap}} \right)}{Thm.\ref{thm:LQ-relational}}
    \hstep{ \scalebox{0.8}{\Comp{R/Rel,  P_{\Img(R)}/CoMap, P_{\Img(R)}/CoMap }}}{Thm.\ref{thm:LQ-relational}}
    \hstep{ \scalebox{0.8}{\Comp{R/Rel,  P_{\Img(R)}/CoMap} }}{Lem.\ref{lem:proj-props}}
    \hstep{\scalebox{0.8}{\Rel{\Sur(R)}}}{Thm.\ref{thm:LQ-relational}}
  \end{hcalculation}
\end{proof}

\begin{proof}[Proof of Lemma~\ref{lem:comp-idempotent}]
  (Commutativity)
  We prove item~\ref{lem:comp-commutative:sur}.
  The remaining items are proved analogously.
  \begin{hcalculation}[=]{\scalebox{0.8}{\Rel{\Sur\left( \Inj(R) \right)}}}
    \hstep{ \Sur\left( \scalebox{0.8}{\Comp{P_{\Ker(R)}/CoMapGray, R/Rel}} \right) }{Thm.\ref{thm:LQ-relational}}
    \hstep{ \scalebox{0.8}{\Comp{P_{\Ker(R)}/CoMapGray, R/Rel, P_{\Img(R)}/CoMap}} }{Thm.\ref{thm:LQ-relational}}
    \hstep{ \Inj\left( \scalebox{0.8}{\Comp{R/Rel, P_{\Img(R)}/CoMap}} \right)     }{Thm.\ref{thm:LQ-relational}}
    \hstep{ \scalebox{0.8}{\Rel{\Inj\left( \Sur(R) \right)} }}{Thm.\ref{thm:LQ-relational}}
  \end{hcalculation}
\end{proof}

\subsection{Pseudoinverse}\label{ap:pseudoinverse}

\begin{proof}[Proof of Lemma~\ref{lem:pseudo-relation}]
  \begin{hcalculation}[=]{\scalebox{0.8}{\CoMap{R^+}}}
    \hstep{\scalebox{0.8}{\CoMap{\Sur\left ( \Inj(R) \right )}}}{Def.\ref{def:pseudo-rel}}
    \hstep{\scalebox{0.8}{\input{MQ_Pseudo_pseudo-relation_step1.tikz}}}{Thm.\ref{thm:LQ-relational}}
    \hstep{\scalebox{0.8}{\input{MQ_Pseudo_pseudo-relation_step2.tikz}}}{Def.~\ref{def:proj-rel}}
    \hstep{\scalebox{0.8}{\input{MQ_Pseudo_pseudoinverse.tikz}}}{Def.\ref{def:fund_subspaces}: Fundamental subspace}
  \end{hcalculation}
\end{proof}

\begin{proof}[Proof of Corollary~\ref{lem:pseudoinverse}]
  \begin{hcalculation}[=]{\scalebox{0.8}{\CoMap{A^+}}}
    \hstep{\scalebox{0.8}{\input{MQ_Pseudo_pseudoinverseA.tikz}}}{Lem.\ref{lem:pseudo-relation}}
    \hstep{\scalebox{0.8}{\input{MQ_Pseudo_pseudoinverseA2.tikz}}}{Fig.\ref{fig:Flip}}
    \hstep{\scalebox{0.8}{\input{MQ_Pseudo_pseudoinverseA3.tikz}}}{Fig.\ref{fig:Commutative_Comonoid}}
    \hstep{\scalebox{0.8}{\input{MQ_Pseudo_pseudoinverseA4.tikz}}}{Lem.\ref{lem:weak-inverse}: \scalebox{0.7}{\Comp{A/Map,A/CoMap,A/Map} = \Map{A}}}
    \hstep{\scalebox{0.8}{\input{MQ_Pseudo_Pseudo-function_pseudo-function.tikz}}}{Lem.\ref{lem:weak-inverse}: \scalebox{0.7}{\Comp{A/Map,A/CoMap,A/Map} = \Map{A}}}
  \end{hcalculation}
\end{proof}

\begin{proof}[Proof of Theorem~\ref{thm:moore-penrose-pseudo} (Moore-Penrose pseudoinverse)]
  We provide proofs for equations~\eqref{thm:moore-penrose-pseudo1} and~\eqref{thm:moore-penrose-pseudo3},
  since the other two are analogous, following equivalent reasoning.
  \begin{enumerate}
    \item[\eqref{thm:moore-penrose-pseudo1}]
      \begin{hcalculation}[=]{\scalebox{0.8}{\begin{tikzpicture}
	\begin{pgfonlayer}{nodelayer}
		\node [style=none] (0) at (-2.5, 0) {};
		\node [style=Map] (1) at (-1, 0) {$A$};
		\node [style=Map] (2) at (1, 0) {$A^+$};
		\node [style=Map] (3) at (3, 0) {$A$};
		\node [style=none] (4) at (4.5, 0) {};
	\end{pgfonlayer}
	\begin{pgfonlayer}{edgelayer}
		\draw (0.center) to (4.center);
	\end{pgfonlayer}
\end{tikzpicture}
}}
        \hstep{\scalebox{0.8}{\input{MQ_Pseudo_Pseudo-function_eq1_step1.tikz}}}{Lem.\ref{lem:pseudoinverse}}
        \hstep{\scalebox{0.8}{\input{MQ_Pseudo_Pseudo-function_eq1_step2.tikz}}}{Fig.\ref{fig:Flip}}
        \hstep{\scalebox{0.8}{\input{MQ_Pseudo_Pseudo-function_eq1_step3.tikz}}}{Fig.\ref{fig:Commutative_Comonoid}}
        \hstep{\scalebox{0.8}{\input{MQ_Pseudo_Pseudo-function_eq1_step4.tikz}}}{Fig.\ref{fig:Commutative_Comonoid} + Fig.\ref{fig:Flip}}
        \hstep{\scalebox{0.8}{\input{MQ_Pseudo_Pseudo-function_eq1_step5.tikz}}}{Def~\ref{def:composing_diagrams}: \scalebox{0.7}{\CoMapGray{A}=\Map{A^{\top}}}}
        \hstep[\subseteq]{\scalebox{0.8}{\input{MQ_Pseudo_Pseudo-function_eq1_step6.tikz}}}{Fig~\ref{fig:typerelations}- DET: $\scalebox{0.7}{\Comp{A^{\top}/CoMap,  A^{\top}/Map}} \subseteq \scalebox{0.7}{\Id}$}
        \hstep{\scalebox{0.8}{\Map{A}}}{Lem.\ref{lem:weak-inverse}: \scalebox{0.7}{\Comp{A/Map,A/CoMap,A/Map} = \Map{A}}}
      \end{hcalculation}
      So, \begin{hcalculation}[\Leftrightarrow]{\scalebox{0.8}{\begin{tikzpicture}
	\begin{pgfonlayer}{nodelayer}
		\node [style=none] (0) at (-2.5, 0) {};
		\node [style=Map] (1) at (-1, 0) {$A$};
		\node [style=Map] (2) at (1, 0) {$A^+$};
		\node [style=Map] (3) at (3, 0) {$A$};
		\node [style=none] (4) at (4.5, 0) {};
	\end{pgfonlayer}
	\begin{pgfonlayer}{edgelayer}
		\draw (0.center) to (4.center);
	\end{pgfonlayer}
\end{tikzpicture}
}\subseteq \scalebox{0.8}{\Map{A}}}
        \hstep{\scalebox{0.8}{\begin{tikzpicture}
	\begin{pgfonlayer}{nodelayer}
		\node [style=none] (0) at (-2.5, 0) {};
		\node [style=Map] (1) at (-1, 0) {$A$};
		\node [style=Map] (2) at (1, 0) {$A^+$};
		\node [style=Map] (3) at (3, 0) {$A$};
		\node [style=none] (4) at (4.5, 0) {};
	\end{pgfonlayer}
	\begin{pgfonlayer}{edgelayer}
		\draw (0.center) to (4.center);
	\end{pgfonlayer}
\end{tikzpicture}
}= \scalebox{0.8}{\Map{A}}}{Lem.\ref{lem:pair_of_matrices_1}}
      \end{hcalculation}
    \item[\eqref{thm:moore-penrose-pseudo3}]
      \begin{hcalculation}[=]{\scalebox{0.8}{\begin{tikzpicture}
	\begin{pgfonlayer}{nodelayer}
		\node [style=none] (0) at (-2.5, 0) {};
		\node [style=CoMapGray] (1) at (-1, 0) {$A$};
		\node [style=CoMapGray] (2) at (1, 0) {$A^{+}$};
		\node [style=none] (3) at (2.5, 0) {};
	\end{pgfonlayer}
	\begin{pgfonlayer}{edgelayer}
		\draw (0.center) to (3.center);
	\end{pgfonlayer}
\end{tikzpicture}
}}
        \hstep{\scalebox{0.8}{\input{MQ_Pseudo_Pseudo-function_eq3-1_step1.tikz}}}{Lem.\ref{lem:pseudo-relation}}
        \hstep{\scalebox{0.8}{\input{MQ_Pseudo_Pseudo-function_eq3-1_step2.tikz}}}{Fig.\ref{fig:Flip}}
        \hstep{\scalebox{0.8}{\input{MQ_Pseudo_Pseudo-function_eq3-1_step2-1.tikz}}}{$\Ker(AA^\top)=\Ker(A^\top)$}
        \hstep{\scalebox{0.8}{\input{MQ_Pseudo_Pseudo-function_eq3-1_step3.tikz}}}{$\Img(A)=\Img(AA^\top)$}
        \hstep{\scalebox{0.8}{\input{MQ_Pseudo_Pseudo-function_eq3-1_step4.tikz}}}{Fig.\ref{fig:Flip}}
        \hstep{\scalebox{0.8}{\begin{tikzpicture}
	\begin{pgfonlayer}{nodelayer}
		\node [style=none] (0) at (-2.5, 0) {};
		\node [style=Map] (1) at (-1, 0) {$A^{+}$};
		\node [style=Map] (2) at (1, 0) {$A$};
		\node [style=none] (3) at (2.5, 0) {};
	\end{pgfonlayer}
	\begin{pgfonlayer}{edgelayer}
		\draw (0.center) to (3.center);
	\end{pgfonlayer}
\end{tikzpicture}
}}{Lem.\ref{lem:pseudo-relation}}
      \end{hcalculation}
  \end{enumerate}
  It only remains to prove that $\scalebox{0.8}{\Map{A^+}}$ is unique. To do this, assume that there exist linear transformations $\scalebox{0.8}{\Map{X_1}}$ and $\scalebox{0.8}{\Map{X_2}}$ that satisfy the four equations.
  \begin{hcalculation}[=]{\scalebox{0.8}{\Map{X_1}}}
    \hstep{\scalebox{0.8}{\begin{tikzpicture}
	\begin{pgfonlayer}{nodelayer}
		\node [style=Map] (0) at (0, 0) {$A$};
		\node [style=Map] (1) at (-1.75, 0) {$X_1$};
		\node [style=Map] (2) at (1.75, 0) {$X_1$};
		\node [style=none] (3) at (-3.25, 0) {};
		\node [style=none] (4) at (3.25, 0) {};
	\end{pgfonlayer}
	\begin{pgfonlayer}{edgelayer}
		\draw (3.center) to (4.center);
	\end{pgfonlayer}
\end{tikzpicture}
}}{Thm.\ref{thm:moore-penrose-pseudo}: Eq.\eqref{thm:moore-penrose-pseudo2}}
    \hstep{\scalebox{0.8}{\input{MQ_Pseudo_Pseudo-function_Unique_step2.tikz}}}{Thm.\ref{thm:moore-penrose-pseudo}: Eq.\eqref{thm:moore-penrose-pseudo3}}
    \hstep{\scalebox{0.8}{\input{MQ_Pseudo_Pseudo-function_Unique_step3.tikz}}}{Thm.\ref{thm:moore-penrose-pseudo}: Eq.\eqref{thm:moore-penrose-pseudo1}}
    \hstep{\scalebox{0.8}{\input{MQ_Pseudo_Pseudo-function_Unique_step4.tikz}}}{Thm.\ref{thm:moore-penrose-pseudo}: Eq.\eqref{thm:moore-penrose-pseudo3}}
    \hstep{\scalebox{0.8}{\begin{tikzpicture}
	\begin{pgfonlayer}{nodelayer}
		\node [style=Map] (0) at (0.25, 0) {$A$};
		\node [style=Map] (1) at (-1.25, 0) {$X_2$};
		\node [style=Map] (2) at (1.75, 0) {$X_1$};
		\node [style=none] (3) at (-2.75, 0) {};
		\node [style=none] (6) at (3.25, 0) {};
	\end{pgfonlayer}
	\begin{pgfonlayer}{edgelayer}
		\draw (3.center) to (6.center);
	\end{pgfonlayer}
\end{tikzpicture}
}}{Thm.\ref{thm:moore-penrose-pseudo}: Eq.\eqref{thm:moore-penrose-pseudo2}}
    \hstep{\scalebox{0.8}{\input{MQ_Pseudo_Pseudo-function_Unique_step6.tikz}}}{Thm.\ref{thm:moore-penrose-pseudo}: Eq.\eqref{thm:moore-penrose-pseudo2}}
    \hstep{\scalebox{0.8}{\input{MQ_Pseudo_Pseudo-function_Unique_step7.tikz}}}{Thm.\ref{thm:moore-penrose-pseudo}: Eq.\eqref{thm:moore-penrose-pseudo4}}
    \hstep{\scalebox{0.8}{\input{MQ_Pseudo_Pseudo-function_Unique_step8.tikz}}}{Thm.\ref{thm:moore-penrose-pseudo}: Eq.\eqref{thm:moore-penrose-pseudo1}}
    \hstep{\scalebox{0.8}{\begin{tikzpicture}
	\begin{pgfonlayer}{nodelayer}
		\node [style=Map] (0) at (0.25, 0) {$A$};
		\node [style=Map] (1) at (-1.25, 0) {$X_2$};
		\node [style=Map] (2) at (1.75, 0) {$X_2$};
		\node [style=none] (3) at (-2.75, 0) {};
		\node [style=none] (6) at (3.25, 0) {};
	\end{pgfonlayer}
	\begin{pgfonlayer}{edgelayer}
		\draw (3.center) to (6.center);
	\end{pgfonlayer}
\end{tikzpicture}
}}{Thm.\ref{thm:moore-penrose-pseudo}: Eq.\eqref{thm:moore-penrose-pseudo4}}
    \hstep{\scalebox{0.8}{\Map{X_2}}}{Thm.\ref{thm:moore-penrose-pseudo}: Eq.\eqref{thm:moore-penrose-pseudo2}}
  \end{hcalculation}
\end{proof}

\subsubsection{GSVD Algorithm for Pseudoinverse}\label{ap:algopseudoinverse}

\begin{proof}[Proof of Theorem~\ref{thm:GSVD-Pseudoinverse} (GSVD for Pseudoinverse)]
  \begin{hcalculation}[=]{\scalebox{0.8}{\CoMap{R^+}}}
    \hstep{\scalebox{0.8}{\input{MQ_Pseudo_pseudoinverse.tikz}}}{Lem.\ref{lem:pseudo-relation}}
    \hstep{\scalebox{0.8}{\input{MQ_Pseudo_pseudo-gsvd_step_1.tikz}}}{\makecell{Thm.\ref{thm:gsvd-relations}: GSVD\\+ Lem.\ref{lem:subespaces_GSVD}}}
    \hstep{\scalebox{0.8}{\input{MQ_Pseudo_pseudo-gsvd_step_2.tikz}}}{$U$ is orthogonal + Fig.\ref{fig:Inequality}}
    \hstep{\scalebox{0.8}{\input{MQ_Pseudo_pseudo-gsvd_step_3.tikz}}}{Fig.\ref{fig:Commutative_Comonoid} + Fig.\ref{fig:Bialgebra}}
    \hstep{\scalebox{0.8}{\input{GSVD_GSVD_pseudoinversa_withD2.tikz}}}{Def.\ref{cor:GSVD_diagonal_relation}}
  \end{hcalculation}
\end{proof}

\begin{proof}[Proof of Theorem~\ref{thm:SVDpseudoinverse} (SVD for Pseudoinverse)]
  \begin{hcalculation}[=]{\scalebox{0.8}{\CoMap{A^+}}}
    \hstep{\scalebox{0.8}{\input{MQ_Pseudo_Pseudo-function_pseudo-function.tikz}}}{Lem.\ref{lem:pseudoinverse}}
    \hstep{\scalebox{0.8}{\input{MQ_Pseudo_pseudo-svd_step_1.tikz}}}{SVD}
    \hstep{\scalebox{0.8}{\input{MQ_Pseudo_pseudo-svd_step_2.tikz}}}{Lem.\ref{lem:subespaces_GSVD}}
    \hstep{\scalebox{0.8}{\input{MQ_Pseudo_pseudo-svd_step_3.tikz}}}{Fig.\ref{fig:Flip}}
    \hstep{\scalebox{0.8}{\input{MQ_Pseudo_pseudo-svd_step_4.tikz}}}{U and V are orthogonal}
    \hstep{\scalebox{0.8}{\input{MQ_Pseudo_pseudo-svd_step_5.tikz}}}{Fig.\ref{fig:Commutative_Comonoid}}
  \end{hcalculation}
\end{proof}

\subsubsection{Optimization problems}\label{ap:functionalizer}

\begin{proposition}[Triple of linear transformations for a linear relation]\label{prop:tripla}
For every linear relation $\TypedRel{R}{m}{n}$, there exist $p,q \in \mathbb{N}$ and linear transformations $\TypedMap{A}{q}{p}$, $\TypedMap{B}{q}{n}$ and $\TypedMap{C}{m}{p}$ such that:
\[\Rel{R} =\begin{tikzpicture}
	\begin{pgfonlayer}{nodelayer}
		\node [style=Map] (0) at (-1.75, 0) {$C$};
		\node [style=CoMap] (1) at (0, 0) {$A$};
		\node [style=Map] (2) at (1.75, 0) {$B$};
		\node [style=none] (3) at (3, 0) {};
		\node [style=none] (4) at (-3, 0) {};
	\end{pgfonlayer}
	\begin{pgfonlayer}{edgelayer}
		\draw (4.center) to (3.center);
	\end{pgfonlayer}
\end{tikzpicture}
.\]
\end{proposition}
\begin{proof}
\begin{hcalculation}[=]{\scalebox{0.8}{\Rel{R}}}
\hstep{\scalebox{0.8}{\CoSpan{C}{\tilde{A}}}}{Prop.\ref{prop:span-cospan}}
\hstep{\scalebox{0.8}{\begin{tikzpicture}
	\begin{pgfonlayer}{nodelayer}
		\node [style=Map] (0) at (-1.75, 0) {$C$};
		\node [style=CoMap] (1) at (0, 0) {$A$};
		\node [style=Map] (2) at (1.75, 0) {$B$};
		\node [style=none] (3) at (3, 0) {};
		\node [style=none] (4) at (-3, 0) {};
	\end{pgfonlayer}
	\begin{pgfonlayer}{edgelayer}
		\draw (4.center) to (3.center);
	\end{pgfonlayer}
\end{tikzpicture}
}}{Prop.\ref{prop:span-cospan}: $\scalebox{0.8}{\CoMap{\tilde{A}}}$is also a linear relation}
\end{hcalculation}
\end{proof}

For the theorem below, consider the following notations presented in~\cite{Golub:1991}:
\begin{align*}
  r_a=\rank(A),\; r_b &=\rank(B),\; r_c=\rank(C),\\
  r_{ac}=\rank\begin{pmatrix}A & C\end{pmatrix},\; r_{ab} &=\rank\begin{pmatrix}A\\B\end{pmatrix},
  \; r_{abc}=\rank\begin{pmatrix}A & C\\ B & 0\end{pmatrix}.
\end{align*}

\begin{theorem}[Theorem~$1$ of \cite{Golub:1991}: The RSVD Theorem]\label{thm:RSVD}
Every triplet of linear transformations $\TypedMap{A}{n}{m}$, $\TypedMap{B}{n}{q}$ and
$\TypedMap{C}{p}{m}$ can be factorized as
\begin{align*}
\TypedMap{A}{n}{m} &= \input{Golub_rsvd1.tikz},\\
\TypedMap{B}{n}{q} &= \input{Golub_rsvd2.tikz}, \\
\TypedMap{C}{p}{m} &= \input{Golub_rsvd3.tikz},
\end{align*}
where $\TypedInv{P}{m}{m}$ and $\TypedInv{Q}{n}{n}$ are invertible linear transformations, and
$\TypedOrtho{\bar{U}}{q}{q}$ and $\TypedOrtho{\bar{V}}{p}{p}$ are orthogonal linear transformations.
$\TypedMap{S_a}{n}{m}$, $\TypedMap{S_b}{n}{q}$, and $\TypedMap{S_c}{p}{m}$
are real quasi-diagonal linear transformations with nonnegative elements.
%
\end{theorem}

\begin{corollary}[RSVD of a linear relation]\label{cor:rsdv-relation}
For every linear relation $\TypedRel{R}{m}{n}$ there exist orthogonal linear transformations $\TypedOrtho{\bar{U}}{n}{n}$ and $\TypedOrtho{\bar{V}}{m}{m}$, and a diagonal linear relation $\TypedRel{\bar{D}}{m}{n}$ such that:
\[\Rel{R}=\input{Golub_rsvd_withD.tikz}.\]
\end{corollary}
\begin{proof}
\begin{hcalculation}[=]{\scalebox{0.8}{\Rel{R}}}
\hstep{\scalebox{0.8}{\begin{tikzpicture}
	\begin{pgfonlayer}{nodelayer}
		\node [style=Map] (0) at (-1.75, 0) {$C$};
		\node [style=CoMap] (1) at (0, 0) {$A$};
		\node [style=Map] (2) at (1.75, 0) {$B$};
		\node [style=none] (3) at (3, 0) {};
		\node [style=none] (4) at (-3, 0) {};
	\end{pgfonlayer}
	\begin{pgfonlayer}{edgelayer}
		\draw (4.center) to (3.center);
	\end{pgfonlayer}
\end{tikzpicture}
}}{Prop.\ref{prop:tripla}}
\hstep{\scalebox{0.8}{\input{Golub_rvsd_step1.tikz}}}{Thm.\ref{thm:RSVD}}
\hstep{\scalebox{0.8}{\input{Golub_rvsd_step2.tikz}}}{\scalebox{0.7}{\Inv{P}} and \scalebox{0.7}{\Inv{Q}} are invertible}
\hstep{\scalebox{0.8}{\input{Golub_rsvd_withD.tikz}}}{$\scalebox{0.7}{\Rel{\bar{D}}} \coloneqq \scalebox{0.7}{\begin{tikzpicture}
	\begin{pgfonlayer}{nodelayer}
		\node [style=Map] (1) at (-1, 0) {$S_c$};
		\node [style=none] (2) at (-2.75, 0) {};
		\node [style=CoMap] (10) at (0.75, 0) {$S_a$};
		\node [style=Map] (15) at (2.5, 0) {$S_b$};
		\node [style=none] (18) at (4, 0) {};
	\end{pgfonlayer}
	\begin{pgfonlayer}{edgelayer}
		\draw (2.center) to (18.center);
	\end{pgfonlayer}
\end{tikzpicture}
}$}
\end{hcalculation}
\end{proof}

\begin{remark}\label{rmk:rsvd-to-gsvd} According to Theorem~$2$ of \cite{Golub:1991}, for every linear transformations $\TypedMap{A}{n}{m}$, $\TypedMap{B}{n}{q}$ and $\TypedMap{C}{p}{m}$, we have
\begin{enumerate}
  \item If $\TypedMap{C}{m}{p}=\TypedId{m}$ then $\operatorname{RSVD}\left(A,B,I_m\right)=\operatorname{GSVD}\left(A,B\right)$;
  \item If $\TypedMap{B}{q}{n}=\TypedId{n}$ then $\operatorname{RSVD}\left(A,I_n,C\right)=\operatorname{GSVD}\left(A,C\right)$.
\end{enumerate}
In these cases, the orthogonal linear transformations $\Ortho{\bar{U}}$ and $\Ortho{\bar{V}}$ coincide with the orthogonal linear transformations $\Ortho{U}$ and $\Ortho{V}$ from the GSVD (Theorem~\ref{thm:gsvd-relations}), while the diagonal linear relation satisfies $\Rel{\bar{D}}=\Rel{D}$, where $\Rel{D}$ is the diagonal linear relation defined in Definition~\ref{def:diagonalrelation}.
\end{remark}

\citet{Golub:1991} use rank specifications as conditions of uniqueness and existence in some of their theorems. Here we interpret these conditions in terms of the fundamental properties (injectivity, surjectivity, determinism, and totality), since we can associate them with the threads of the GSVD (Theorem~\ref{thm:gsvd-relations}). For comparison purposes, we write the wire dimensions $k_I$, $k_S$, $k_T$, $k_D$ and $d$ in terms of the ranks of $(A,B,C)$:
\begin{align}
  k_I &= r_{ab} - r_{abc} + m,\label{eq:k_i} \\
  k_S &= r_{ac} - r_{abc} + n,\label{eq:k_s} \\
  k_T &= r_{ac} - r_a,\label{eq:k_T} \\
  k_D &= r_{ab} - r_a,\label{eq:k_r}\\
  d   &= r_{abc} + r_{a} - r_{ac} - r_{ab}.\label{eq:r}
\end{align}

\begin{proposition}[Rank conditions for injectivity and surjectivity]\label{prop:wiresthm4} A linear relation
  $\Rel{R} = \Span{A}{B}$
  is \emph{injective} and \emph{surjective} if and only if
  \[r_b = n,\quad\quad r_c = m,\quad\quad r_{abc} = r_{ac} + r_b = r_{ab} + r_c.\]
\end{proposition}
\begin{proof}
Since $C=I_m$ then $r_c=r_{ac}=m$.

\noindent $(\Rightarrow)$ Let $R$ be an injective and surjective linear relation.
\begin{hcalculation}[\Leftrightarrow ]{R\;\text{is injective}}
\hstep{k_I=0}{Lem.\ref{lem:no_lolipop_GSVD}}
\hstep{r_{ab} - r_{abc} + m = 0}{Eq.(\ref{eq:k_i})}
\hstep{r_{ab} - r_{abc} + r_c = 0}{Hypothesis: $C=I_m$}
\hstep{r_{abc} = r_{ab} + r_c}{Algebra}
\end{hcalculation}
On the other hand,
\begin{hcalculation}[\Leftrightarrow ]{R\;\text{is surjective}}
\hstep{k_S=0}{Lem.\ref{lem:no_lolipop_GSVD}}
\hstep{r_{ac} - r_{abc} + n = 0}{Eq.(\ref{eq:k_s})}
\hstep{r_{abc} = r_{ac} + n}{Algebra}
\hstep{r_{abc} = r_{ab} + n\leq r_{ac} + r_b}{Rank property: $r_{abc} \leq r_{ac} + r_b$}
\hstep{n\leq r_b}{Algebra}
\hstep{n = r_b}{$r_b\leq n$ because $B \in \mathbb{R}^{n\times q}$}
\end{hcalculation}
Therefore, $r_{abc} = r_{ac} + n = r_{ac} + r_b$.

\noindent $( \Leftarrow )$ Direct application of the hypothesis to Equations~\eqref{eq:k_i} and~\eqref{eq:k_s}.
\end{proof}

\begin{theorem}[Theorem~$4$ of \cite{Golub:1991}]\label{thm:teo4golub} Let $\TypedMap{A}{n}{m}$, $\TypedMap{B}{n}{q}$, $\TypedMap{C}{p}{m}$ and $\TypedMap{S}{n}{m}$ be linear transformations such that $\rank\left(\scalebox{0.8}{\input{Golub_a-bdc.tikz}}\right) = r_{ab} + r_{ac} - r_{abc}$ and assume
that $r_a > r_{ab} + r_{ac} - r_{abc}$. Then
\[\left\{ \input{GSVD_GSVD_pseudoinversa_withD_op_golub.tikz}\right\}  \in \underset{\scalebox{0.7}{\Map{S}}}{\argmin}\left[\rank\left(\input{Golub_a-bdc.tikz}\right)\right]\]
\noindent if and only if:
\[r_b = n,\quad\quad r_c = m,\quad\quad r_{abc} = r_{ac} + r_b = r_{ab} + r_c.\]
\end{theorem}

\begin{corollary}\label{cor:golub4-nosso} Let $\Rel{R} =\Span{A}{B}$ be a linear relation. Then
\[\left\{\input{GSVD_GSVD_pseudoinversa_withD2.tikz}\right\} = \underset{\scalebox{0.7}{\CoMap{S}}}{\argmin}\left[\rank\left(\input{Golub_a-dc.tikz}\right)\right]\]
\noindent if and only if $\Rel{R}$ is injective and surjective.
\end{corollary}
\begin{proof}
\begin{hcalculation}[\Leftrightarrow ]{\scalebox{0.8}{\Rel{R}}\; \text{is injective and surjective}}
\hstep{r_b = n,\;\; r_c = m\;\;\text{and}\;\;r_{abc} = r_{ac} + r_b = r_{ab} + r_c}{Prop.\ref{prop:wiresthm4}}
\hstep{\left\{ \scalebox{0.8}{\input{GSVD_GSVD_pseudoinversa_withD_op_golub.tikz}}\right\}  = \underset{\scalebox{0.7}{\Map{S}}}{\argmin}\left[\rank\left(\scalebox{0.8}{\input{Golub_a-dc.tikz}}\right)\right]}{Thm.\ref{thm:teo4golub} with \scalebox{0.7}{\Map{C}}=\scalebox{0.7}{\Id}}
\hstep{\left\{\scalebox{0.8}{\input{GSVD_GSVD_pseudoinversa_withD2.tikz}}\right\} = \underset{\scalebox{0.7}{\CoMap{S}}}{\argmin}\left[\rank\left(\scalebox{0.8}{\input{Golub_a-dc.tikz}}\right)\right]}{\makecell{Opposite: Def.\ref{def:composing_diagrams}\\+ Remark~\ref{rmk:rsvd-to-gsvd}}}
\end{hcalculation}
\end{proof}

\begin{theorem}[Theorem~$7$ of \cite{Golub:1991}]\label{thm:teo7}
The matrix equation $\begin{tikzpicture}
	\begin{pgfonlayer}{nodelayer}
		\node [style=Map] (0) at (-0.5, 0) {$B$};
		\node [style=Map] (1) at (1, 0) {$S$};
		\node [style=none] (2) at (-2, 0) {};
		\node [style=Map] (3) at (2.5, 0) {$C$};
		\node [style=none] (4) at (4, 0) {};
	\end{pgfonlayer}
	\begin{pgfonlayer}{edgelayer}
		\draw (2.center) to (3);
		\draw (3) to (4.center);
	\end{pgfonlayer}
\end{tikzpicture}
=\Map{A}$ in the unknown linear transformation $\Map{S}$ is consistent if and only if $r_{abc}=r_{ab}+r_c$.
\noindent The minimum norm solution corresponds to
\[\input{GSVD_GSVD_pseudoinversa_withD_op_golub.tikz}  \in \underset{\scalebox{0.7}{\Map{S}}}{\argmin}\left[\left\|\Map{S} \right\|\;:\;\Map{A}=\begin{tikzpicture}
	\begin{pgfonlayer}{nodelayer}
		\node [style=Map] (0) at (-0.5, 0) {$B$};
		\node [style=Map] (1) at (1, 0) {$S$};
		\node [style=none] (2) at (-2, 0) {};
		\node [style=Map] (3) at (2.5, 0) {$C$};
		\node [style=none] (4) at (4, 0) {};
	\end{pgfonlayer}
	\begin{pgfonlayer}{edgelayer}
		\draw (2.center) to (3);
		\draw (3) to (4.center);
	\end{pgfonlayer}
\end{tikzpicture}
 \right],\]
\noindent where $\left\|. \right\|$ be any unitarily invariant norm.
\end{theorem}

\begin{corollary}\label{cor:golub7-nosso} Let $\Rel{R} =\Span{A}{B}$ be a linear relation. The matrix equation $\begin{tikzpicture}
	\begin{pgfonlayer}{nodelayer}
		\node [style=Map] (0) at (-1, 0) {$B$};
		\node [style=Map] (1) at (1, 0) {$S$};
		\node [style=none] (2) at (-2.5, 0) {};
		\node [style=none] (3) at (2.5, 0) {};
	\end{pgfonlayer}
	\begin{pgfonlayer}{edgelayer}
		\draw (2.center) to (3.center);
	\end{pgfonlayer}
\end{tikzpicture}
=\Map{A}$ in the unknown linear transformation $\Map{S}$ is consistent if and only if $\Rel{R}$ is injective. The minimum norm solution corresponds to
\[\input{GSVD_GSVD_pseudoinversa_withD2.tikz}  \in \underset{\scalebox{0.7}{\CoMap{S}}}{\argmin}\left[\left\|\Map{S} \right\|_F\;:\;\Map{A}= \begin{tikzpicture}
	\begin{pgfonlayer}{nodelayer}
		\node [style=Map] (0) at (-1, 0) {$B$};
		\node [style=Map] (1) at (1, 0) {$S$};
		\node [style=none] (2) at (-2.5, 0) {};
		\node [style=none] (3) at (2.5, 0) {};
	\end{pgfonlayer}
	\begin{pgfonlayer}{edgelayer}
		\draw (2.center) to (3.center);
	\end{pgfonlayer}
\end{tikzpicture}
 \right].\]
\end{corollary}
\begin{proof}
\begin{hcalculation}[\Leftrightarrow ]{\Rel{R}\; \text{is injective}}
\hstep{r_c = m\;\;\text{and}\;\;r_{abc} = r_{ab} + r_c}{Prop.\ref{prop:wiresthm4}}
\hstep{\scalebox{0.8}{\input{GSVD_GSVD_pseudoinversa_withD_op_golub.tikz}}  \in \underset{\scalebox{0.7}{\Map{S}}}{\argmin}\left[\left\|\scalebox{0.8}{\Map{S}} \right\|_F\;:\;\scalebox{0.8}{\Map{A}}=\scalebox{0.8}{\begin{tikzpicture}
	\begin{pgfonlayer}{nodelayer}
		\node [style=Map] (0) at (-1, 0) {$B$};
		\node [style=Map] (1) at (1, 0) {$S$};
		\node [style=none] (2) at (-2.5, 0) {};
		\node [style=none] (3) at (2.5, 0) {};
	\end{pgfonlayer}
	\begin{pgfonlayer}{edgelayer}
		\draw (2.center) to (3.center);
	\end{pgfonlayer}
\end{tikzpicture}
} \right]}{Thm.\ref{thm:teo7} with \scalebox{0.7}{\Map{C}}=\scalebox{0.7}{\Id}}
\hstep{\scalebox{0.8}{\input{GSVD_GSVD_pseudoinversa_withD2.tikz}}  \in \underset{\scalebox{0.7}{\CoMap{S}}}{\argmin}\left[\left\|\scalebox{0.8}{\Map{S}} \right\|_F\;:\;\scalebox{0.8}{\Map{A}}= \scalebox{0.8}{\begin{tikzpicture}
	\begin{pgfonlayer}{nodelayer}
		\node [style=Map] (0) at (-1, 0) {$B$};
		\node [style=Map] (1) at (1, 0) {$S$};
		\node [style=none] (2) at (-2.5, 0) {};
		\node [style=none] (3) at (2.5, 0) {};
	\end{pgfonlayer}
	\begin{pgfonlayer}{edgelayer}
		\draw (2.center) to (3.center);
	\end{pgfonlayer}
\end{tikzpicture}
} \right]}{\makecell{Opposite: Def.\ref{def:composing_diagrams}\\+ Remark~\ref{rmk:rsvd-to-gsvd}}}
\end{hcalculation}
\end{proof}

\subsection{Generalizing Low-Rank Approximation to Linear Relations}\label{ap:proof_cap_low-rank}

Table~\ref{tab:gsvd_withLowRank} contains all the characterizations involving the diagonal relation presented in this work.

\renewcommand{\arraystretch}{1.05}
\setlength{\tabcolsep}{3pt}

\begin{table}[pos=htbp]
\centering
\scriptsize
\setlength{\tabcolsep}{2.5pt}

\resizebox{\textwidth}{!}{%
\begin{tabular}{C{0.07\textwidth}|C{0.23\textwidth}|M{0.35\textwidth}}
\textbf{Name} &
\textbf{Graphical Syntax} &
\textbf{Matrix notation} \\ \hline

$D$ &
$\scalebox{0.8}{\input{GSVD_GSVD_D.tikz}}$ &
$xDy\iff\exists w;\; \tilde{C}w=x$
$\text{and}\; \tilde{S}w=y$ where
$ \tilde{C}=\begin{bmatrix}
0_{k_T \times k_D} & 0 & 0 \\
0 & C & 0 \\
0 & 0 & I_{k_I}
\end{bmatrix}$
$ \tilde{S}=\begin{bmatrix}
I_{k_D} & 0 & 0 \\
0 & S & 0 \\
0 & 0 & 0_{k_S \times k_I}
\end{bmatrix}$
\\ \hline

$\Sur(D)$ &
$\scalebox{0.8}{\input{Connect_disconnect_GSVD_sur_D_cospan.tikz}}$ &
$xDy\iff\tilde{C}x=\tilde{S}y$
where
$ \tilde{C}=\begin{bmatrix}
I_{k_T} & 0 & 0 \\
0 & C^{-1} & 0_{r \times k_I}
\end{bmatrix}$
$ \tilde{S}=\begin{bmatrix}
0_{k_T \times k_D} & 0 & 0 \\
0 & S^{-1} & 0_{r \times k_S}
\end{bmatrix}$
\\ \hline

$\Inj(D)$ &
$\scalebox{0.8}{\input{Connect_disconnect_GSVD_inj_D_span.tikz}}$ &
$xDy\iff\exists w;\; \tilde{C}w=x$
$\text{and}\; \tilde{S}w=y$ where
$ \tilde{C}=\begin{bmatrix}
0_{k_T \times k_D} & 0 \\
0 & C \\
0 & 0_{k_I \times r}
\end{bmatrix}$
$ \tilde{S}=\begin{bmatrix}
I_{k_D} & 0 \\
0 & S \\
0 & 0_{k_S \times r}
\end{bmatrix}$
\\ \hline

$\Tot(D)$ &
$\scalebox{0.8}{\input{Connect_disconnect_GSVD_tot_D_cospan.tikz}}$ &
$xDy\iff \tilde{C}x=\tilde{S}y$
where
$ \tilde{C}=\begin{bmatrix}
0_{r\times k_T} & 0 & 0 \\
0 & C^{-1} & 0_{k_S \times k_I}
\end{bmatrix}$
$ \tilde{S}=\begin{bmatrix}
0_{r \times k_D} & 0 & 0 \\
0 & S^{-1} & I_{k_S}
\end{bmatrix}$
\\ \hline

$\Det(D)$ &
$\scalebox{0.8}{\input{Connect_disconnect_GSVD_det_D_span.tikz}}$ &
$xDy\iff\exists w;\; \tilde{C}w=x$
$\text{and}\; \tilde{S}w=y$ where
$ \tilde{C}=\begin{bmatrix}
0_{k_T \times r} & 0 \\
C & 0 \\
0 & I_{k_I}
\end{bmatrix}$
$ \tilde{S}=\begin{bmatrix}
0_{k_D\times r} & 0 \\
S & 0 \\
0 & 0_{k_S \times k_I}
\end{bmatrix}$
\\ \hline

$D^+$ &
$\scalebox{0.8}{\input{GSVD_GSVD_pseudoinversa.tikz}}$ &
$x=D^+y$ where
$\left(D^+\right)^{op} =
\begin{bmatrix}
0_{k_T \times k_D} & 0 & 0 \\
0 & CS^{-1} & 0 \\
0 & 0 & 0_{k_I \times k_S}
\end{bmatrix}$
\\ \hline

$D_{\overline{\ell}}^{+}$ &
$\scalebox{0.8}{\input{GSVD_GSVD_lowrank.tikz}}$ &
$x=D_{\overline{\ell}}^{+}y$ where
$\left(D_{\overline{\ell}}^{+}\right)^{op} =
\begin{bmatrix}
0_{\tilde{k}_T \times \tilde{k}_D} & 0 & 0 \\
0 & C_{\overline{\ell}}S_{\overline{\ell}}^{-1} & 0 \\
0 & 0 & 0_{k_I \times k_S}
\end{bmatrix}$

\end{tabular}}

\caption{Equivalent notations for the diagonal linear relation with its relational operations, pseudoinverse and low-rank approximation.}
\label{tab:gsvd_withLowRank}
\end{table}

\begin{theorem}[Theorem~$6$ of \cite{Golub:1991}]\label{thm:teo6}
Consider all linear transformations $\Map{S}$ satisfying \[r_{ab}+r_{ac}-r_{abc}\leq k =\rank\left(\input{Golub_a-bdc.tikz}\right)< r_a\]
where $k$ is a given integer and let $\|\cdot\|$ be any unitarily invariant norm.
A linear transformation $\Map{S}$ of minimal norm $\left\|\Map{S} \right\|$ is given by
\[\input{GSVD_GSVD_pseudoinversa_withD_r-k_op.tikz}  \in \underset{\scalebox{0.8}{\Map{S}}}{\argmin}\left[\left\|\Map{S} \right\|\;:\;\rank\left(\input{Golub_a-bdc.tikz}\right)=k\right],
\]
\noindent where $\left\|. \right\|$ be any unitarily invariant norm.
\end{theorem}

\begin{corollary}\label{cor:golub6-nosso} Let $\Rel{R} =\Span{A}{B}$ be a linear relation with rank $r$. Consider all linear transformations $\Map{S}$ satisfying \[k_I\leq k =\rank\left(\input{Golub_r-d.tikz}\right)< k_I +r\]
where $k$ is a given integer and let $\|\cdot\|$ be any unitarily invariant norm.
A linear transformation $\Map{S}$ of minimal norm $\left\|\Map{S} \right\|_F$ is given by
\[\input{GSVD_GSVD_pseudoinversa_withD_r-k.tikz}  \in \underset{\scalebox{0.7}{\CoMap{S}}}{\argmin}\left[\left\|\Map{S} \right\|_F\;:\;\rank\left(\input{Golub_a-dc.tikz}\right)=k\right].\]
\end{corollary}
\begin{proof}
\begin{hcalculation}[\Leftrightarrow ]{k_I\leq k =\rank\left(\scalebox{0.7}{\input{Golub_r-d.tikz}}\right)< k_I +r}\hstep{k_I\leq k =\rank\left(\scalebox{0.7}{\input{Golub_a-dc.tikz}}\right)< k_I + r}{Lem.\ref{lem:rank_R-S}}
\hstep{r_{ab}+r_{ac}-r_{abc}\leq k =\rank\left(\scalebox{0.7}{\input{Golub_a-dc.tikz}}\right)< r_a}{\makecell{Eqs.\eqref{eq:k_i}--\eqref{eq:r}\\with $r_c=r_{ac}=m$}}
\hstep{\scalebox{0.7}{\input{GSVD_GSVD_pseudoinversa_withD_r-k_op.tikz}}  \in \underset{\scalebox{0.7}{\Map{S}}}{\argmin}\left[\left\|\scalebox{0.7}{\Map{S}} \right\|_F\;:\;\rank\left(\scalebox{0.7}{\input{Golub_a-dc.tikz}}\right)=k\right]}{\makecell{Thm.\ref{thm:teo6} with\\ \scalebox{0.7}{\Map{C}}=\scalebox{0.7}{\Id}}}
\hstep{\scalebox{0.7}{\input{GSVD_GSVD_pseudoinversa_withD_r-k.tikz}}  \in \underset{\scalebox{0.7}{\CoMap{S}}}{\argmin}\left[\left\|\scalebox{0.7}{\Map{S}} \right\|_F\;:\;\rank\left(\scalebox{0.7}{\input{Golub_a-dc.tikz}}\right)=k\right]}{\makecell{Opposite: Def.\ref{def:composing_diagrams}\\+ Remark~\ref{rmk:rsvd-to-gsvd}}}
\end{hcalculation}
\end{proof}

\section*{Acknowledgements}

This study was financed in part by the Coordenação de Aperfeiçoamento de Pessoal de Nível Superior - Brasil (CAPES) - Finance Code 001.

\bibliographystyle{cas-model2-names}
\bibliography{cas-refs}

\end{document}